\documentclass[11pt]{article}
\usepackage[margin=1in]{geometry}
\usepackage[T1]{fontenc}
\usepackage{lmodern,microtype}
\usepackage{amsmath,amssymb,amsthm,mathtools,bm}
\usepackage{graphicx,booktabs,array,tikz,pgfplots}
\usetikzlibrary{arrows.meta,positioning,calc,fit,backgrounds,decorations.pathreplacing}
\pgfplotsset{compat=1.18}
\usepackage{caption,placeins,etoolbox}
\AtBeginEnvironment{thebibliography}{\raggedright}
\usepackage[colorlinks=true,linkcolor=blue!45!black,citecolor=blue!45!black,urlcolor=blue!45!black,bookmarksnumbered=true]{hyperref}
\usepackage[nameinlink,noabbrev]{cleveref}
\numberwithin{equation}{section}
\newtheorem{theorem}{Theorem}[section]
\newtheorem{lemma}[theorem]{Lemma}
\newtheorem{proposition}[theorem]{Proposition}
\newtheorem{corollary}[theorem]{Corollary}
\theoremstyle{definition}
\theoremstyle{remark}\newtheorem{remark}[theorem]{Remark}
\DeclareMathOperator{\Tr}{Tr}
\DeclareMathOperator{\id}{id}
\DeclareMathOperator{\supp}{supp}
\DeclareMathOperator{\diag}{diag}

\newcommand{\N}{\mathcal N}\newcommand{\M}{\mathcal M}
\newcommand{\EF}{E_{\mathrm F}}\newcommand{\UF}{U_{\mathrm F}}\newcommand{\VF}{V_{\mathrm F}^{\rightarrow}}\newcommand{\Rone}{R_{\rightarrow}}\newcommand{\Rtwo}{R_{\leftrightarrow}}
\newcommand{\Q}{\mathcal Q}\newcommand{\Msg}{W}\newcommand{\Est}{\widehat{W}}\newcommand{\eps}{\varepsilon}
\newcommand{\ket}[1]{|#1\rangle}\newcommand{\bra}[1]{\langle#1|}\newcommand{\proj}[1]{|#1\rangle\!\langle#1|}
\newcommand{\Cfb}{C_{\leftarrow}}\newcommand{\ent}{\mathsf H}
\newcommand{\cD}{\mathcal D}\newcommand{\cA}{\mathcal A}\newcommand{\cE}{\mathcal E}\newcommand{\cF}{\mathcal F}
\tikzset{op/.style={draw,rounded corners=2pt,fill=white,minimum height=.68cm,inner sep=5pt,align=center},
 qwire/.style={-{Latex[length=2mm]},thick},cwire/.style={-{Latex[length=2mm]},thick,dashed}}
\allowdisplaybreaks[2]
\title{Entanglement-of-formation potentials for feedback-assisted classical communication and randomness distillation}
\author{Mark M. Wilde\thanks{Email: \href{mailto:wilde@cornell.edu}{wilde@cornell.edu}}\\
\textit{School of Electrical and Computer Engineering, Cornell University}\\
\textit{Ithaca, New York 14850, USA}}
\date{October 2, 2026}
\hypersetup{pdftitle={Entanglement-of-formation potentials for feedback-assisted classical communication and randomness distillation},pdfauthor={Mark M. Wilde}}
\begin{document}
\maketitle
\begin{abstract}
We develop entanglement-of-formation potentials for feedback-assisted classical communication and one-way randomness distillation. For a finite-dimensional quantum channel without initial shared entanglement, an exact amortization theorem bounds the increase of message information plus entanglement remaining across the communicating laboratories. This proves a classical-feedback capacity bound using a mixed-convex-roof expression previously studied by Winter and Yang in the context of potential capacities. We determine the capacity of flagged mixtures of noiseless and entanglement-breaking channels, improve a published depolarizing-channel converse, and evaluate the new bound exactly for the amplitude damping channel. A state analogue gives a finite-blocklength upper bound on the net shared randomness obtainable with one-way classical communication. Amplitude-damping Choi states provide exact benchmarks, while a full-rank generalized amplitude-damping Choi state demonstrates a useful improvement where we do not evaluate the operational rate exactly. We also distinguish one-way from two-way randomness distillation: a two-round weak-measurement protocol strictly exceeds the exact one-way rate of every nontrivial qubit isotropic state. The protocol preserves quantum side information during its first decoding and extracts additional randomness in the reverse direction. This shows why the one-way state potential is not a general two-way converse. Together, these results establish a common entanglement-based approach to sharper communication and randomness bounds, while revealing how retained quantum correlations enable two-way distillation to surpass one-way limits.
\end{abstract}

\setcounter{tocdepth}{2}
\tableofcontents

\section{Introduction}\label{sec:intro}
\subsection{Background and motivation}
Classical communication plays different roles in two basic quantum Shannon-theoretic tasks. In channel coding, a sender Alice holds a message before a protocol begins and uses a noisy quantum channel to communicate it to a receiver Bob. Noiseless classical feedback allows Bob to report measurement outcomes before later inputs are prepared. In randomness distillation, the parties instead begin with copies of a correlated bipartite state and seek nearly uniform shared random bits. Classical communication helps them agree on these bits, but its length must be subtracted from their gross yield; otherwise communication alone would trivialize the task~\cite{DW04,LRWW23,NW24}. Both tasks require understanding how local measurements convert quantum correlations into useful classical information.
Feedback-assisted channel coding introduces an additional complication absent from ordinary ensemble coding: an early transmission may establish entanglement between the laboratories, and a later input may be correlated with Bob's retained memory~\cite{Bowen04,BN05,Wilde17}. For example, Alice can first distribute one share of a Bell state and then use dense coding~\cite{BW92}. The two noiseless qubit transmissions reveal zero and two message bits, respectively. A converse that bounds each transmission's message-information increase by one bit would therefore fail, even though the total communication respects the two-qubit dimension bound. Figure~\ref{fig:protocol-overview} illustrates the adaptive structure without the bookkeeping needed for its analysis.

The static randomness distillation problem has a related distinction. An instrument can extract a classical variable while leaving a quantum system that remains correlated with the other party. With communication in only one direction, the optimal asymptotic rate is characterized by a regularized measurement-induced Holevo information~\cite{DW04}. This characterization is exact but is not generally a tractable numerical formula. A computable or certifiably evaluable one-way upper bound is therefore valuable even when the underlying channel-coding problem has no known closed solution. With two-way communication, the parties can alternate extraction and reconciliation; the final shared variable need not be fixed at the beginning.
These observations motivate an accounting method that distinguishes classical information already extracted from quantum correlations available for later use. Entanglement of formation~\cite{BDSW96} is a natural candidate for the latter term because it is expressed as a convex roof and cannot increase on average under local instruments. The challenge is to prove an amortized channel estimate for arbitrary mixed memory states, and to determine which part of the same construction remains valid for state distillation. The different resource conventions make a direct identification of the two operational tasks inappropriate.
\begin{figure}
 \centering\begin{tikzpicture}[x=1cm,y=1cm,font=\small]
\draw[-{Latex},thin] (1.0,4.35)--node[above,font=\footnotesize]{time}(12.1,4.35);
\node[anchor=east] at (.0,3.1) {Alice};
\node[anchor=east] at (.0,0) {Bob};
\node[op,minimum width=2.35cm] (a1) at (1.7,3.1) {Prepare input};
\node[op,minimum width=2.6cm] (a2) at (7.5,3.1) {Choose next input};
\node[op,minimum width=2.4cm] (b1) at (4.4,0) {Process output};
\node[op,minimum width=2.7cm] (b2) at (10.3,0) {Process / decode};
\node[op,minimum width=.85cm] (n1) at (3.05,1.55) {$\N$};
\node[op,minimum width=.85cm] (n2) at (8.9,1.55) {$\N$};
\draw[qwire] (.15,3.1)--(a1.west);
\draw[qwire] (a1.east)--node[above=5pt,font=\footnotesize]{retained memory}(a2.west);
\draw[qwire] (a2.east)--(12.15,3.1);
\draw[qwire] (.15,0)--(b1.west);
\draw[qwire] (b1.east)--node[below=5pt,font=\footnotesize]{retained memory}(b2.west);
\draw[cwire] (b2.east)--node[above]{$\widehat m$}(12.45,0);
\draw[qwire] (a1.south)--(n1.north);
\draw[qwire] (n1.south)--(b1.north);
\draw[qwire] (a2.south)--(n2.north);
\draw[qwire] (n2.south)--(b2.north);
\draw[cwire] (b1.north)--node[sloped,above=4pt,font=\footnotesize]{classical feedback}(a2.south west);
\node[font=\footnotesize] at (1.7,3.85) {message $m$};
\node[font=\footnotesize] at (7.5,3.85) {$m$ and previous outcomes};
\node[font=\footnotesize,align=center] at (6.2,-.9) {Forward systems move down and right; feedback moves up and right.};
\end{tikzpicture}
 \caption{Classical communication with receiver-to-sender feedback. Alice's later input depends on Bob's earlier classical outcome. Each party may retain quantum memory, and earlier forward transmissions may entangle those memories. Time flows from left to right. Forward quantum systems move diagonally down and right, while the classical feedback moves up and right. Solid arrows indicate quantum systems and dashed arrows classical communication. No initial shared entanglement or free forward classical communication is supplied. Figure~\ref{fig:protocol} later adds the state and potential labels used in the converse.}
 \label{fig:protocol-overview}
\end{figure}
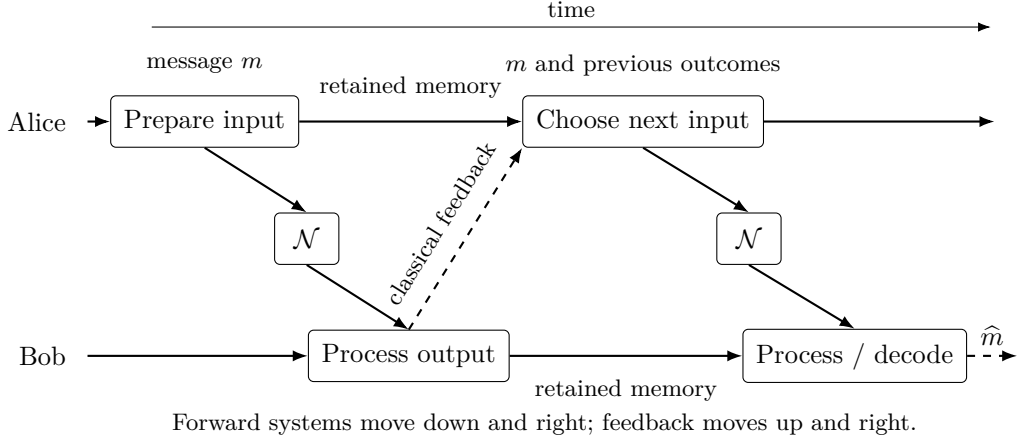

\subsection{Summary of results}
\subsubsection{Classical-feedback-assisted classical communication}

For a channel $\N\colon A\to B$ and an input state $\rho_A$, let $f_\N(\rho)$ denote the entanglement of formation between a purifying reference and the channel output. We use the following channel quantity:
\begin{equation}
 \UF(\N)\coloneqq\sup_{(p_x,\rho^x)_x}
 \left\{S(\N(\bar\rho))-\sum_xp_xS(\N(\rho^x))
 +\sum_xp_x f_\N(\rho^x)\right\},\qquad
 \bar\rho\coloneqq\sum_xp_x\rho^x,
 \label{eq:intro-UF}
\end{equation}
where $S$ denotes the von Neumann entropy.
The ensemble states need not be pure. The first two terms form an output Holevo information, while the last term measures entanglement left after transmitting locally prepared purifications. Section~\ref{sec:model} defines the quantities precisely and gives a finite cardinality bound for the optimization.

Our first main result is the feedback-capacity converse
\begin{equation}
 \Cfb(\N)\leq\UF(\N).
 \label{eq:intro-bound}
\end{equation}
Here $\Cfb$ is the ordinary vanishing-error classical capacity with arbitrary noiseless classical feedback, arbitrary local quantum memories, and no initial shared entanglement. More precisely, every $n$-use code with $M$ equiprobable messages and average error $\eps$ satisfies
\begin{equation}
 (1-\eps)\log_2M\leq n\UF(\N)+h_2(\eps).
 \label{eq:intro-finite}
\end{equation}
This is a weak converse, not an exact finite-error formula or a strong-converse statement.

The proof uses the potential $I(\Msg;B'T)+\EF(A';B'|\Msg T)$, where $\Msg$ is the message register, $T$ is the shared transcript, and $A'$ and $B'$ are the current memories of Alice and Bob. This potential has its roots in the conditional-entropy accounting of Eq.~(17) of~\cite{DQSW19}. Entanglement of formation agrees with that entropy term on conditionally pure laboratory states and can be smaller for mixed states. The pointwise comparison alone is not a converse; the mixed-state amortization proof supplies the needed channel estimate. We prove that local operations and classical feedback cannot increase the new potential, and that one use of $\N$ increases it by at most $\UF(\N)$. The proof first treats pure message-conditioned inputs and then refines mixed states into pure ensembles while retaining the information associated with the refinement.

The channel applications illustrate two complementary uses of this accounting. First, a flagged mixture of a noiseless channel and an entanglement-breaking channel has an exactly determined feedback-assisted capacity. For the qubit depolarizing channel, this simulation gives an analytic improvement over the published $\Upsilon$ bound in~\cite{DKQSW23}, while direct evaluation of $\UF$ is tighter still. At depolarizing probability $p=1/2$, $\UF\approx0.1889430634$ bits per use, close to the unassisted capacity $0.1887218755$ and far below the published $\Upsilon$ value $0.4056390622$. Second, for an amplitude damping channel, we obtain the exact channel-bound evaluation
\begin{equation}
 \UF(\cA_\gamma)=h_2\!\left(\frac{1+\sqrt\gamma}{2}\right).
 \label{eq:intro-AD}
\end{equation}
A fixed measurement of a mathematical Stinespring environment proves the scalar convexity needed for this formula. We compare it with the unassisted classical capacity, whose known single-use formula was proved additive by Tang, Zhu, Bai, and Wang~\cite{GF05,KSW20,TZBW26}, the semidefinite-programming bound $C_\beta$~\cite{WXD18}, and the maximum-output-entropy bound~\cite{DQSW19}.

\subsubsection{Randomness distillation assisted by local operations and classical communication}

The second main theme of this paper is randomness distillation. The task asks how many bits of nearly uniform shared randomness Alice and Bob can create from many copies of a bipartite state after subtracting the classical communication used to create them~\cite{DW04,LRWW23,NW24}. We introduce the one-way state quantity as a static counterpart of \eqref{eq:intro-UF}:
\begin{equation}
 \VF(\rho_{AB})\coloneqq
 \sup_{\mathcal M_{A\to XA'}}
 \left\{I(X;B)_\omega+\EF(A';B|X)_\omega\right\},
 \label{eq:intro-VF}
\end{equation}
where $\omega_{XA'B}\coloneqq(\mathcal M_{A\to XA'}\otimes\id_B)(\rho_{AB})$ and $X$ is classical. If $\psi_{ABE}$ purifies $\rho_{AB}$, then
\begin{equation}
 \VF(\rho_{AB})=S(B)_\rho-G_E(\rho_{BE}),
 \label{eq:intro-VF-roof}
\end{equation}
with $G_E$ the mixed convex roof of the classical correlation obtained by measuring $E$. A weighted decoder test gives a finite-blocklength converse for the net randomness yield. Combining it with the entanglement-of-formation lower bound of~\cite{YHHS05,BEHY11} and Koashi--Winter duality gives, in particular,
\begin{equation}
 \Rone(\rho_{AB})\leq\VF(\rho_{AB}).
 \label{eq:intro-random-upper}
\end{equation}
This is an upper bound on the one-way LOCC-assisted distillable randomness. It is not asserted to upper bound unrestricted LOCC-assisted randomness distillation.

Qubit isotropic states provide a bridge between the static and dynamical parts of the paper. For
$\rho^{(p)}_{AB}\coloneqq(1-p)\Phi_{AB}+pI_{AB}/4$, the exact one-way rate is equal to
\begin{equation}
 \Rone(\rho^{(p)})=1-h_2(p/2),
 \label{eq:intro-isotropic-oneway}
\end{equation}
and covariance implies that
\begin{equation}
 \VF(\rho^{(p)})=\UF(\cD_p).
 \label{eq:intro-state-channel}
\end{equation}
Thus the same entanglement-of-formation penalty appears in the isotropic state and depolarizing channel. The one-way isotropic evaluation is a direct corollary of Devetak--Winter and King~\cite{DW04,King03}; we use it as a benchmark and do not present it as a new independent coding theorem.

Choi states of amplitude-damping channels demonstrate the value of the one-way bound beyond that symmetric benchmark. For ordinary amplitude damping, the state potential equals the Choi entanglement of formation, and a recent support-additivity theorem~\cite{TZBW26} supplies exact operational rates for comparison. More importantly, for the full-rank generalized amplitude-damping Choi state $\rho^\star$ specified in Section~\ref{sec:damping-states}, we prove that
\begin{equation}
 0.43226906<\Rone(\rho^\star)\leq\VF(\rho^\star)<0.45252360.
 \label{eq:intro-thermal-bracket}
\end{equation}
Both computable state-$\Upsilon$ restrictions from~\cite{LRWW23} have values exceeding $0.66600769$ on this state. An affine lower bound on the complete input penalty, verified with outward interval arithmetic, controls every instrument rather than only the instruments sampled numerically. We do not evaluate $\Rone(\rho^\star)$ exactly, compare with the unrestricted bilinear $\Gamma$ optimization, or claim an improved two-way or strong-converse bound. A small white-noise perturbation gives a second explicit state with the same type of improvement.

The two-way problem behaves differently. We construct a two-round LOCC protocol in which Bob first performs a weak $Z$ measurement, compresses its outcome using Alice's quantum systems as side information, and Alice decodes gently enough to preserve the residual state. Alice then measures in the complementary basis and sends a compressed description of her outcome, which Bob decodes using his retained quantum output. For measurement strength $s\in[0,1]$, the net achievable rate is
\begin{equation}
 \begin{aligned}
 r_2(p,s)\coloneqq{}&1-h_2\!\left(\frac{1+(1-p)s}{2}\right)
 +h_2\!\left(\frac{1+s}{2}\right)\\
 &-h_2\!\left(\frac{1+\sqrt{(1-p)^2+[1-(1-p)^2]s^2}}{2}\right).
 \end{aligned}
 \label{eq:intro-two-round}
\end{equation}
We prove that for every $0<p<1$ there exists $s>0$ such that
\begin{equation}
 r_2(p,s)>1-h_2(p/2)=\Rone(\rho^{(p)}).
 \label{eq:intro-two-way-separation}
\end{equation}
For example, $s=4/5$ and $p=1/2$ give $0.2114185019\ldots$ bits per copy. The protocol therefore improves the explicit isotropic-state lower bounds used in~\cite{LRWW23,NW24}. It also shows that the straightforward two-way extension of~\eqref{eq:intro-random-upper} is false: for $2/3\leq p<1$, $\VF(\rho^{(p)})=1-h_2(p/2)$ while the two-round protocol is strictly larger. This obstruction identifies the essential difficulty for a sharper two-way converse: the final common random variable can be assembled by alternating classical extraction from residual quantum correlations in both directions.

Finally, $\UF$ in \eqref{eq:intro-UF} is monotone under output processing, convex in the channel, and subadditive under tensor products. Independent parallel feedback codes give the regularized channel bound $\Cfb(\N)\leq\inf_{k\in\mathbb{N}} k^{-1}\UF(\N^{\otimes k})$. The randomness-distillation results suggest a corresponding future program for state potentials that account for successive rounds of classical communication in both directions, with the communication cost subtracted from the randomness yield.

\subsection{Relation to prior work}
Bowen studied quantum feedback and its connection to entanglement-assisted communication~\cite{Bowen04}. Bowen and Nagarajan established no-increase results for classical feedback, including for entanglement-breaking channels~\cite{BN05}. Ref.~\cite{DW18} proved a strong converse and exponent bounds for feedback-assisted communication over entanglement-breaking channels. In contrast to these no-increase results, Smith and Smolin constructed random phase-coupling channels for which classical feedback gives an unbounded additive advantage as the input dimension grows~\cite{SS09}. For arbitrarily large $d$, these channels have unassisted classical capacity at most two bits per use, whereas a three-use protocol with classical feedback achieves at least $(2/3)\log_2d$ bits per use. The protocol first establishes entanglement and then uses it for dense coding, directly illustrating the need to account for entanglement retained between transmissions. The construction also admits finite-dimensional outputs by replacing the continuous random-unitary ensemble with a finite unitary two-design. Quantum feedback, by contrast, permits a backward quantum channel and must be distinguished from the classical feedback considered here~\cite{CMW16}.

Ding \textit{et al}.~obtained a maximum-output-entropy upper bound and applied it to erasure and bosonic channels~\cite{DQSW19}. Ding \textit{et al}.~subsequently developed strong-converse bounds for forward classical communication over bipartite channels~\cite{DKQSW23}. The depolarizing comparison in this paper uses the $\Upsilon$ curve in Figure~6 of that work. Its underlying semidefinite comparison set and channel divergences originate in~\cite{WXD18,WFT19,FF21}. The smaller numbers obtained here bound ordinary capacity; they do not automatically improve those strong-converse thresholds.

The expression $\UF$ equals the mixed-convex-roof expression in Winter and Yang's Theorem~10 on potential Holevo capacity~\cite{WY16}, based on the environment-correlation construction of Brand\~ao, Eisert, Horodecki, and Yang~\cite{BEHY11}. Koashi--Winter duality~\cite{KW04}, with the classical correlation of Henderson and Vedral~\cite{HV01}, gives the identification. Our channel contribution is the exact amortization argument for the actual-laboratory potential and its use in unrestricted classical-feedback protocols.

The one-way randomness-distillation task was solved by Devetak and Winter~\cite{DW04}. Lami \textit{et al}.~introduced strong-converse upper bounds that also apply to general LOCC-assisted distillation~\cite{LRWW23}; Nuradha \textit{et al.}~subsequently developed fidelity-smoothed one-shot and second-order versions~\cite{NW24}. The mixed-convex-roof quantity $G_E$ used here already appeared in~\cite{BEHY11,WY16}, and its lower-bound relation to entanglement cost is precisely what turns it into a one-way randomness upper bound. Our two-round construction instead belongs to the achievability side. Its coding ingredients are classical data compression with quantum side information~\cite{DW03} and gentle decoding; related multi-user common-randomness protocols and simultaneous quantum-side-information decoding were developed by Salek and Winter~\cite{SW22}. We do not claim that the general idea of retaining quantum outputs after a local instrument is new. The specific weak-measurement specialization, its closed rate formula, and its strict isotropic-state separation are the results established here.

Flagged extensions are an established technique for bounding channel capacities~\cite{SS08,SSW08,KFG22}. Our flagged-channel theorem concerns classical communication with feedback and explicitly accounts for entanglement created by noiseless outcomes. For depolarizing channels, King's theorem determines the unassisted classical capacity~\cite{King03}; the qubit calculation uses Wootters' formula~\cite{Wootters98} and concurrence factorization~\cite{Konrad08}. For amplitude damping, we compare with established Holevo and converse bounds~\cite{GF05,KSW20,WXD18}. Recent work on the amplitude damping channel is discussed where the corresponding quantities are used.

The damping examples use the generalized amplitude-damping convention and channel properties of Ref.~\cite{KSW20}. The scalar environment-correlation penalty and its convexification were derived by Fang~\cite{Fang26}; our state calculation retains the maximally mixed average input required by the Choi source. We therefore do not claim that scalar reduction as a separate new construction or an improvement over every fixed-input refinement of Fang's method. The ordinary damping rate is obtained by combining Koashi--Winter with the support theorem of Tang, Zhu, Bai, and Wang~\cite{KW04,TZBW26}. It is an exact benchmark, whereas the finite-temperature state provides the new certified application of the one-way converse. The two state-$\Upsilon$ restrictions used in the comparison are those of~\cite[Section~VI]{LRWW23}, not the channel-$\Upsilon$ quantity of~\cite{DKQSW23}.

\subsection{Organization}
The rest of our paper is organized as follows.
Sections~\ref{sec:model}--\ref{sec:structure} develop the communication model, the laboratory potential, its amortization theorem, and structural properties of the channel bound. Section~\ref{sec:flagged} proves the flagged-channel theorem. The channel applications then remain together: Section~\ref{sec:depolarizing} treats depolarizing noise and Section~\ref{sec:AD} evaluates amplitude damping. Section~\ref{sec:randomness} introduces the static task and one-way state bound, including a general Choi-state correspondence. Section~\ref{sec:damping-states} uses that correspondence for ordinary and generalized amplitude damping, with the latter supplying the principal effective one-way example. Sections~\ref{sec:isotropic} and~\ref{sec:twoway-randomness} present the isotropic benchmark and the two-round separation. Section~\ref{sec:numerics} describes the numerical checks, and Section~\ref{sec:conclusion} concludes. The appendices give concurrence calculations, the channel-$\Upsilon$ evaluation, the scalar depolarizing optimizer, reproducible certificates for the state comparisons, and a self-contained proof of the formation inequality used in the one-way converse.

\section{Notation, communication model, and channel bound}\label{sec:model}
We first distinguish the resources available in the communication task from the mathematical quantities used in its analysis. Quantum input and output systems are denoted by $A$ and $B$, while retained memories are denoted by $A'$ and $B'$. Subscripts index successive uses when necessary. No initial entanglement or free forward classical channel is supplied.

\subsection{States and information quantities}
We take all quantum systems to be finite dimensional throughout, and we take all logarithms to be binary unless written as $\ln$. The trace norm is denoted by $\left\|X\right\|_1$, and the operator norm by $\left\|X\right\|$. A quantum state is defined as a positive semidefinite operator of trace one. The set of states on $A$ is denoted by $\mathcal S(A)$, and $d_A$ denotes its Hilbert-space dimension. For a real number $x$, its positive part is defined as $x_+\coloneqq\max\{x,0\}$. A quantum channel is defined as a completely positive, trace-preserving linear map. A quantum instrument is defined as an indexed tuple of completely positive, trace-nonincreasing linear maps whose sum is a channel. A \emph{positive operator-valued measure} (POVM) is defined as an indexed tuple $(\Lambda_y)_y$ of positive semidefinite operators satisfying $\sum_y\Lambda_y=I$.

The von Neumann entropy of a state $\rho$ is defined as
\begin{equation}
 S(\rho)\coloneqq-\Tr[\rho\log_2\rho].
 \label{eq:entropy-def}
\end{equation}
For a multipartite state such as $\rho_{ABC}$, we write $S(A)_\rho\coloneqq S(\rho_A)$ to denote the entropy of the marginal state $\rho_A$. The quantum conditional entropy, quantum mutual information, and quantum conditional mutual information are defined, respectively, as
\begin{align}
 S(A|B)_\rho&\coloneqq S(AB)_\rho-S(B)_\rho,\\
 I(A;B)_\rho&\coloneqq S(A)_\rho+S(B)_\rho-S(AB)_\rho,\\
 I(A;B|C)_\rho&\coloneqq S(AC)_\rho+S(BC)_\rho-S(C)_\rho-S(ABC)_\rho.
 \label{eq:information-defs}
\end{align}
The Shannon entropy of a probability distribution and the binary entropy are defined as
\begin{equation}
 H((p_x)_x)\coloneqq-\sum_xp_x\log_2p_x,\qquad
 h_2(t)\coloneqq-t\log_2t-(1-t)\log_2(1-t),
 \label{eq:classical-entropy-defs}
\end{equation}
with $0\log_2 0\coloneqq0$ by continuity. Classical-register entropies are written as $H$ or $S$; their values agree. The fidelity of states, in the transition-probability convention of Uhlmann~\cite{Uhlmann76}, is defined as
\begin{equation}
 F(\rho,\sigma)\coloneqq\left\|\sqrt\rho\sqrt\sigma\right\|_1^2.
 \label{eq:fidelity-def}
\end{equation}
It equals the maximum squared overlap of their purifications. We consistently use this convention, including in every finite-error distillation criterion. The Fuchs--van de Graaf inequalities~\cite{FvdG99} give
\begin{equation}
 1-\sqrt{F(\rho,\sigma)}
 \leq\frac12\left\|\rho-\sigma\right\|_1
 \leq\sqrt{1-F(\rho,\sigma)}.
 \label{eq:fidelity-trace-distance}
\end{equation}
Thus fidelity at least $1-\eps$ implies normalized trace distance at most $\sqrt\eps$.

The quantum relative entropy of Umegaki~\cite{Umegaki62} is defined as
\begin{equation}
 D(\rho\|\sigma)\coloneqq
 \begin{cases}
 \Tr[\rho(\log_2\rho-\log_2\sigma)],&\supp\rho\subseteq\supp\sigma,\\
 +\infty,&\text{otherwise}.
 \end{cases}
 \label{eq:relative-entropy-def}
\end{equation}
The same expression is used for a positive semidefinite, not necessarily normalized second argument when explicitly required in the channel and state comparison programs. Its data-processing inequality~\cite{Lindblad75,Uhlmann77} and joint convexity state that
\begin{align}
 D(\rho\|\sigma)&\geq D(\N(\rho)\|\N(\sigma)),\\
 \sum_xp_xD(\rho^x\|\sigma^x)
 &\geq D\!\left(\sum_xp_x\rho^x\middle\|\sum_xp_x\sigma^x\right).
 \label{eq:relative-entropy-properties}
\end{align}
We also use strong subadditivity, equivalently $I(A;B|C)\geq0$, and its consequence that conditional entropy is concave~\cite{LR73}. These properties are reviewed in~\cite{Wilde17}.

For an ensemble $(p_x,\omega^x)_x$, its average state is defined as $\bar\omega\coloneqq\sum_xp_x\omega^x$, and its Holevo information is defined as
\begin{equation}
 \chi((p_x,\omega^x)_x)\coloneqq S(\bar\omega)-\sum_xp_xS(\omega^x)
 =\sum_xp_xD(\omega^x\|\bar\omega).
 \label{eq:Holevo}
\end{equation}
The channel Holevo information is defined as
\begin{equation}
 \chi(\N)\coloneqq\sup_{(p_x,\rho^x)_x}\chi((p_x,\N(\rho^x))_x).
 \label{eq:channel-Holevo}
\end{equation}
It is an achievable unassisted rate~\cite{Holevo98,SW97}; its regularization equals the unassisted classical capacity, denoted by $C(\N)$. No general additivity assumption is made.

\subsection{Entanglement of formation}
The second term of our potential measures entanglement that is actually shared by the laboratories. For a state $\rho_{AB}$, its entanglement of formation is defined as the convex-roof average of pure-state entanglement~\cite{BDSW96}:
\begin{equation}
 \EF(A;B)_\rho\coloneqq
 \inf\left\{\sum_zq_zS(A)_{\phi^z}:\rho_{AB}=\sum_zq_z\phi^z_{AB},\ \phi^z_{AB}\text{ pure}\right\}.
 \label{eq:EF}
\end{equation}
In finite dimensions the infimum is attained, and $\EF$ is continuous~\cite{Nielsen00}. It is invariant under local isometries, convex, and nonincreasing on average under local instruments~\cite{Vidal00}. On a pure state, it equals either marginal entropy. For an arbitrary state,
\begin{equation}
 0\leq\EF(A;B)_\rho\leq\min\{S(A)_\rho,S(B)_\rho\}.
 \label{eq:EF-entropy}
\end{equation}
For example, every pure-state decomposition gives $\sum_zq_zS(A)_{\phi^z}\leq S(A)_\rho$ by entropy concavity. Taking the infimum gives the first marginal upper bound, and exchanging the systems gives the second.

A readable classical flag permits an exact rather than merely convex upper bound. Define the flagged state
\begin{equation}
 \sigma_{ABY}\coloneqq\sum_yq_y\rho^y_{AB}\otimes\proj y_Y.
 \label{eq:flagged-state}
\end{equation}
Then
\begin{equation}
 \EF(A;BY)_\sigma=\sum_yq_y\EF(A;B)_{\rho^y}.
 \label{eq:EF-flags}
\end{equation}
To see the upper bound, apply convexity and invariance under appending a local pure state:
\begin{equation}
 \EF(A;BY)_\sigma\leq\sum_yq_y\EF(A;BY)_{\rho^y\otimes\proj y}
 =\sum_yq_y\EF(A;B)_{\rho^y}.
\end{equation}
For the reverse inequality, Bob measures $Y$ in its displayed orthonormal basis. Outcome $y$ has probability $q_y$ and leaves $\rho^y_{AB}$, so strong monotonicity of $\EF$ gives the reverse inequality. Thus the flag does not hide entanglement by mixing states whose labels Bob can distinguish. This identity is the main tool for evaluating the flagged simulations in Section~\ref{sec:flagged}.

For a classical--quantum state
\begin{equation}
\rho_{\Msg T A'B'}\coloneqq\sum_{m,t}p(m,t)\proj m_{\Msg}\otimes\proj t_T\otimes\rho^{m,t}_{A'B'},
\end{equation}
the conditional entanglement of formation is defined as
\begin{equation}
 \EF(A';B'|\Msg T)_\rho\coloneqq\sum_{m,t}p(m,t)\EF(A';B')_{\rho^{m,t}}.
 \label{eq:conditional-EF}
\end{equation}
Here conditioning means an ordinary average over classical labels, not a quantum conditional extension of $\EF$. When only one classical register is present, the same definition is used with the other register trivial. For instance, it vanishes when every conditional laboratory state is separable, even when those laboratories have classical correlations. The corresponding entropy identity
$S(B'|\Msg T)_\rho=\sum_{m,t}p(m,t)S(B')_{\rho^{m,t}}$
will connect our potential to~\cite[Eq.~(17)]{DQSW19}.

\subsection{The feedback-assisted task}\label{subsec:task}
We now define the protocols and the error criterion before taking a capacity limit. The same model will be used in the channel-use inequality and in the simulation arguments. Distinguishing the message from its classical register avoids confusing a uniform message alphabet with a posterior distribution after feedback.

Alice receives a message $m\in[M]\coloneqq\{1,\ldots,M\}$, recorded in the classical register $\Msg$. The symbols $W$ and $\widehat W$ denote the message and estimate registers; $M$ denotes only the number of messages. Unless conditioning is stated, the message is uniform. She and Bob use $\N\colon A\to B$ exactly $n$ times. At use $i$, Alice prepares $A_i$ while retaining $A_i'$, and Bob receives $B_i$ while holding his previous memory $B_i'$. Bob can apply an instrument to $B_iB_i'$ and report its outcome. After the intervening local operations and feedback, the parties hold the systems needed for the next use. The final estimate is denoted by $\widehat m$, and its classical register by $\Est$. Figure~\ref{fig:protocol} depicts three uses and the pre- and post-channel states that enter the converse. Its encoder $\mathcal L_i^{m,t}\colon A_{i-1}'\to A_i'A_i$ includes all of Alice's local processing between uses; the dependence on the controls is shown above each block. Here $A_0'$ is her initial local memory. Bob's instrument after use $i<n$ maps $B_i'B_i$ to his next memory $B_{i+1}'$ and a classical outcome.

There is no initial shared entanglement, free forward classical communication, or backward quantum communication. Bob may send a message-independent initial random seed. Initially shared separable correlations cause no additional difficulty: revealing a label in their separable decomposition can only enlarge the resources and leaves product states conditioned on that independent label. All memories and feedback alphabets may be arbitrarily large but finite. The schedule contains exactly $n$ forward uses and is fixed independently of the message, so timing is not an additional resource.

For a protocol $\mathcal P$, define its individual, average, and maximum error probabilities as
\begin{align}
 \eps_m(\mathcal P)&\coloneqq\Pr(\widehat m\neq m|m),\\
 \eps_{\rm av}(\mathcal P)&\coloneqq\frac1M\sum_{m=1}^M \eps_m(\mathcal P),\\
 \eps_{\max}(\mathcal P)&\coloneqq\max_{m\in[M]}\eps_m(\mathcal P).
 \label{eq:code-errors}
\end{align}
The notation in the first line means conditioning on the true message having the displayed value. We call $\mathcal P$ an $(n,M,\eps)$ protocol when $\eps_{\max}(\mathcal P)\leq\eps$. Following~\cite[Eqs.~(283)--(284)]{DKQSW23}, the $n$-shot feedback-assisted classical capacity in bits per use is defined as
\begin{equation}
 C_{\leftarrow}^{n,\eps}(\N)\coloneqq
 \sup\left\{\frac1n\log_2M:M\in\mathbb N,\ \text{there exists an }(n,M,\eps)\text{ protocol}\right\},
 \label{eq:n-shot-capacity}
\end{equation}
and the ordinary feedback-assisted classical capacity is defined as
\begin{equation}
 \Cfb(\N)\coloneqq\inf_{\eps\in(0,1)}\liminf_{n\to\infty}C_{\leftarrow}^{n,\eps}(\N).
 \label{eq:capacity-definition}
\end{equation}
In particular, the finite-error quantity is a rate, not a total number of bits. Definitions without feedback are obtained by removing the backward communication resource.

In this feedback model the average-error version of~\eqref{eq:n-shot-capacity} has the same value. Given an average-error code, Bob chooses a uniform cyclic shift $s\in\{0,\ldots,M-1\}$ and sends it initially to Alice. They run that code on label $1+((m-1+s)\bmod M)$ and undo the shift at the end. Each original message now has error exactly $\eps_{\rm av}(\mathcal P)$, averaged over the shared seed. No free forward communication has been introduced. We prove the finite-blocklength converse whenever the average error is at most the prescribed tolerance, so it also bounds the above maximum-error capacity. We do not identify average and maximum error for arbitrary no-feedback finite-blocklength codes.

\subsection{The quantity \texorpdfstring{$\UF$}{U-F}}\label{subsec:UF}
A single channel use can both convey message information and generate entanglement from a locally prepared purification. We combine these two contributions in a channel quantity before proving that it bounds a use within an adaptive protocol. The reference introduced here is a mathematical purification system; it is not transmitted for free.

For an input state $\rho_A$, choose a purification $\phi^\rho_{RA}$ with $R\simeq A$. Define
\begin{equation}
 \omega^\rho_{RB}\coloneqq(\id_R\otimes\N_{A\to B})(\phi^\rho_{RA}),\qquad
 f_\N(\rho)\coloneqq\EF(R;B)_{\omega^\rho}.
 \label{eq:fN}
\end{equation}
Different purifications are related by a local reference isometry, so the value does not depend on the choice. The entanglement-of-formation channel bound is defined as
\begin{equation}
 \UF(\N)\coloneqq\sup_{(p_x,\rho^x)_x}
 \left\{\chi((p_x,\N(\rho^x))_x)+\sum_xp_xf_\N(\rho^x)\right\},
 \label{eq:UF}
\end{equation}
where $x$ is a generic ensemble index and each state $\rho^x$ may be mixed. Its role differs from the physical message $m$, though an ensemble can be chosen by a classical control. For a pure input state the reference can be trivial and $f_\N(\rho^x)=0$. Restricting~\eqref{eq:UF} to pure ensembles gives $\chi(\N)$, but that restriction is not valid for the converse: a mixed input may be entangled with Alice's retained memory. Even a one-element ensemble can give a positive value through $f_\N$ alone.

The optimization has a finite alphabet bound. This also ensures that it is a maximum, not only a supremum.
\begin{proposition}[Cardinality and attainment]\label{prop:cardinality}
The supremum in~\eqref{eq:UF} is attained by an ensemble with at most $d_A^2$ members.
\end{proposition}
\begin{proof}
Define $g_\N(\rho)\coloneqq S(\N(\rho))-f_\N(\rho)$. This is continuous: a canonical purification formed using $\sqrt\rho$ depends continuously on $\rho$, and finite-dimensional entanglement of formation is continuous~\cite{Nielsen00}. Choose affine real coordinates $v(\rho)\in\mathbb R^{d_A^2-1}$ for the trace-one Hermitian input matrices, so that $v(\sum_xp_x\rho^x)=\sum_xp_xv(\rho^x)$. The graph
\begin{equation}
 \mathcal K\coloneqq\{(v(\rho),g_\N(\rho)):\rho\in\mathcal S(A)\}
 \subseteq\mathbb R^{d_A^2}
\end{equation}
is compact and connected, because the state space is compact and connected and the displayed map is continuous. An ensemble determines the point
$\big(v(\bar\rho),\sum_xp_xg_\N(\rho^x)\big)$ in its convex hull, where $\bar\rho\coloneqq\sum_xp_x\rho^x$. The Fenchel--Eggleston refinement of Carath\'eodory's theorem~\cite{Eggleston58} represents every such point using at most $d_A^2$ points of $\mathcal K$. The new ensemble therefore preserves both the average input and the average penalty. Its objective $S(\N(\bar\rho))-\sum_xp_xg_\N(\rho^x)$ is unchanged.

Allow zero probabilities to pad each reduced ensemble to $d_A^2$ members. The resulting simplex times $\mathcal S(A)^{d_A^2}$ is compact, and the objective is continuous. It consequently attains its maximum. No purity restriction on the ensemble members was used.
\end{proof}

\section{A potential that is monotone under classical feedback}\label{sec:potential}
The quantity~\eqref{eq:UF} describes a transmission with no receiver memory. To use it in an adaptive protocol, we need an accounting rule that also allows correlations created by earlier transmissions. We now define that rule on an arbitrary state of the actual laboratories and compare it with its earlier entropy-based counterpart.

\subsection{Message information and retained entanglement}
Let Alice's current laboratory be $A'$, Bob's current laboratory be $B'$, and let the classical transcript $T$ be available to both. Define the state
\begin{equation}
 \rho_{\Msg T A'B'}\coloneqq\sum_{m,t}p(m,t)\proj m_{\Msg}\otimes\proj t_T\otimes\rho^{m,t}_{A'B'}.
 \label{eq:potential-state}
\end{equation}
The joint distribution is arbitrary; after feedback it need not factorize. The potential is defined as
\begin{equation}
 \Q(\rho)\coloneqq I(\Msg;B'T)_\rho+\EF(A';B'|\Msg T)_\rho.
 \label{eq:potential}
\end{equation}
The physical bipartition and the classical labels are understood whenever $\Q$ is evaluated. Just before a transmission Alice's laboratory is $A'A$, and immediately after it Bob's is $B'B$; these full laboratories replace $A'$ and $B'$ in~\eqref{eq:potential}.

In our notation, the quantity of~\cite[Eq.~(17)]{DQSW19} is
\begin{equation}
 \Q_S(\rho)\coloneqq I(\Msg;B'T)_\rho+S(B'|\Msg T)_\rho.
 \label{eq:old-potential}
\end{equation}
For each conditionally pure laboratory state, $\EF(A';B')_{\rho^{m,t}}=S(B')_{\rho^{m,t}}$, so $\Q=\Q_S$. In general,~\eqref{eq:EF-entropy} gives
\begin{equation}
 \Q(\rho)\leq\Q_S(\rho),\qquad
 \Q_S(\rho)-\Q(\rho)
 =\sum_{m,t}p(m,t)\big[S(B')_{\rho^{m,t}}-\EF(A';B')_{\rho^{m,t}}\big].
 \label{eq:potential-comparison}
\end{equation}
Conditional receiver entropy can include mixing correlated with an inaccessible environment, rather than entanglement held by Alice. The earlier proof deals with this through an enlarged purified description. Our potential instead assigns entanglement only to the actual laboratories and treats their mixed states directly. The pointwise inequality~\eqref{eq:potential-comparison} alone does not bound potential increments; the amortization theorem in the next section is the additional step that makes the improvement operationally valid.

The entanglement term is an accounting term, not a claim that every noisy ebit has a fixed future communication value. For an explicit example, take four equiprobable messages and two noiseless qubit transmissions. Before the first, Alice holds a Bell state locally, so the pair $(I(\Msg;B'),\EF(A';B'|\Msg))$ is $(0,0)$. Sending one share produces a message-independent Bell state across the laboratories, giving $(0,1)$. Alice then encodes the message by a Pauli operator on her share and sends it. Bob now holds four orthogonal Bell states, so $I(\Msg;B')=2$, while Alice's remaining system is trivial and the entanglement term is zero. Thus
\begin{equation}
 (I,\EF):\ (0,0)\longrightarrow(0,1)\longrightarrow(2,0),\qquad
 \Q:\ 0\longrightarrow1\longrightarrow2.
 \label{eq:dense-accounting}
\end{equation}
The second message-information increase is two bits, but each transmission increases the total potential by one~\cite{BW92}. Figure~\ref{fig:dense-coding} shows this transfer between the terms.

\begin{figure}[t]\centering\begin{tikzpicture}[x=1cm,y=1cm,font=\small]
\draw[-{Latex},thin] (.5,4.0)--node[above,font=\footnotesize]{time}(12.2,4.0);
\node[anchor=east] at (0,2.85) {Alice};
\node[anchor=east] at (0,0) {Bob};
\node[op,minimum width=2.05cm] (prep) at (1.65,2.85) {Prepare Bell state};
\node[op,minimum width=2.05cm] (pauli) at (7.0,2.85) {Encode $m$\\by a Pauli};
\node[op,minimum width=1.7cm] (hold) at (4.1,0) {Retain share};
\node[op,minimum width=2.2cm] (dec) at (10.4,0) {Bell measurement};
\draw[qwire] (prep.east)--node[above,font=\footnotesize]{retained share}(pauli.west);
\draw[qwire] (prep.south)--node[sloped,above=3pt,font=\footnotesize]{qubit 1}(hold.north);
\draw[qwire] (hold.east)--node[below,font=\footnotesize]{retained share}(dec.west);
\draw[qwire] (pauli.south)--node[sloped,above=3pt,font=\footnotesize]{qubit 2}(dec.north);
\draw[cwire] (dec.east)--(12.55,0);
\node[font=\footnotesize,anchor=east] at (12.55,.65) {$\widehat m=m$};
\node[font=\footnotesize,align=center] at (7.0,3.65) {$m\in\{1,2,3,4\}$};
\draw[thin] (.1,-.8)--(12.5,-.8);
\node[align=center] at (1.8,-1.6) {before transmission\\$(I,E_F)=(0,0)$\\$\Q=0$};
\node[align=center] at (6.2,-1.6) {after the first use\\$(I,E_F)=(0,1)$\\$\Q=1$};
\node[align=center] at (10.6,-1.6) {after the second use\\$(I,E_F)=(2,0)$\\$\Q=2$};
\end{tikzpicture}
\caption{The two contributions to the potential in a two-use noiseless example. Initially both Bell shares are at Alice. The first forward qubit leaves a shared Bell state but no message information at Bob. After Pauli encoding and the second forward qubit, all four Bell states are at Bob and can be distinguished perfectly. The entanglement remaining across the laboratories has been consumed. The forward transfers move diagonally down and right; the horizontal wire retains a Bell share locally.}\label{fig:dense-coding}\end{figure}
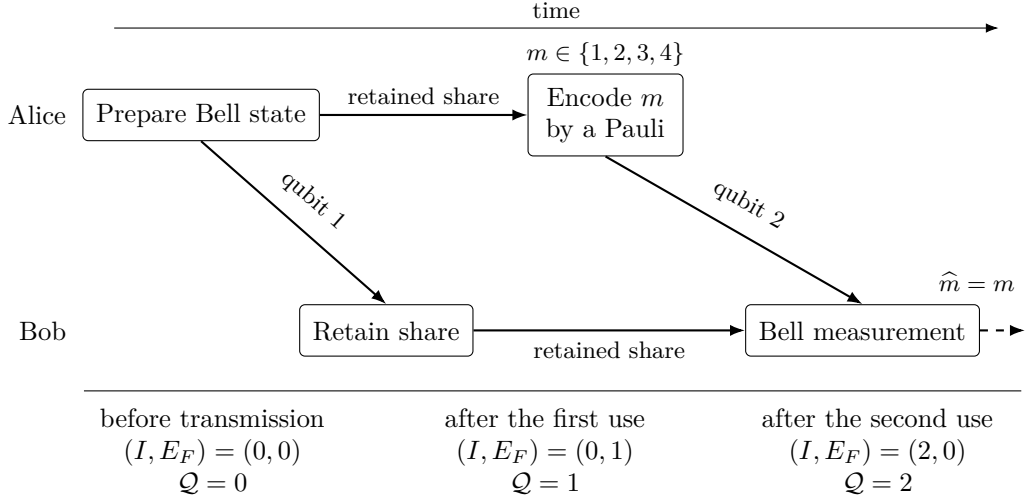

\begin{lemma}[Monotonicity under free operations]\label{lem:free}
The potential~\eqref{eq:potential} cannot increase under Alice's local channels, Bob's local channels, or receiver-to-sender classical feedback, including their compositions. It is initially equal to zero in the stated model, and $\Q(\rho)\geq I(\Msg;B'T)_\rho$.
\end{lemma}
\begin{proof}
The last claim follows term by term:
\begin{equation}
 \Q(\rho)-I(\Msg;B'T)_\rho
 =\sum_{m,t}p(m,t)\EF(A';B')_{\rho^{m,t}}\geq0.
\end{equation}
Let $T$ record an independent initial seed with distribution $(q(t))_t$. The initial joint state has the form
\begin{equation}
 \rho^0_{\Msg T A'B'}\coloneqq\frac1M\sum_{m,t}q(t)\proj m_{\Msg}\otimes\proj t_T
 \otimes\alpha_{A'}^{m,t}\otimes\beta_{B'}^t.
 \label{eq:initial-product-state}
\end{equation}
Its receiver marginal factorizes as
$\rho^0_{\Msg B'T}=\pi_{\Msg}\otimes\sum_tq(t)\beta_{B'}^t\otimes\proj t_T$, where $\pi_{\Msg}\coloneqq I_{\Msg}/M$.
Consequently $I(\Msg;B'T)_{\rho^0}=0$. Every conditional laboratory state is product, so its entanglement of formation is zero. Hence $\Q(\rho^0)=0$.

For Alice's local channel controlled by $m,t$, define
$\theta^{m,t}_{A''B'}\coloneqq(\mathcal L^{m,t}_{A'\to A''}\otimes\id_{B'})(\rho^{m,t}_{A'B'})$.
Trace preservation implies that
\begin{equation}
 \Tr_{A''}\theta^{m,t}_{A''B'}=\Tr_{A'}\rho^{m,t}_{A'B'}.
 \label{eq:Alice-no-signalling}
\end{equation}
For example, pairing either side with an arbitrary operator $O_{B'}$ gives the same value because $(\mathcal L^{m,t})^\dagger(I_{A''})=I_{A'}$. Thus the full marginal on $\Msg TB'$ is unchanged, and so is its mutual information. Local monotonicity also gives
\begin{equation}
 \sum_{m,t}p(m,t)\EF(A'';B')_{\theta^{m,t}}
 \leq\sum_{m,t}p(m,t)\EF(A';B')_{\rho^{m,t}}.
\end{equation}
This proves the claim for Alice, including private measurements whose records are retained in her memory.

For Bob's step, let $(\mathcal B_{y|t}:B'\to B'')_y$ be his transcript-controlled instrument. Define
$q(y|m,t)\coloneqq\Tr[(\id_{A'}\otimes\mathcal B_{y|t})(\rho^{m,t})]$
and the normalized conditional state $\theta^{m,t,y}_{A'B''}$ by dividing the latter output by $q(y|m,t)$ when it is positive. Terms of zero probability are omitted. The resulting joint state is defined as
\begin{equation}
 \theta_{\Msg T A'B''Y}\coloneqq\sum_{m,t,y}p(m,t)q(y|m,t)\proj m_{\Msg}\otimes\proj t_T\otimes\theta^{m,t,y}_{A'B''}\otimes\proj y_Y.
\end{equation}
The total operation on $B'T$ is a channel to $B''TY$, so data processing gives
\begin{equation}
 I(\Msg;B''TY)_\theta\leq I(\Msg;B'T)_\rho.
 \label{eq:free-MI}
\end{equation}
For each $m,t$, strong monotonicity gives
\begin{equation}
 \sum_yq(y|m,t)\EF(A';B'')_{\theta^{m,t,y}}
 \leq\EF(A';B')_{\rho^{m,t}}.
 \label{eq:free-EF}
\end{equation}
Multiply by $p(m,t)$ and sum to obtain the conditional-entanglement comparison with transcript $TY$. Bob keeps $Y$ and sends a copy to Alice; the copy is already accounted for by conditioning on $TY$. Adding this comparison to~\eqref{eq:free-MI} proves the claim. A deterministic Bob channel is the one-outcome case, and unreported outcomes can remain inside his quantum memory. Subsequent local channels are covered by the preceding argument.
\end{proof}

\subsection{What remains to be proved}
Free-operation monotonicity establishes how the potential behaves between channel uses, but does not yet bound a forward use. The input can be entangled with the receiver's memory, and its conditional state can be mixed. We therefore need a one-step theorem for an arbitrary state on both laboratories, not only a purification prepared at Alice. The next section first relates two entanglement cuts, then proves the channel increment for pure conditional states, and finally controls the extra information introduced by refining mixed states.

\section{Exact amortization and the feedback converse}\label{sec:amortization}
We now connect the single-use channel expression to the potential of an arbitrary adaptive protocol. The main issue is that the system purifying a channel input may be distributed between Alice and Bob. An elementary comparison of entanglement across two cuts handles this issue. A second argument extends the calculation from pure to mixed conditional inputs without making the virtual decomposition label available to either party.

\subsection{An entanglement relocation inequality}
At a channel output, let $A'$ be Alice's retained memory, $B'$ Bob's old memory, and $B$ the new output. The physical output entanglement is $\EF(A';B'B)$. For a pure conditional input, however, $A'B'$ purifies the transmitted system, so $f_\N$ involves $\EF(A'B';B)$. Figure~\ref{fig:relocation} displays the difference between these cuts. The following lemma compares them on the same state, without moving any physical system.

\begin{lemma}[Relocation inequality]\label{lem:relocation}
For an arbitrary state $\sigma_{A'B'B}$,
\begin{equation}
 \EF(A';B'B)_\sigma\leq\EF(A'B';B)_\sigma+S(B'|B)_\sigma.
 \label{eq:relocation}
\end{equation}
\end{lemma}
\begin{proof}
Choose a decomposition $\sigma_{A'B'B}=\sum_zq_z\psi^z_{A'B'B}$ into pure states attaining $\EF(A'B';B)_\sigma$. The same decomposition is admissible, though not necessarily optimal, for $\EF(A';B'B)_\sigma$. Therefore,
\begin{align}
 \EF(A';B'B)_\sigma
 &\leq\sum_zq_zS(B'B)_{\psi^z}\\
 &=\sum_zq_zS(B)_{\psi^z}+\sum_zq_zS(B'|B)_{\psi^z}\\
 &\leq\EF(A'B';B)_\sigma+S(B'|B)_\sigma.
 \label{eq:relocation-proof}
\end{align}
The last inequality follows from concavity of conditional entropy. More explicitly, for $\zeta_{ZA'B'B}\coloneqq\sum_zq_z\proj z_Z\otimes\psi^z_{A'B'B}$, strong subadditivity gives $S(B'|BZ)_\zeta\leq S(B'|B)_\zeta=S(B'|B)_\sigma$. The left side is the required average conditional entropy. This argument remains valid when $S(B'|B)_\sigma$ is negative.
\end{proof}

\begin{figure}[t]\centering\begin{tikzpicture}[x=1cm,y=1cm,font=\small]
\node[font=\bfseries] at (2.5,2.2) {actual laboratory cut};
\node[font=\bfseries] at (9.8,2.2) {comparison cut};
\foreach \x/\name in {1/A',2.7/B',4.4/B,8.1/A',9.8/B',11.5/B}{
 \node[op,minimum width=1.1cm,minimum height=1cm] at (\x,.9) {$\name$};}
\draw[dashed,thick] (1.85,-.1)--(1.85,1.8);
\draw[dashed,thick] (10.65,-.1)--(10.65,1.8);
\node at (2.7,-.55) {$\EF(A';B'B)_\sigma$};
\node at (9.8,-.55) {$\EF(A'B';B)_\sigma$};
\draw[<->,thick] (5.35,.9)--node[above,align=center,font=\footnotesize]{same\\state}(7.15,.9);
\node at (6.25,-1.45) {$\EF(A';B'B)_\sigma\leq\EF(A'B';B)_\sigma+S(B'|B)_\sigma$};
\end{tikzpicture}
\caption{Two bipartitions of $\sigma_{A'B'B}$. The physical output cut is $A'|B'B$. In the comparison cut $A'B'|B$, both old memories are treated mathematically as the reference for the transmitted input. The correction in Lemma~\ref{lem:relocation} is $S(B'|B)$, which may be negative. The comparison does not grant free transmission of Bob's old memory.}\label{fig:relocation}\end{figure}
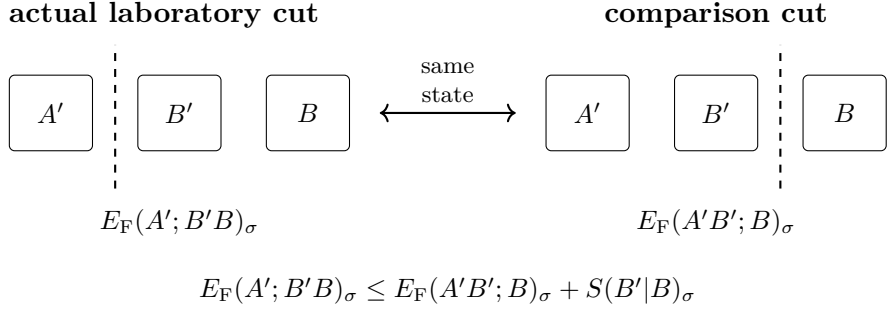

\subsection{The one-step theorem}
The relocation inequality converts the physical entanglement term into one appearing in $f_\N$. We now combine it with the receiver's message information. The classical label is allowed to have an arbitrary distribution: conditioning an adaptive protocol on its transcript generally changes the posterior message probabilities. We suppress that transcript here and restore it in the operational converse.

\begin{theorem}[Exact potential amortization]\label{thm:amortization}
Let $\rho_{\Msg A'AB'}\coloneqq\sum_mp_m\proj m_{\Msg}\otimes\rho^m_{A'AB'}$ be an arbitrary classical--quantum state, with $A'A$ at Alice and $B'$ at Bob. Define
$\omega_{\Msg A'B'B}\coloneqq(\id_{\Msg A'B'}\otimes\N_{A\to B})(\rho_{\Msg A'AB'})$,
with tensor factors ordered as indicated. Then
\begin{equation}
 I(\Msg;B'B)_\omega+\EF(A';B'B|\Msg)_\omega-I(\Msg;B')_\rho-\EF(A'A;B'|\Msg)_\rho\leq\UF(\N).
 \label{eq:amortization}
\end{equation}
The supremum of the left-hand side over finite memories and classical--quantum input states is exactly $\UF(\N)$.
\end{theorem}
\begin{proof}
The potential~\eqref{eq:potential}, with trivial transcript and the appropriate physical cut, defines
\begin{align}
 \Q_{\rm in}&\coloneqq I(\Msg;B')_\rho+\EF(A'A;B'|\Msg)_\rho,\\
 \Q_{\rm out}&\coloneqq I(\Msg;B'B)_\omega+\EF(A';B'B|\Msg)_\omega.
 \label{eq:input-output-potentials}
\end{align}
Thus, we aim to prove that~\eqref{eq:amortization} bounds $\Q_{\rm out}-\Q_{\rm in}$.

\par\medskip\noindent\textbf{Pure conditional input states.}
Suppose first that every $\rho^m_{A'AB'}=\psi^m_{A'AB'}$ is pure. Its entanglement across $A'A|B'$ equals $S(B')_{\psi^m}$, so
\begin{align}
 \Q_{\rm in}
 &=S(B')_\rho-\sum_mp_mS(B')_{\psi^m}+\sum_mp_m\EF(A'A;B')_{\psi^m}\\
 &=S(B')_\rho.
 \label{eq:pure-in}
\end{align}
Define $\rho_A^m\coloneqq\Tr_{A'B'}\psi^m_{A'AB'}$. Since $A'B'$ purifies this input marginal, local-isometry invariance of entanglement of formation gives
\begin{equation}
 \EF(A'B';B)_{\omega^m}=f_\N(\rho_A^m).
 \label{eq:joint-reference}
\end{equation}
Apply Lemma~\ref{lem:relocation} separately for every $m$:
\begin{align}
 \Q_{\rm out}
 &=S(B'B)_\omega-\sum_mp_mS(B'B)_{\omega^m}+\sum_mp_m\EF(A';B'B)_{\omega^m}\\
 &\leq S(B'B)_\omega-\sum_mp_mS(B'B)_{\omega^m}
   +\sum_mp_m\big[f_\N(\rho_A^m)+S(B'|B)_{\omega^m}\big]\\
 &=S(B'B)_\omega-\sum_mp_mS(B)_{\omega^m}+\sum_mp_mf_\N(\rho_A^m).
 \label{eq:pure-out}
\end{align}
The last equality follows from $S(B'B)-S(B'|B)=S(B)$, separately on each conditional state. The channel leaves $B'$ unchanged, so $S(B')_\rho=S(B')_\omega$. Subtracting~\eqref{eq:pure-in} gives
\begin{align}
 \Q_{\rm out}-\Q_{\rm in}
 &\leq S(B|B')_\omega-\sum_mp_mS(\N(\rho_A^m))+\sum_mp_mf_\N(\rho_A^m)\\
 &\leq S\!\left(\N\!\left(\sum_mp_m\rho_A^m\right)\right)
       -\sum_mp_mS(\N(\rho_A^m))+\sum_mp_mf_\N(\rho_A^m)\\
 &\leq\UF(\N).
 \label{eq:pure-amortization}
\end{align}
The second inequality follows from $S(B|B')\leq S(B)$, and the last follows from~\eqref{eq:UF}. The marginal inputs $\rho_A^m$ can be mixed even though the full conditional inputs were pure.

\par\medskip\noindent\textbf{Mixed conditional input states.}
For each $m$, choose a pure decomposition
\begin{equation}
 \rho^m_{A'AB'}=\sum_zq(z|m)\psi^{mz}_{A'AB'}
 \label{eq:refine-state}
\end{equation}
attaining the \emph{initial} $\EF(A'A;B')_{\rho^m}$. Define the refined states
\begin{align}
 \widetilde\rho_{\Msg ZA'AB'}&\coloneqq\sum_{m,z}p_mq(z|m)\proj m_{\Msg}\otimes\proj z_Z\otimes\psi^{mz}_{A'AB'},\\
 \widetilde\omega_{\Msg ZA'B'B}&\coloneqq(\id_{\Msg ZA'B'}\otimes\N)(\widetilde\rho_{\Msg ZA'AB'}).
 \label{eq:refined-states}
\end{align}
Their marginals after discarding $Z$ are the original states. The refined potentials are defined explicitly as
\begin{align}
 \widetilde\Q_{\rm in}
 &\coloneqq I(\Msg Z;B')_{\widetilde\rho}
       +\EF(A'A;B'|\Msg Z)_{\widetilde\rho},
 \label{eq:refined-potential-in}\\
 \widetilde\Q_{\rm out}
 &\coloneqq I(\Msg Z;B'B)_{\widetilde\omega}
       +\EF(A';B'B|\Msg Z)_{\widetilde\omega}.
 \label{eq:refined-potential-out}
\end{align}
These are the same function~\eqref{eq:potential} with the combined classical label $\Msg Z$. The virtual label $Z$ is not communicated or supplied to either laboratory.

At the input, the chosen decomposition preserves the average entanglement term. The chain rule for mutual information therefore gives
\begin{equation}
 \widetilde\Q_{\rm in}-\Q_{\rm in}=I(Z;B'|\Msg)_{\widetilde\rho}.
 \label{eq:refinement-in}
\end{equation}
At the output, the excess is
\begin{align}
 \widetilde\Q_{\rm out}-\Q_{\rm out}
 &=I(Z;B'B|\Msg)_{\widetilde\omega}
 +\sum_mp_m\!\left[\sum_zq(z|m)\EF(A';B'B)_{\omega^{mz}}-\EF(A';B'B)_{\omega^m}\right]\\
 &\geq I(Z;B'B|\Msg)_{\widetilde\omega}\\
 &\geq I(Z;B'|\Msg)_{\widetilde\omega}
 = I(Z;B'|\Msg)_{\widetilde\rho}.
 \label{eq:refinement-out}
\end{align}
The first inequality follows from convexity of $\EF$. The second follows from the chain rule and nonnegativity of $I(Z;B|\Msg B')$, and the last equality follows because the channel leaves $\Msg ZB'$ unchanged. To see the direction of the resulting increment comparison, first write the exact identity
\begin{align}
 \Q_{\rm out}-\Q_{\rm in}
 &=\widetilde\Q_{\rm out}-\widetilde\Q_{\rm in}
   -\big(\widetilde\Q_{\rm out}-\Q_{\rm out}\big)
   +\big(\widetilde\Q_{\rm in}-\Q_{\rm in}\big)\\
 &\leq\widetilde\Q_{\rm out}-\widetilde\Q_{\rm in}.
 \label{eq:refinement-subtraction}
\end{align}
The inequality follows because~\eqref{eq:refinement-out} is at least~\eqref{eq:refinement-in}. Applying the already proved pure-conditional-state result to the refined label gives
\begin{equation}
 \Q_{\rm out}-\Q_{\rm in}\leq\widetilde\Q_{\rm out}-\widetilde\Q_{\rm in}\leq\UF(\N).
 \label{eq:mixed-amortization}
\end{equation}
The comparison keeps track of both excesses, not only the output one. Ignoring either excess would not justify replacing a mixed input by a pure one.

\par\medskip\noindent\textbf{Tightness for the potential.}
For an ensemble in~\eqref{eq:UF}, let Bob's initial system be trivial and let Alice hold a purification of each ensemble state. The initial potential is zero. After one transmission its mutual-information term is exactly the output Holevo information and its entanglement term is the average $f_\N$. Thus the potential increment is the objective of that ensemble. Taking the supremum gives equality in the optimized statement; Proposition~\ref{prop:cardinality} also guarantees an attaining finite ensemble.
\end{proof}

\FloatBarrier

\subsection{The operational converse}
\begin{figure}[t]\centering\begin{tikzpicture}[x=1cm,y=1cm,font=\small]
\draw[-{Latex},thin] (.6,4.45)--node[above,font=\footnotesize]{time}(13.05,4.45);
\node[anchor=east] at (0,3.25) {Alice};
\node[anchor=east] at (0,0) {Bob};
\foreach \i/\ax/\bx in {1/1.35/3.35,2/5.55/7.55,3/9.75/11.75}{
 \node[op,minimum width=1.25cm] (a\i) at (\ax,3.25) {$\mathcal L_\i$};
 \node[op,minimum width=.7cm] (n\i) at ({(\ax+\bx)/2},1.62) {$\N$};
 \draw[qwire] (a\i.south)--node[left=2pt,pos=.38,font=\footnotesize]{$A_\i$}(n\i.north);
 \node[font=\footnotesize,anchor=east,align=right] at ({\ax+.15},2.2) {\smash{\raisebox{0.70cm}{$\rho^{(\i)}$}}\; $\Q_\i^-$};
 \draw[densely dotted] ({\ax+.25},2.2)--({\ax+.64},2.2);
 \node[font=\footnotesize,anchor=east,align=right] at ({\bx-1.15},.88) {$\omega^{(\i)}$\; $\Q_\i^+$};
 \draw[densely dotted] ({\bx-1.03},.88)--({\bx-.55},.88);
}
\node[op,minimum width=1.8cm] (b1) at (3.35,0) {$(\mathcal B_{f_1})_{f_1}$};
\node[op,minimum width=1.9cm] (b2) at (7.55,0) {$(\mathcal B_{f_2|f_1})_{f_2}$};
\node[op,minimum width=2.1cm] (b3) at (11.75,0) {$(\Lambda_{\widehat m|f})_{\widehat m}$};
\foreach \i in {1,2,3}{
 \draw[qwire] (n\i.south)--node[left=1pt,pos=.67,font=\footnotesize]{$B_\i$}(b\i.north);
}
\draw[qwire] (.05,3.25)--node[above,font=\footnotesize]{$A_0'$}(a1.west);
\draw[qwire] (a1.east)--node[above,font=\footnotesize]{$A_1'$}(a2.west);
\draw[qwire] (a2.east)--node[above,font=\footnotesize]{$A_2'$}(a3.west);
\draw[qwire] (a3.east)--node[above,font=\footnotesize]{$A_3'$}(13.15,3.25);
\draw[qwire] (.15,0)--node[below,font=\footnotesize]{$B_1'$}(b1.west);
\draw[qwire] (b1.east)--node[below,font=\footnotesize]{$B_2'$}(b2.west);
\draw[qwire] (b2.east)--node[below,font=\footnotesize]{$B_3'$}(b3.west);
\draw[cwire] (b3.east)--node[above,font=\footnotesize]{$\widehat m$}(13.5,0);
\draw[cwire] (b1.north)--node[sloped,above=3pt,font=\footnotesize]{$f_1$}(a2.south west);
\draw[cwire] (b2.north)--node[sloped,above=3pt,font=\footnotesize]{$f_2$}(a3.south west);
\node[font=\footnotesize] at (1.35,3.95) {$m$};
\node[font=\footnotesize] at (5.55,3.95) {$m,f_1$};
\node[font=\footnotesize] at (9.75,3.95) {$m,f_1,f_2$};
\node[font=\footnotesize,align=center] at (6.8,-.9) {$f=(f_1,f_2)$; dotted ticks label the full pre- and post-channel states.};
\end{tikzpicture}
\caption{Three uses of $\N\colon A_i\to B_i$ with classical feedback. Retained memories carry primes. The pre-channel state $\rho^{(i)}$ contains $A_i'A_i$ at Alice and $B_i'$ at Bob; the post-channel state $\omega^{(i)}$ contains $A_i'$ at Alice and $B_i'B_i$ at Bob. The dotted markers label these complete states, not only the adjacent wire. Each interval between uses contains the receiver instrument, classical feedback, and the next encoder. The converse bounds $\Q_i^+-\Q_i^-$ at every channel use and $\Q_{i+1}^--\Q_i^+$ between uses. The initial seed is suppressed.}\label{fig:protocol}\end{figure}
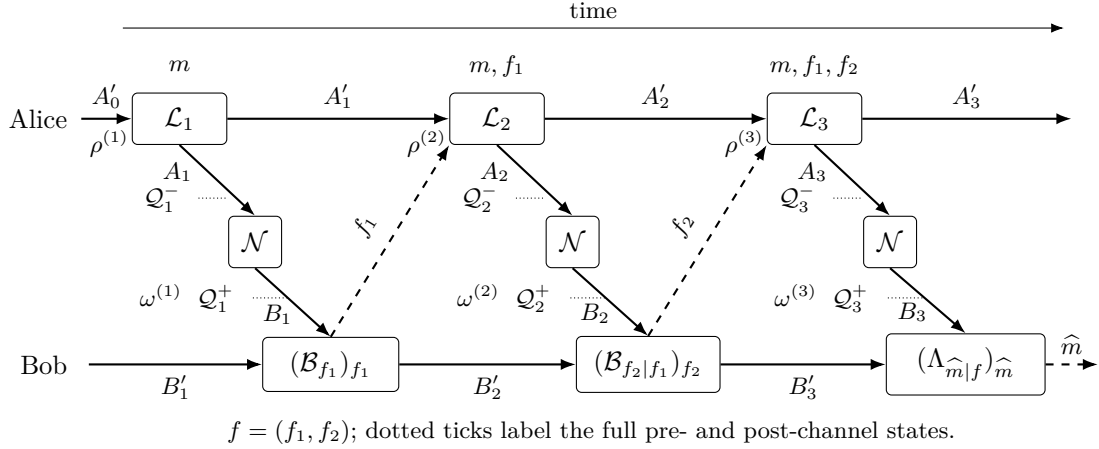

The one-step theorem allows arbitrary correlated laboratories and arbitrary posterior message distributions. It can therefore be applied to every pre-channel state in Figure~\ref{fig:protocol}. We make the sum of increments explicit, following the organization of the proof of~\cite[Theorem~4]{DQSW19}. The new step is the use of actual-laboratory entanglement of formation and its $\UF$ increment in place of the earlier conditional entropy and output-entropy increment.

\begin{samepage}
\begin{theorem}[Classical-feedback capacity bound]\label{thm:converse}
For an $n$-use feedback-assisted code in the model of Section~\ref{sec:model}, with $M$ equiprobable messages and average error at most $\eps\in[0,1)$, the following inequality holds:
\begin{equation}
 (1-\eps)\log_2M\leq n\UF(\N)+h_2(\eps).
 \label{eq:finite-converse}
\end{equation}
In particular, $\Cfb(\N)\leq\UF(\N)$ for the capacity defined in~\eqref{eq:capacity-definition}.
\end{theorem}
\end{samepage}
\begin{proof}
Let $T_{i-1}$ contain a complete record of all receiver-to-sender classical communication before channel use $i$, including the initial message-independent seed when present. No earlier feedback is discarded from this register. Its value determines the classical controls used by the next encoder, while both laboratories may also retain quantum memory. Denote the states immediately before and after use $i$ by
$\rho^{(i)}_{\Msg T_{i-1}A_i'A_iB_i'}$ and
$\omega^{(i)}_{\Msg T_{i-1}A_i'B_i'B_i}$, respectively, as marked in Figure~\ref{fig:protocol}. Define their potentials as
\begin{align}
 \Q_i^-&\coloneqq I(\Msg;B_i'T_{i-1})_{\rho^{(i)}}+\EF(A_i'A_i;B_i'|\Msg T_{i-1})_{\rho^{(i)}},\\
 \Q_i^+&\coloneqq I(\Msg;B_i'B_iT_{i-1})_{\omega^{(i)}}+\EF(A_i';B_i'B_i|\Msg T_{i-1})_{\omega^{(i)}}.
 \label{eq:step-potentials}
\end{align}
Let $\Q_0$ and $\Q_{\rm fin}$ denote the values before initial local processing and after the final decoder. These are all instances of~\eqref{eq:potential}, not different potential functions.

During channel use $i$, the joint classical state of $\Msg T_{i-1}$ does not change. The chain rule
$I(\Msg;B_i'T_{i-1})=I(\Msg;T_{i-1})+I(\Msg;B_i'|T_{i-1})$
therefore cancels the transcript-only term in the increment. Applying Theorem~\ref{thm:amortization} conditional on every transcript value gives
\begin{equation}
 \Q_i^+-\Q_i^-
 =\sum_t\Pr(T_{i-1}=t)\,\Delta\Q_{i,t}
 \leq\sum_t\Pr(T_{i-1}=t)\UF(\N)=\UF(\N),
 \label{eq:conditional-channel-increment}
\end{equation}
where $\Delta\Q_{i,t}$ is defined as the one-step difference in~\eqref{eq:amortization} on that conditional state. Lemma~\ref{lem:free} gives
\begin{equation}
 \Q_1^-\leq\Q_0=0,\qquad
 \Q_{i+1}^-\leq\Q_i^+\ (1\leq i<n),\qquad
 \Q_{\rm fin}\leq\Q_n^+.
\end{equation}
Adding and subtracting the intermediate potentials produces the exact telescoping identity
\begin{align}
 \Q_{\rm fin}-\Q_0
 &= (\Q_1^--\Q_0)+\sum_{i=1}^n(\Q_i^+-\Q_i^-)
   +\sum_{i=1}^{n-1}(\Q_{i+1}^--\Q_i^+)+(\Q_{\rm fin}-\Q_n^+)\\
 &\leq\sum_{i=1}^n(\Q_i^+-\Q_i^-)\leq n\UF(\N).
 \label{eq:telescoping}
\end{align}
Every term except the forward increments is nonpositive. Figure~\ref{fig:accounting} summarizes exactly these terms. Since the final Bob laboratory contains $\Est$, nonnegativity of $\EF$ and data processing imply $I(\Msg;\Est)\leq\Q_{\rm fin}\leq n\UF(\N)$.

\begin{figure}[t]\centering\begin{tikzpicture}[x=1cm,y=1cm,font=\small]
\node[op,minimum width=1.7cm,minimum height=.9cm] (qm) at (1,0) {$\Q_i^-$};
\node[op,minimum width=1.7cm,minimum height=.9cm] (qp) at (6,0) {$\Q_i^+$};
\node[op,minimum width=1.7cm,minimum height=.9cm] (qn) at (11,0) {$\Q_{i+1}^-$};
\draw[qwire] (qm.east)--node[above=5pt,align=center]{forward channel\\use $i$}(qp.west);
\draw[cwire] (qp.east)--node[above=5pt,align=center]{feedback and\\local processing}(qn.west);
\node[align=center] at (3.5,-1.05) {$\Q_i^+-\Q_i^-\leq\UF(\N)$};
\node[align=center] at (8.5,-1.05) {$\Q_{i+1}^--\Q_i^+\leq0$};
\draw[thin] (.1,-1.65)--(11.9,-1.65);
\node[align=center] at (6,-2.25) {$\Q_0=0$\quad and\quad $I(\Msg;\Est)\leq\Q_{\rm fin}$};
\node[align=center] at (6,-3.05) {Add every increment:\quad
$\displaystyle\Q_{\rm fin}-\Q_0=\sum_{\text{all steps}}\Delta\Q\leq n\UF(\N)$};
\end{tikzpicture}
\caption{The potential increases only at forward channel uses. One representative forward use and the subsequent local-processing interval are shown. The notation $\Q_i^-$ and $\Q_i^+$ denotes the pre- and post-channel values marked in the detailed protocol of Figure~\ref{fig:protocol}. Each forward increment is at most $\UF(\N)$; each intervening local/feedback increment is nonpositive. Adding all increments gives the explicit telescoping estimate in the proof of Theorem~\ref{thm:converse}.}\label{fig:accounting}\end{figure}
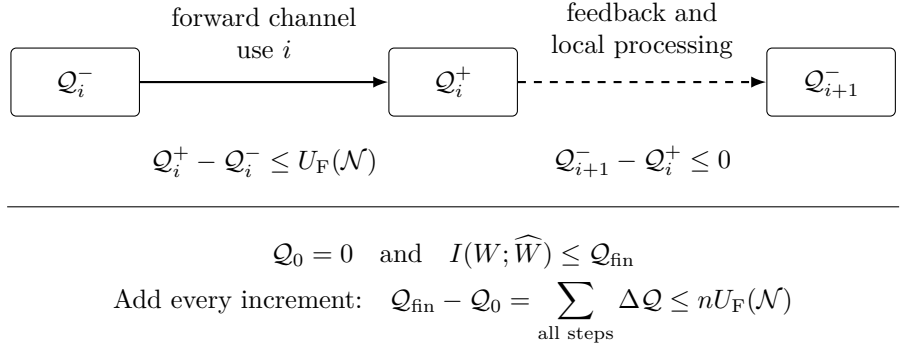

Suppose first that $M\geq2$ and $0\leq\eps\leq1-1/M$. Fano's inequality gives an upper bound on $H(\Msg|\Est)$ in terms of the actual error. The function $h_2(x)+x\log_2(M-1)$ is nondecreasing on $[0,1-1/M]$, because its derivative is $\log_2[(1-x)(M-1)/x]$. Since the actual error is at most $\eps$, it follows that
\begin{align}
 \log_2M
 &=I(\Msg;\Est)+H(\Msg|\Est)\\
 &\leq n\UF(\N)+h_2(\eps)+\eps\log_2(M-1)\\
 &\leq n\UF(\N)+h_2(\eps)+\eps\log_2M.
 \label{eq:fano}
\end{align}
Rearrangement proves~\eqref{eq:finite-converse} in this range. If $\eps>1-1/M$, then the claimed bound already follows from
\begin{equation}
 h_2(\eps)=h_2(1-\eps)
 \geq(1-\eps)\log_2\frac1{1-\eps}
 \geq(1-\eps)\log_2M,
 \label{eq:fano-large-tolerance}
\end{equation}
and $\UF\geq0$. The case $M=1$ is immediate. Thus the theorem uses only an upper bound on error, for every $\eps<1$, without assuming that binary entropy itself is monotone on the whole interval.

Applying the result to every admissible protocol gives
\begin{equation}
 C_{\leftarrow}^{n,\eps}(\N)
 \leq\frac{\UF(\N)+h_2(\eps)/n}{1-\eps}.
 \label{eq:finite-capacity-converse}
\end{equation}
Taking the lower limit in $n$ and then the infimum over $\eps\in(0,1)$, as prescribed by~\eqref{eq:capacity-definition}, proves $\Cfb(\N)\leq\UF(\N)$.
\end{proof}

\begin{remark}[Scope and comparison with the earlier entropy bound]\label{rem:weak}
Theorem~\ref{thm:amortization} identifies the least state-independent upper bound on the increment of this particular potential. It does not assert achievability of that increment as a communication rate. Likewise, the factor $1-\eps$ in~\eqref{eq:finite-converse} prevents this argument from establishing a strong converse at $\UF$.

The telescoping proof parallels~\cite[Theorem~4]{DQSW19}, whose purified-protocol potential is~\eqref{eq:old-potential} and whose channel increment is at most $S(B)$. Our mixed-state refinement supplies the step needed to work on the actual laboratories. The simple estimate $f_\N(\rho)\leq S(\N(\rho))$ then shows from~\eqref{eq:UF} that $\UF(\N)\leq\sup_\rho S(\N(\rho))$. Thus the new numerical capacity upper bound never weakens the earlier maximum-output-entropy bound, although its proof requires more than the pointwise potential comparison. Strong-converse thresholds in~\cite{DKQSW23,FF21} remain distinct statements.
\end{remark}

\section{Earlier convex-roof bounds and structural properties}\label{sec:structure}
The amortization proof establishes a feedback interpretation of $\UF$. We now identify the earlier channel expression to which it is equal and record properties useful for simulations and possible improvements. This separates the operational theorem from the history of the optimized quantity itself.

\subsection{Identification with the Winter--Yang expression}
Choose a Stinespring isometric extension $V\colon A\to BE$ of the channel $\mathcal{N}_{A\to B}$. For a bipartite state $\sigma_{BE}$, the environment-measured classical correlation is defined as
\begin{equation}
 J(E;B)_\sigma\coloneqq S(B)_\sigma-
 \inf_{(M_y)_y}\sum_yq_yS(B)_{\sigma^y},
 \label{eq:J-env}
\end{equation}
where $(M_y)_y$ is a POVM on $E$. Its outcome probabilities are defined as $q_y\coloneqq\Tr[(I_B\otimes M_y)\sigma_{BE}]$, and its conditional output states are defined as $\sigma_B^y\coloneqq q_y^{-1}\Tr_E[(I_B\otimes M_y)\sigma_{BE}]$ for $q_y>0$. Terms with $q_y=0$ are omitted. The measured subsystem is specified by the notation and will not be interchanged. This is the classical correlation of~\cite{HV01}. Its mixed convex roof is defined as
\begin{equation}
 G_E(\sigma)\coloneqq\inf\left\{\sum_xp_xJ(E;B)_{\sigma^x}:
 \sigma=\sum_xp_x\sigma^x\right\}.
 \label{eq:mixed-roof}
\end{equation}
The component states in~\eqref{eq:mixed-roof} may be mixed. This feature distinguishes $G_E$ from a pure-state convex roof.

\begin{proposition}[Identification of the channel expression]\label{prop:WY}
For every finite-dimensional channel,
\begin{equation}
 \UF(\N)=\sup_\rho\left\{S(\N(\rho))-G_E(V\rho V^\dagger)\right\}.
 \label{eq:WY-identification}
\end{equation}
The right side is the mixed-convex-roof upper bound on potential Holevo capacity in~\cite[Theorem~10]{WY16}, based on~\cite{BEHY11}.
\end{proposition}
\begin{proof}
For an input $\rho_A$, define $\sigma_{BE}^\rho\coloneqq V\rho_AV^\dagger$. Applying $V$ to a purification $\phi_{RA}^\rho$ gives a pure state on $RBE$. Koashi--Winter duality~\cite{KW04}, with the measurement on $E$, gives
\begin{equation}
 S(\N(\rho))-f_\N(\rho)=J(E;B)_{\sigma^\rho}.
 \label{eq:KW-identity}
\end{equation}
Fix an average input $\bar\rho_A$. Every ensemble decomposition $\bar\rho_A=\sum_xp_x\rho_A^x$ produces the decomposition $V\bar\rho_AV^\dagger=\sum_xp_xV\rho_A^xV^\dagger$. We must also check the converse direction because $G_E$ optimizes over all mixed-state decompositions of the output state.

Define $P\coloneqq VV^\dagger$, the projection onto the range of $V$. If $V\bar\rho_AV^\dagger=\sum_xp_x\sigma^x_{BE}$, then
\begin{equation}
 0=\Tr[(I-P)V\bar\rho_AV^\dagger]
   =\sum_xp_x\Tr[(I-P)\sigma^x_{BE}].
\end{equation}
Each summand is nonnegative. Thus $\Tr[(I-P)\sigma^x_{BE}]=0$ for every $x$ of positive weight. Positivity gives $(I-P)(\sigma^x_{BE})^{1/2}=0$, and hence $\sigma^x_{BE}=P\sigma^x_{BE}P$. The operator $\rho_A^x\coloneqq V^\dagger\sigma^x_{BE}V$ is therefore positive, has trace one, and satisfies $V\rho_A^xV^\dagger=\sigma^x_{BE}$. Moreover, $\sum_xp_x\rho_A^x=\bar\rho_A$. This gives a bijection between the two classes of decompositions.

Consequently, for fixed $\bar\rho_A$,
\begin{equation}
 \inf_{\sum_xp_x\rho_A^x=\bar\rho_A}
 \sum_xp_x\big[S(\N(\rho_A^x))-f_\N(\rho_A^x)\big]
 =G_E(V\bar\rho_AV^\dagger).
 \label{eq:WY-fixed-average}
\end{equation}
The average-output entropy $S(\N(\bar\rho_A))$ is fixed in this minimization. Subtracting the common minimum from that entropy and then optimizing over $\bar\rho_A$ yields~\eqref{eq:WY-identification}. Thus the identification preserves the full mixed-ensemble optimization rather than replacing it by a pure-state roof.
\end{proof}

This identification does not identify feedback capacity with potential Holevo capacity. It says that an existing expression has the feedback amortization property proved in Theorem~\ref{thm:amortization}. In particular, ordinary Holevo additivity for a channel need not make~\eqref{eq:WY-identification} equal to that channel's capacity. The strict depolarizing gap in Section~\ref{subsec:strict-UF} illustrates this distinction.

\subsection{Basic bounds and output processing}
The bound should recover elementary cases before it is used in less symmetric examples. It should also respect a receiver's ability to locally process an output. Both properties follow directly from its ensemble formulation.

\begin{proposition}[Basic properties]\label{prop:basic}
For every finite-dimensional channel,
\begin{equation}
 \max\{\chi(\N),\sup_\rho f_\N(\rho)\}
 \leq\UF(\N)\leq\sup_\rho S(\N(\rho)).
 \label{eq:basic-bounds}
\end{equation}
For every channel $\M$ acting on its output,
\begin{equation}
 \UF(\M\circ\N)\leq\UF(\N).
 \label{eq:UF-post}
\end{equation}
In particular, $\UF(\N)\leq\log_2d_A$. Moreover, $\UF$ is convex as a function of the channel.
\end{proposition}
\begin{proof}
Restrict to pure input ensembles for the first lower bound, and use one-element ensembles for the second. Inequality~\eqref{eq:EF-entropy} gives $f_\N(\rho)\leq S(\N(\rho))$, so all subtracted conditional output entropies in~\eqref{eq:UF} dominate the entanglement terms. This proves the upper bound.

For a fixed ensemble, output processing cannot increase its Holevo information or any entanglement term. Taking the supremum proves~\eqref{eq:UF-post}. For the identity channel, $f_{\id}(\rho)=S(\rho)$, so $\UF(\id_A)=\log_2d_A$. Apply~\eqref{eq:UF-post} with $\N=\id_A$.

Finally, for a fixed input ensemble the Holevo expression in~\eqref{eq:Holevo} is convex in the channel by joint convexity of relative entropy. Each entanglement term is convex by convexity of $\EF$ on the output of its fixed purification. Their sum is convex, and a supremum of convex functions is convex.
\end{proof}

For an entanglement-breaking channel $\cE$, $f_\cE=0$, so $\UF(\cE)=\chi(\cE)$. Theorem~\ref{thm:converse} and ordinary Holevo achievability recover the known equality $\Cfb(\cE)=C(\cE)=\chi(\cE)$~\cite{BN05,Shor02}. For a replacer channel $\UF=0$. For a $d$-dimensional erasure channel of erasure probability $p$, the orthogonal flag gives
\begin{equation}
 S(\cE_p(\rho))=h_2(p)+(1-p)S(\rho),\qquad
 f_{\cE_p}(\rho)=(1-p)S(\rho),
 \label{eq:erasure-check}
\end{equation}
and hence $\UF(\cE_p)=(1-p)\log_2d$, recovering the established first-order feedback capacity~\cite{DQSW19,DKQSW23}; the exact finite-blocklength result is given in~\cite{Wilde26Erasure}.

Convexity in Proposition~\ref{prop:basic} is a property of the upper bound. We do not assume convexity of the operational feedback capacity. A channel mixture can create correlations and resources that must still be accounted for; the flagged-channel proof in Section~\ref{sec:flagged} does so explicitly.

\subsection{Subadditivity and a regularized converse}
A tensor-product evaluation can improve a single-letter upper bound only if its use is justified in the adaptive model. We give that justification by parallel coding, rather than treating consecutive adaptive inputs as if they had been prepared together.

\begin{proposition}[Tensor products and regularization]\label{prop:regularization}
For finite-dimensional channels $\N$ and $\M$,
\begin{equation}
 \UF(\N\otimes\M)\leq\UF(\N)+\UF(\M).
 \label{eq:subadditive}
\end{equation}
Moreover,
\begin{equation}
 \Cfb(\N)\leq\inf_{k\in\mathbb{N}}\frac1k\UF(\N^{\otimes k})
 =\lim_{k\to\infty}\frac1k\UF(\N^{\otimes k}).
 \label{eq:regularized}
\end{equation}
\end{proposition}
\begin{proof}
Write $\N\colon A_1\to B_1$ and $\M\colon A_2\to B_2$. Fix an arbitrary finite ensemble $(p_x,\tau^x_{A_1A_2})_x$, and let $\phi^x_{RA_1A_2}$ purify each member, with all three quantum systems initially held by Alice. Define
\begin{align}
 \zeta^0_{XRA_1A_2}
 &\coloneqq\sum_xp_x\proj x_X\otimes\phi^x_{RA_1A_2},\\
 \zeta^1_{XRA_2B_1}
 &\coloneqq(\id_{XRA_2}\otimes\N_{A_1\to B_1})(\zeta^0),\\
 \zeta^2_{XRB_1B_2}
 &\coloneqq(\id_{XRB_1}\otimes\M_{A_2\to B_2})(\zeta^1).
 \label{eq:product-three-states}
\end{align}
Tensor factors are reordered as displayed. Bob's initial laboratory is trivial, so the potential at the first state is $\Q_0=0$. Its subsequent values are
\begin{align}
 \Q_1&\coloneqq I(X;B_1)_{\zeta^1}+\EF(RA_2;B_1|X)_{\zeta^1},\\
 \Q_2&\coloneqq I(X;B_1B_2)_{\zeta^2}+\EF(R;B_1B_2|X)_{\zeta^2}.
 \label{eq:product-potentials}
\end{align}
The first transmission is covered by Theorem~\ref{thm:amortization} with Alice's retained memory $RA_2$. For the second transmission, her retained memory is $R$, her input is $A_2$, and Bob's old memory is $B_1$. The conditional states after the first transmission can be mixed, which is precisely why the mixed-state part of that theorem is needed. Applying it twice gives
\begin{equation}
 \Q_1-\Q_0\leq\UF(\N),\qquad
 \Q_2-\Q_1\leq\UF(\M),\qquad
 \Q_2\leq\UF(\N)+\UF(\M).
 \label{eq:product-increments}
\end{equation}
To identify $\Q_2$, set $\bar\tau\coloneqq\sum_xp_x\tau^x$. The two local channel applications commute, so $\zeta^2$ is the output of the product channel on the original purification ensemble. Therefore
\begin{align}
 \Q_2
 &=S((\N\otimes\M)(\bar\tau))
   -\sum_xp_xS((\N\otimes\M)(\tau^x))
   +\sum_xp_xf_{\N\otimes\M}(\tau^x).
 \label{eq:product-objective}
\end{align}
This is the defining objective of $\UF(\N\otimes\M)$ for the chosen ensemble. Since the ensemble was arbitrary, its supremum is bounded by~\eqref{eq:product-increments}, proving~\eqref{eq:subadditive}.

Fix $k\geq1$ and $\eps\in(0,1)$. Choose an $(n,M,\eps/k)$ protocol whose rate approaches $C_{\leftarrow}^{n,\eps/k}(\N)$. Run $k$ independent copies in parallel, using independent messages and seeds. One use of $\N^{\otimes k}$ advances all copies by one step, and Bob returns the tuple of their feedback outcomes before the next product-channel use. Hence the protocol remains adaptive in each copy. Its message is a tuple in $[M]^k$, so its rate is $k\log_2M/n$. For every message tuple, the union bound gives error at most $k(\eps/k)=\eps$. Taking the supremum over the original protocol gives
\begin{equation}
 C_{\leftarrow}^{n,\eps}(\N^{\otimes k})
 \geq k C_{\leftarrow}^{n,\eps/k}(\N).
 \label{eq:parallel-n-shot}
\end{equation}
Taking lower limits in $n$ and using~\eqref{eq:capacity-definition},
\begin{equation}
 \liminf_{n\to\infty}C_{\leftarrow}^{n,\eps}(\N^{\otimes k})
 \geq k\liminf_{n\to\infty}C_{\leftarrow}^{n,\eps/k}(\N)
 \geq k\Cfb(\N).
\end{equation}
This holds for every $\eps\in(0,1)$, so the infimum over $\eps$ proves $\Cfb(\N^{\otimes k})\geq k\Cfb(\N)$. Apply Theorem~\ref{thm:converse} to $\N^{\otimes k}$ and divide by $k$. Since this holds for every $k$, the infimum in~\eqref{eq:regularized} is a capacity upper bound. Finally,~\eqref{eq:subadditive} makes the sequence of finite numbers $a_k\coloneqq\UF(\N^{\otimes k})$ subadditive, so the subadditive-sequence lemma gives $\lim_k a_k/k=\inf_k a_k/k$.
\end{proof}

The argument does not assert that a single adaptive protocol using $\N$ many times can be partitioned into independent nonadaptive blocks. Instead, independent copies of the entire feedback protocol establish the capacity comparison needed for~\eqref{eq:regularized}.

\section{Flagged-channel simulations}\label{sec:flagged}
The preceding output-processing inequality has an operational interpretation. If $\N=\mathcal T\circ\mathcal F$, Bob can simulate each use of $\N$ by applying $\mathcal T$ locally to the output of $\mathcal F$, before carrying out the next instrument in the original feedback code. Thus $\Cfb(\N)\leq\Cfb(\mathcal F)\leq\UF(\mathcal F)$. A useful extension supplies enough extra information to simplify $\mathcal F$, but not so much that it removes the noise completely~\cite{SS08,KFG22}.

For example, an erasure channel can simulate a depolarizing channel: Bob replaces each erasure flag by a maximally mixed state. The exact erasure feedback formula~\cite{Wilde26Erasure} then gives a valid converse, but the replacement event carries no input information, and the resulting depolarizing upper bound is relatively loose. We instead reveal an input-independent choice between a noiseless channel and an informative entanglement-breaking channel. The latter can convey classical information even though it cannot generate entanglement. The next theorem evaluates this flagged extension, explicitly accounting for entanglement generated during earlier noiseless uses.

\subsection{Noiseless and entanglement-breaking components}
Let $\cE\colon A\to B$ be entanglement breaking, and let $0\leq\lambda\leq1$. The flagged channel is defined as
\begin{equation}
 \cF_{\lambda,\cE}(\rho)\coloneqq(1-\lambda)\rho\oplus\lambda\cE(\rho).
 \label{eq:flagged-channel}
\end{equation}
The first sector is a copy of the input space, and the second is the output space of $\cE$. Bob can read the sector without disturbing its quantum state. The choice is independent of the input. The direct-sum definition does not require the two sector dimensions to agree. When the component outputs are identified with the same space, discarding the flag gives their ordinary convex mixture. An erased-output flag is a special case with $\cE$ a replacer. Figure~\ref{fig:flagged} depicts the construction.

\begin{theorem}[Exact capacity of the flagged channel]\label{thm:flagged}
For the channel in~\eqref{eq:flagged-channel},
\begin{align}
 \Cfb(\cF_{\lambda,\cE})
 &=C(\cF_{\lambda,\cE})=\chi(\cF_{\lambda,\cE})=\UF(\cF_{\lambda,\cE})\\
 &=\sup_{(p_x,\psi^x)_x}\left\{(1-\lambda)S(\bar\rho)
             +\lambda\chi((p_x,\cE(\psi^x))_x)\right\},
 \label{eq:flagged-capacity}
\end{align}
where the supremum in the last line is over pure input ensembles $(p_x,\psi_A^x)_x$, and their generally mixed average is defined as $\bar\rho_A\coloneqq\sum_xp_x\psi_A^x$.
\end{theorem}
\begin{proof}
For a purification of $\rho$, the noiseless component retains entanglement $S(\rho)$ and the entanglement-breaking component has none. Equation~\eqref{eq:EF-flags} gives
\begin{equation}
 f_{\cF_{\lambda,\cE}}(\rho)=(1-\lambda)S(\rho).
 \label{eq:flagged-EF}
\end{equation}
The flag entropy $h_2(\lambda)$ cancels from the output Holevo information. Adding~\eqref{eq:flagged-EF} to that information yields
\begin{equation}
 \UF(\cF_{\lambda,\cE})
 =\sup_{(p_x,\rho^x)_x}\left\{(1-\lambda)S(\bar\rho)
             +\lambda\chi((p_x,\cE(\rho^x))_x)\right\}.
 \label{eq:flagged-objective}
\end{equation}
Here the average state for the mixed ensemble is defined as $\bar\rho\coloneqq\sum_xp_x\rho^x$. Refining each member into a pure-state ensemble preserves this average. By concavity of output entropy, the refined average of the conditional output entropies cannot increase, so its Holevo information cannot decrease. For a pure input ensemble, the expression is exactly the Holevo information of $\cF_{\lambda,\cE}$. Thus $\UF(\cF_{\lambda,\cE})=\chi(\cF_{\lambda,\cE})$. The inequalities
\begin{equation}
 \chi(\cF_{\lambda,\cE})\leq C(\cF_{\lambda,\cE})
 \leq\Cfb(\cF_{\lambda,\cE})\leq\UF(\cF_{\lambda,\cE})
 \label{eq:flagged-capacity-sandwich}
\end{equation}
prove all the asserted equalities.
\end{proof}

The theorem does not in general identify the result with a weighted sum of the separate component capacities. It gives that sum as an upper bound, but one ensemble must optimize both terms to attain it. The depolarizing application in Section~\ref{sec:depolarizing} has precisely this property. The proof also explains why merely averaging component capacities would be insufficient: the noiseless uses can distribute entanglement that remains available during later uses.

The general theorem now leads to a concrete simulation bound for depolarizing noise. We develop that application next, before evaluating $\UF$ directly. The later state examples will reuse the same convex-roof structure with a fixed average input.

\begin{figure}[t]\centering\begin{tikzpicture}[font=\small,x=1cm,y=1cm]
\node (input) at (0,0) {$\rho_A$};
\node[op,minimum width=1.7cm,minimum height=1.1cm] (choice) at (2,0) {random\\choice};
\node[op,minimum width=2.8cm,minimum height=1.1cm] (id) at (6,1.65) {noiseless channel\\$\id_A$};
\node[op,minimum width=2.8cm,minimum height=1.3cm] (eb) at (6,-1.65) {entanglement-breaking\\channel $\cE$};
\node[op,minimum width=2.25cm,minimum height=1.1cm] (out) at (10.3,0) {Bob's quantum\\output and flag};
\draw[qwire] (input)--(choice);
\draw[qwire] (choice.north east)--node[above,sloped]{$1-\lambda$}(id.west);
\draw[qwire] (choice.south east)--node[below,sloped]{$\lambda$}(eb.west);
\draw[qwire] (id.east)--node[above,sloped]{flag $0$}(out.north west);
\draw[qwire] (eb.east)--node[below,sloped]{flag $1$}(out.south west);
\node[align=center] at (6,-3.0) {$\cF_{\lambda,\cE}(\rho)=(1-\lambda)\rho\oplus\lambda\cE(\rho)$\\[5pt]
Input-independent flag; it may be returned to Alice by classical feedback.};
\end{tikzpicture}
\caption{An input-independent choice between a noiseless channel and an entanglement-breaking channel. Bob retains a readable flag. When the component output spaces agree, discarding it gives their unflagged mixture, so every feedback code for that mixture can be implemented using the flagged channel. The converse for the flagged channel includes entanglement generated by its noiseless uses; it is not obtained by assuming that operational feedback capacities are convex.}\label{fig:flagged}\end{figure}

\FloatBarrier

\section{The qubit depolarizing channel}\label{sec:depolarizing}
We now specialize the general channel results to depolarizing noise and collect two complementary converse constructions. First, the general flagged-extension theorem of Section~\ref{sec:flagged} gives an analytic simulation bound. Second, the ensemble formulation of $\UF$ gives a stronger qubit evaluation because covariance removes the optimization over average states and two-qubit concurrence evaluates the entanglement term. We then compare both bounds with prior converses and examine the small residual gap above the unassisted capacity.

\subsection{Flagged and erasure simulation bounds}
We now choose the entanglement-breaking component at the depolarizing threshold. Full unitary covariance makes the same input ensemble optimal for both terms of the flagged-channel formula, so its value can be evaluated without a numerical optimization. We use the convention
\begin{equation}
 \cD_{p,d}(\rho)\coloneqq(1-p)\rho+pI_d/d,\qquad 0\leq p\leq1.
 \label{eq:dep-def}
\end{equation}
Define the threshold $p_*\coloneqq d/(d+1)$. The channel is entanglement breaking for $p\geq p_*$. For example, its normalized Choi state is isotropic with maximally entangled overlap $1-p+p/d^2$; separability holds exactly when this overlap is at most $1/d$~\cite{Shor02,DKQSW23}.

For $0\leq p\leq p_*$, set $\lambda\coloneqq p/p_*$. Then
\begin{equation}
 \cD_{p,d}=(1-\lambda)\id_d+\lambda\cD_{p_*,d}.
 \label{eq:dep-mixture}
\end{equation}
Discarding the flag of $\cF_{\lambda,\cD_{p_*,d}}$ implements this channel exactly inside an arbitrary adaptive protocol. The uniform computational-basis ensemble simultaneously maximizes the input entropy and the output Holevo information of $\cD_{p_*,d}$. Every pure input to that channel has spectrum $2/\left(d+1\right),1/\left(d+1\right),\ldots,1/\left(d+1\right)$. Theorem~\ref{thm:flagged} therefore gives
\begin{equation}
 \Cfb(\cD_{p,d})\leq U_{\rm flag}(p,d)
 \coloneqq\log_2d-p\left[\frac{d+1}{d}\log_2(d+1)-\frac2d\right].
 \label{eq:flagged-dep-bound}
\end{equation}
In addition, output-processing monotonicity gives $\UF(\cD_{p,d})\leq U_{\rm flag}(p,d)$.

For comparison, a $d$-dimensional erasure channel of erasure probability $p$ simulates~\eqref{eq:dep-def} when Bob replaces the erasure flag by $I_d/d$. The resulting bound is $(1-p)\log_2d$. It is valid for adaptive protocols, but for qubits it is weaker than~\eqref{eq:flagged-dep-bound}. Replacing the erasure probability by a larger value cannot improve this one-use simulation: the trace distance of two orthogonal inputs becomes $2(1-p)$ under $\cD_{p,d}$, whereas an erasure probability $q$ followed by output processing gives at most $2(1-q)$. Thus $q\leq p$ is necessary.

\subsection{Reduction to one real variable}
The goal of this subsection is to derive an analytical formula for $\UF(\cD_p)$, the entanglement-of-formation bound of the qubit depolarizing channel. The flagged extension gives an upper bound by revealing additional noise information; we now evaluate the potential of the original channel without revealing that information. Two symmetries make the ensemble optimization manageable: unitary covariance makes each conditional penalty depend only on the input eigenvalues, and Pauli averaging makes the average output maximally mixed without changing those penalties. The remaining entanglement term is a two-qubit quantity evaluated by Wootters' formula. These observations reduce the full optimization to one input Bloch-vector length.

Define $\cD_p\coloneqq\cD_{p,2}$ and let $r\in[0,1]$ be the input Bloch-vector length. Its output entropy is
\begin{equation}
 S(\cD_p(\rho))=h_2\!\left(\frac{1+(1-p)r}{2}\right).
 \label{eq:dep-output-entropy}
\end{equation}
The two-qubit formation function is defined as
\begin{equation}
 \mathcal E(c)\coloneqq h_2\!\left(\frac{1+\sqrt{1-c^2}}2\right),
 \qquad 0\leq c\leq1.
 \label{eq:concurrence-entropy}
\end{equation}
Wootters' formula gives $\EF=\mathcal E(c)$ for a two-qubit state of concurrence $c$~\cite{Wootters98}. For the depolarizing channel, define $\kappa\coloneqq(1-3p/2)_+$, its normalized Choi-state concurrence. This parameter enters both the scalar optimization and the small-noise and near-pure-input analyses below. Sending one subsystem of a purification of $\rho$ through $\cD_p$ produces concurrence $\kappa\sqrt{1-r^2}$; Appendix~\ref{app:concurrence} proves this using the local-filter factorization in~\cite{Konrad08}.

\begin{proposition}[Scalar depolarizing formula]\label{prop:dep-UF}
For $0\leq p\leq1$,
\begin{equation}
 \UF(\cD_p)=\max_{0\leq r\leq1}
 \left\{1-h_2\!\left(\frac{1+(1-p)r}{2}\right)
       +\mathcal E\!\left(\kappa\sqrt{1-r^2}\right)\right\}.
 \label{eq:dep-UF}
\end{equation}
\end{proposition}
\begin{proof}
Define the single-state penalty $g_p(\rho)\coloneqq S(\cD_p(\rho))-f_{\cD_p}(\rho)$. For every qubit unitary $U$, the depolarizing channel satisfies
$\cD_p(U\rho U^\dagger)=U\cD_p(\rho)U^\dagger$. A purification of $U\rho U^\dagger$ is obtained by applying $U$ to the input subsystem of a purification of $\rho$. Covariance then relates the two joint reference--output states by a local output unitary, up to a reference isometry. Hence both their output entropy and entanglement of formation are unchanged. Therefore $g_p(\rho)$ depends only on the eigenvalues of $\rho$, or equivalently on its Bloch-vector length $r$.

For an arbitrary ensemble, its objective satisfies
\begin{equation}
 S(\cD_p(\bar\rho))-\sum_xp_xg_p(\rho^x)
 \leq1-\sum_xp_xg_p(\rho^x)
 \leq1-\min_{\rho}g_p(\rho).
\end{equation}
The first inequality uses the qubit output dimension. To attain this upper bound for a chosen input $\rho$, let $(P_j)_{j=0}^3$ be the identity and the three Pauli matrices. Use the ensemble $(1/4,P_j\rho P_j^\dagger)_{j=0}^3$. Its average is $I_2/2$, because Pauli averaging cancels every component of the Bloch vector, and its average output is also $I_2/2$. All four penalties equal $g_p(\rho)$ by the preceding covariance argument. Thus its objective is exactly $1-g_p(\rho)$. Choosing a minimizing input establishes
\begin{equation}
 \UF(\cD_p)=1-\min_{\rho}g_p(\rho).
\end{equation}
Finally, the eigenvalues of $\rho$ are $(1\pm r)/2$, so $\det\rho=(1-r^2)/4$. The concurrence identity in Appendix~\ref{app:concurrence} gives $f_{\cD_p}(\rho)=\mathcal E(\kappa\sqrt{1-r^2})$, while~\eqref{eq:dep-output-entropy} gives the output entropy. Substitution proves~\eqref{eq:dep-UF}, with a maximum attained by continuity on $r\in[0,1]$.
\end{proof}

For $p\geq2/3$, $\kappa=0$, and~\eqref{eq:dep-UF} reduces to $1-h_2(p/2)$. This is the exact classical feedback-assisted capacity in the entanglement-breaking regime. Without feedback, King's theorem gives the same formula $C(\cD_p)=1-h_2(p/2)$ for the entire interval $0\leq p\leq1$~\cite{King03}. For smaller noise, the bound~\eqref{eq:dep-UF} permits mixed ensemble elements and can exceed this unassisted rate.

\subsection{Comparison with the published converse}
To compare with the depolarizing curve in~\cite[Figure~6]{DKQSW23}, it is useful to evaluate the symmetry-reduced $\Upsilon$ expression explicitly. Appendix~\ref{app:upsilon} gives the underlying semidefinite calculation. The result is
\begin{equation}
 \Upsilon(\cD_p)=
 \begin{cases}
 1-\dfrac{3p}{4}\log_2 3,&0\leq p\leq2/3,\\
 2-H(1-3p/4,p/4,p/4,p/4),&2/3\leq p\leq1,
 \end{cases}
 \label{eq:dep-upsilon}
\end{equation}
where $H$ is the Shannon entropy of the displayed probability vector. The geometric-divergence formulation used for the published comparison has the same commuting Choi limit here~\cite{WFT19,FF21,DKQSW23}.

In the non-entanglement-breaking regime, the flagged expression becomes
\begin{equation}
 U_{\rm flag}(p,2)=1-\left(\frac32\log_2 3-1\right)p,
 \label{eq:qubit-flag}
\end{equation}
and its improvement is analytic:
\begin{equation}
 \Upsilon(\cD_p)-U_{\rm flag}(p,2)
 =p\left(\frac34\log_2 3-1\right)>0\qquad(0<p\leq2/3).
 \label{eq:flag-improvement}
\end{equation}
Combining the results gives
\begin{equation}
 C(\cD_p)\leq\Cfb(\cD_p)\leq\UF(\cD_p)
 \leq U_{\rm flag}(p,2)<\Upsilon(\cD_p)
 \quad(0<p\leq2/3).
 \label{eq:dep-chain}
\end{equation}
At $p=2/3$, $C$, $\Cfb$, $\UF$, and $U_{\rm flag}$ coincide, while $\Upsilon$ remains larger. Beyond this point the flagged construction~\eqref{eq:dep-mixture} is not used; the curve is continued by the known exact entanglement-breaking capacity.

Figure~\ref{fig:dep-bounds} plots these expressions. The entries in Table~\ref{tab:dep} are numerical evaluations of~\eqref{eq:dep-UF}, not estimates of the unknown feedback capacity. Section~\ref{sec:numerics} describes the global numerical checks. The analytical inequality~\eqref{eq:dep-chain} establishes improvement over the specified published bound independently of the numerical precision.

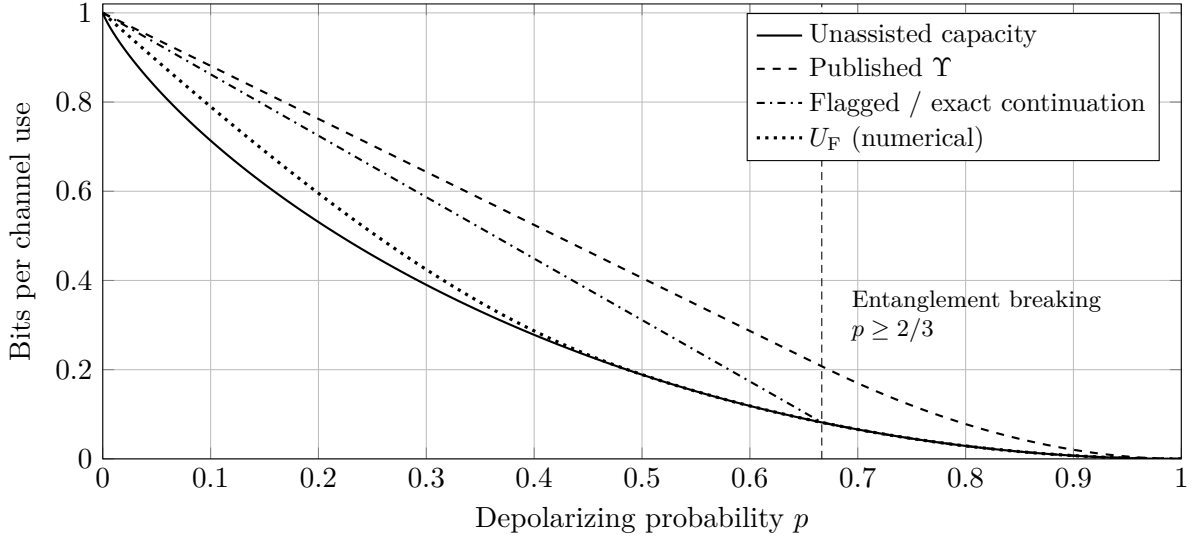
\begin{figure}[t]\centering\begin{tikzpicture}
\begin{axis}[width=.96\linewidth,height=7.6cm,xmin=0,xmax=1,ymin=0,ymax=1.02,
 xlabel={Depolarizing probability $p$},ylabel={Bits per channel use},grid=major,
 legend style={font=\small,at={(.98,.98)},anchor=north east},legend cell align=left]
\addplot[thick] table[x=p,y=C,col sep=comma]{data/depolarizing.csv};\addlegendentry{Unassisted capacity}
\addplot[thick,dashed] table[x=p,y=Upsilon,col sep=comma]{data/depolarizing.csv};\addlegendentry{Published $\Upsilon$}
\addplot[thick,dashdotted] table[x=p,y=Flagged,col sep=comma]{data/depolarizing.csv};\addlegendentry{Flagged / exact continuation}
\addplot[very thick,dotted] table[x=p,y=UF,col sep=comma]{data/depolarizing.csv};\addlegendentry{$\UF$ (numerical)}
\draw[densely dashed] (axis cs:.6666667,0)--(axis cs:.6666667,1.02);
\node[anchor=west,font=\footnotesize,align=left] at (axis cs:.685,.32) {Entanglement breaking\\$p\geq2/3$};
\end{axis}\end{tikzpicture}
\caption{Qubit depolarizing comparison in bits per channel use. The unassisted capacity is exact. The $\Upsilon$ curve is~\eqref{eq:dep-upsilon}; the flagged expression~\eqref{eq:qubit-flag} is continued beyond $p=2/3$ by the exact entanglement-breaking capacity, not by linear extrapolation. The $\UF$ curve is the numerically maximized expression~\eqref{eq:dep-UF}. The improvements concern ordinary capacity; an improved strong-converse threshold is not proved.}\label{fig:dep-bounds}\end{figure}

\begin{table}[t]\centering
\caption{Depolarizing bounds, in bits per channel use. The last column evaluates the channel bound, not an operational capacity.}\label{tab:dep}
\begin{tabular}{@{}rrrrr@{}}\toprule
$p$&$C(\cD_p)$&$\Upsilon(\cD_p)$&$U_{\rm flag}(p,2)$&$\UF(\cD_p)$\\\midrule
0.10&0.7136030429&0.8811278124&0.8622556249&0.7897376813\\
0.25&0.4564355568&0.7028195311&0.6556390622&0.5064360850\\
0.40&0.2780719051&0.5245112498&0.4490224996&0.2868844512\\
0.50&0.1887218755&0.4056390622&0.3112781245&0.1889430634\\\bottomrule
\end{tabular}\end{table}

\subsection{Why a small gap is not equality}\label{subsec:strict-UF}
Near the entanglement-breaking threshold, a maximizing input can be extremely close to pure. The graph on a linear scale then conceals the distinction between $\UF$ and the unassisted capacity. We prove that the gap is nevertheless positive throughout the open non-entanglement-breaking interval.

\begin{proposition}[Strictness of the entanglement-of-formation bound]\label{prop:strict-UF}
For every $0<p<2/3$,
\begin{equation}
 \UF(\cD_p)>1-h_2(p/2).
 \label{eq:strict-UF}
\end{equation}
\end{proposition}
\begin{proof}
Set $b\coloneqq p/2$, $t\coloneqq1-p$, and $\delta\coloneqq(1-r)/2$. For fixed $p$ in the stated interval, $b>0$ and $\kappa>0$. The objective in~\eqref{eq:dep-UF} is
\begin{equation}
 1-h_2(b+t\delta)+\mathcal E(2\kappa\sqrt{\delta(1-\delta)}).
 \label{eq:dep-delta-objective}
\end{equation}
The output-entropy increase from its value at $\delta=0$ is
\begin{equation}
 h_2(b+t\delta)-h_2(b)
 =t\log_2\!\left(\frac{1-b}{b}\right)\delta+O(\delta^2).
 \label{eq:dep-entropy-small}
\end{equation}
By symmetry of $h_2$, write the entanglement term as $h_2(z_\delta)$, where
\begin{equation}
 z_\delta\coloneqq\frac{1-\sqrt{1-4\kappa^2\delta(1-\delta)}}2
 =\kappa^2\delta+O(\delta^2).
 \label{eq:dep-small-eigen}
\end{equation}
Since $h_2(z)=z\log_2(1/z)+O(z)$ near zero, the entanglement term is
$\kappa^2\delta\log_2(1/\delta)+O(\delta)$. Subtracting $1-h_2(b)$ from~\eqref{eq:dep-delta-objective} therefore gives a strictly positive quantity for sufficiently small $\delta>0$: after division by $\delta$, the positive logarithmic term diverges and the other terms remain bounded.
\end{proof}

Figure~\ref{fig:dep-gaps} displays the gaps logarithmically. At $p=1/2$, $\UF-C\approx2.2118783\times10^{-4}$ bits per use. At $p=0.6$, high-precision evaluation gives approximately $2.7549\times10^{-15}$. The proof of positivity is Proposition~\ref{prop:strict-UF}, not subtraction of rounded plotted values.

The same argument identifies a limitation of modifying the potential by a constant weight. For every fixed $\eta>0$, the feedback-monotone quantity $I(\Msg;B' T)+\eta\EF(A';B'|\Msg T)$ has a one-use increase greater than $C(\cD_p)$ for a sufficiently weakly mixed Pauli-orbit ensemble. Its entanglement contribution is of order $\eta\kappa^2\delta\log_2(1/\delta)$, while the Holevo loss is only of order $\delta$. Thus such reweighting cannot establish a universal sharp per-use increment. This is an obstruction to that potential method, not a protocol demonstrating a feedback advantage.

\begin{figure}[t]\centering\begin{tikzpicture}
\begin{semilogyaxis}[width=.96\linewidth,height=7.5cm,xmin=.02,xmax=.6,ymin=1e-16,ymax=1,
 xlabel={Depolarizing probability $p$},ylabel={Upper bound minus unassisted capacity},grid=major,
 legend style={font=\small,at={(.03,.03)},anchor=south west},legend cell align=left,
 ytick={1,1e-4,1e-8,1e-12,1e-16}]
\addplot[thick,dashed] table[x=p,y=Upsilon_gap,col sep=comma]{data/depolarizing_gaps.csv};\addlegendentry{$\Upsilon-C$}
\addplot[thick,dashdotted] table[x=p,y=Flagged_gap,col sep=comma]{data/depolarizing_gaps.csv};\addlegendentry{$U_{\rm flag}-C$}
\addplot[very thick,dotted] table[x=p,y=UF_gap,col sep=comma]{data/depolarizing_gaps.csv};\addlegendentry{$\UF-C$ (numerical)}
\node[font=\footnotesize,align=left,anchor=north west] at (axis cs:.265,2e-5) {$p=1/2$: gap $\approx2.2119\times10^{-4}$};
\draw[->] (axis cs:.44,1e-5)--(axis cs:.5,2.2118783e-4);
\end{semilogyaxis}\end{tikzpicture}
\caption{Gaps above the unassisted capacity, on a logarithmic vertical scale. The $\UF-C$ gap is evaluated directly in a cancellation-resistant form using the smaller input eigenvalue $\delta$, rather than subtracting rounded capacities. The plotted interval stops at $p=0.6$ to avoid an unresolved near-threshold scale. The positive gap is proved analytically for every $0<p<2/3$.}\label{fig:dep-gaps}\end{figure}
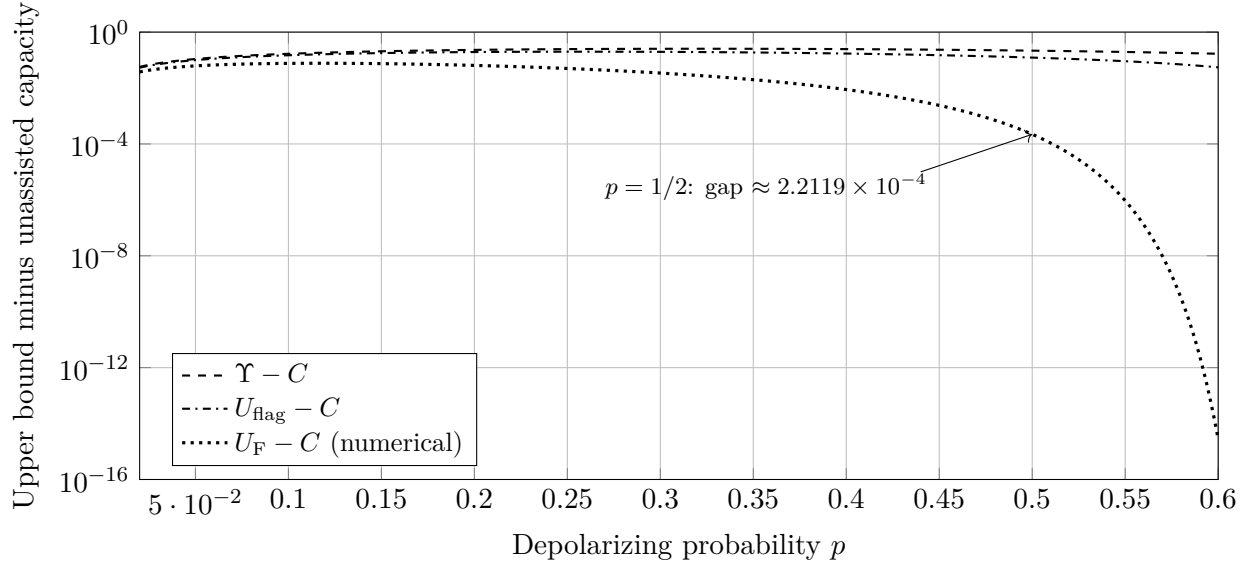

\section{An exact evaluation for amplitude damping}\label{sec:AD}
The depolarizing simplification relies on full unitary covariance. Amplitude damping does not have that symmetry, so an arbitrary ensemble could in principle contain different populations and different coherences. We now show that its $\UF$ optimization still has a closed form. The key is not to assume a diagonal ensemble, but to prove that discarding input coherences can only lower the relevant penalty and that the resulting scalar penalty is convex.

\subsection{A determinant representation}
For $0\leq\gamma\leq1$, the amplitude damping channel is defined as $\cA_\gamma(\rho)\coloneqq\sum_{j=0}^1K_j\rho K_j^\dagger$, with Kraus operators
\begin{equation}
 K_0\coloneqq\begin{pmatrix}1&0\\0&\sqrt{1-\gamma}\end{pmatrix},\qquad
 K_1\coloneqq\sqrt\gamma\ket0\bra1.
 \label{eq:AD-kraus}
\end{equation}
Set $a\coloneqq1-\gamma$. Parameterize an input by its excited-state population $t$ and coherence $z$. The corresponding state is defined as
\begin{equation}
 \rho(t,z)\coloneqq\begin{pmatrix}1-t&z\\z^*&t\end{pmatrix},\quad
 0\leq t\leq1,\quad |z|^2\leq t(1-t),\quad u\coloneqq\det\rho=t(1-t)-|z|^2.
 \label{eq:AD-input}
\end{equation}
The entropy function parametrized by the qubit determinant is defined as
\begin{equation}
 \ent(v)\coloneqq h_2\!\left(\frac{1-\sqrt{1-4v}}2\right),\qquad 0\leq v\leq\frac14.
 \label{eq:det-entropy}
\end{equation}
Its continuous endpoint values are $\ent(0)=0$ and $\ent(1/4)=1$.

The output and its determinant are
\begin{equation}
 \cA_\gamma(\rho(t,z))=
 \begin{pmatrix}1-at&\sqrt a\,z\\\sqrt a\,z^*&at\end{pmatrix},
 \qquad\det\cA_\gamma(\rho)=a(u+\gamma t^2).
 \label{eq:AD-output}
\end{equation}
The Choi-state concurrence of $\cA_\gamma$ is $\sqrt a$, so the concurrence factorization in Appendix~\ref{app:concurrence} gives
\begin{equation}
 f_{\cA_\gamma}(\rho)=\ent(a u).
 \label{eq:AD-formation}
\end{equation}
Consequently the penalty in~\eqref{eq:UF} is
\begin{equation}
 S(\cA_\gamma(\rho))-f_{\cA_\gamma}(\rho)
 =\ent\big(a(u+\gamma t^2)\big)-\ent(a u).
 \label{eq:AD-penalty-general}
\end{equation}
This scalar difference is closely related to the environment-correlation penalties used for generalized amplitude damping in~\cite{Fang26}. Here we evaluate the full zero-temperature $\UF$ optimization analytically.

\subsection{Removing coherences and proving convexity}
For fixed population $t$, coherence reduces $u$. We need the effect of this reduction on the difference in~\eqref{eq:AD-penalty-general}. The following lemma also supplies the convexity needed to average over different populations.

\begin{lemma}[Amplitude-damping penalty]\label{lem:AD-penalty}
Define
\begin{equation}
 g_\gamma(t)\coloneqq h_2(a t)-\ent\big(a t(1-t)\big),\qquad 0\leq t\leq1.
 \label{eq:AD-g}
\end{equation}
For every input in~\eqref{eq:AD-input},
\begin{equation}
 S(\cA_\gamma(\rho(t,z)))-f_{\cA_\gamma}(\rho(t,z))\geq g_\gamma(t).
 \label{eq:AD-diagonal-reduction}
\end{equation}
Moreover, $g_\gamma$ is convex on $[0,1]$.
\end{lemma}
\begin{proof}
First, $\ent$ is increasing and concave. For $0<v<1/4$, set $s\coloneqq\sqrt{1-4v}$. Differentiation and the integral representation of $\operatorname{atanh}$ give
\begin{equation}
 \ent'(v)=\frac{2\operatorname{atanh}s}{s\ln2}
 =\frac2{\ln2}\int_0^1\frac{dy}{1-(1-4v)y^2}.
 \label{eq:det-entropy-derivative}
\end{equation}
The integrand is positive and decreases with $v$, proving both properties, with extension by continuity. For a fixed increment $c\geq0$, the difference $\ent(v+c)-\ent(v)$ is therefore nonincreasing in $v$. In~\eqref{eq:AD-penalty-general}, $c=a\gamma t^2$ is fixed when $t$ is fixed. The largest permitted $u$ is $t(1-t)$, attained at $z=0$. Evaluating at that point proves~\eqref{eq:AD-diagonal-reduction}, because the diagonal output has entropy $h_2(at)$.

To prove convexity of $g_\gamma$, consider a mathematical Stinespring environment with
\begin{equation}
 V\ket0\coloneqq\ket{0}_B\ket{0}_E,\qquad
 V\ket1\coloneqq\sqrt a\ket{1}_B\ket{0}_E+\sqrt\gamma\ket{0}_B\ket{1}_E.
 \label{eq:AD-isometry}
\end{equation}
For the diagonal input $\rho_t\coloneqq\diag(1-t,t)$, measure $E$ in the fixed basis $\ket\pm\coloneqq(\ket0\pm\ket1)/\sqrt2$. Both outcomes have probability $1/2$, independently of $t$. The normalized conditional Bob states and their average are defined as
\begin{equation}
 \sigma_\pm(t)\coloneqq
 \begin{pmatrix}1-at&\pm\sqrt{a\gamma}\,t\\
 \pm\sqrt{a\gamma}\,t&at\end{pmatrix},\qquad
 \bar\sigma(t)\coloneqq\frac{\sigma_+(t)+\sigma_-(t)}2=\diag(1-at,at).
 \label{eq:AD-environment-states}
\end{equation}
Each state is affine in $t$, and $\det\sigma_\pm(t)=a t(1-t)$. Hence the measured classical--quantum mutual information is exactly
\begin{equation}
 g_\gamma(t)=S(\bar\sigma(t))-\frac12\sum_{y=\pm}S(\sigma_y(t))
 =\frac12\sum_{y=\pm}D(\sigma_y(t)\|\bar\sigma(t)).
 \label{eq:AD-g-relative-entropy}
\end{equation}
Joint convexity of relative entropy, together with the affine dependence of both arguments on $t$, proves convexity. All support endpoint cases follow by continuity. Figure~\ref{fig:AD-measurement} illustrates this fixed measurement; Figure~\ref{fig:AD-penalty} plots the convex functions.
\end{proof}

The fixed probabilities in~\eqref{eq:AD-environment-states} are important. Mutual information is not being claimed convex under an arbitrary variation of its joint state. The displayed relative-entropy representation is jointly affine in its two arguments with the same fixed weights, and this is what proves convexity. The environment measurement is not performed by a receiver in the communication protocol.

\begin{figure}[t]\centering\begin{tikzpicture}[font=\small,x=1cm,y=1cm]
\node (rho) at (0,0) {$\rho_t$};
\node[op,minimum width=1cm] (v) at (1.65,0) {$V$};
\node[op,minimum width=4.5cm,minimum height=1.8cm] (states) at (9.3,.1) {quantum output $B$ and outcome $Y$\\[5pt]
 $\sigma_Y(t)$,\quad $Y\in\{+,-\}$\\[4pt]
 $\Pr(Y=+)=\Pr(Y=-)=1/2$};
\node[op,minimum width=2.45cm,minimum height=1.1cm] (meas) at (4.55,-1.5) {measure $E$\\in $(\ket+,\ket-)$};
\draw[qwire] (rho)--(v);
\draw[qwire] (v.north east)--(2.7,1.25)--node[above]{$B$}(6.45,1.25)--(states.north west);
\draw[qwire] (v.south east)--(2.75,-1.5)--node[below]{$E$}(meas.west);
\draw[cwire] (meas.east)--(6.45,-1.5)--node[below right]{$Y$}(states.south west);
\node[align=center] at (6,-3.05) {Fixed probabilities and affine conditional states give\\[5pt]
$g_\gamma(t)=\frac12D(\sigma_+(t)\|\bar\sigma(t))+\frac12D(\sigma_-(t)\|\bar\sigma(t))$.};
\node[font=\footnotesize] at (6,-3.85) {This measurement is a mathematical construction, not a receiver resource.};
\end{tikzpicture}
\caption{A fixed environment measurement used only to prove a scalar convexity statement. For diagonal input population $t$, the outcomes $+$ and $-$ are equiprobable, and both conditional Bob states are affine in $t$. Their Holevo information is exactly $g_\gamma(t)$. Joint convexity of relative entropy proves convexity in $t$. The inaccessible environment is not made available to the communication parties.}\label{fig:AD-measurement}\end{figure}

\begin{figure}[t]\centering\begin{tikzpicture}
\begin{axis}[width=.92\linewidth,height=6.5cm,xmin=0,xmax=1,ymin=0,ymax=1.03,
 xlabel={Excited-state population $t$},ylabel={Penalty $g_\gamma(t)$ (bits)},grid=major,
 legend style={font=\small,at={(.03,.97)},anchor=north west},legend cell align=left]
\addplot[thick] table[x=t,y=gamma025,col sep=comma]{data/ad_penalty.csv};\addlegendentry{$\gamma=1/4$}
\addplot[thick,dashed] table[x=t,y=gamma050,col sep=comma]{data/ad_penalty.csv};\addlegendentry{$\gamma=1/2$}
\addplot[thick,dashdotted] table[x=t,y=gamma075,col sep=comma]{data/ad_penalty.csv};\addlegendentry{$\gamma=3/4$}
\end{axis}\end{tikzpicture}
\caption{The scalar amplitude-damping penalty $g_\gamma(t)$ for three damping probabilities. These plots illustrate, but do not establish, the convexity proved by the fixed-measurement representation~\eqref{eq:AD-g-relative-entropy}. Here $t$ is the excited-state population, not the damping probability.}\label{fig:AD-penalty}\end{figure}
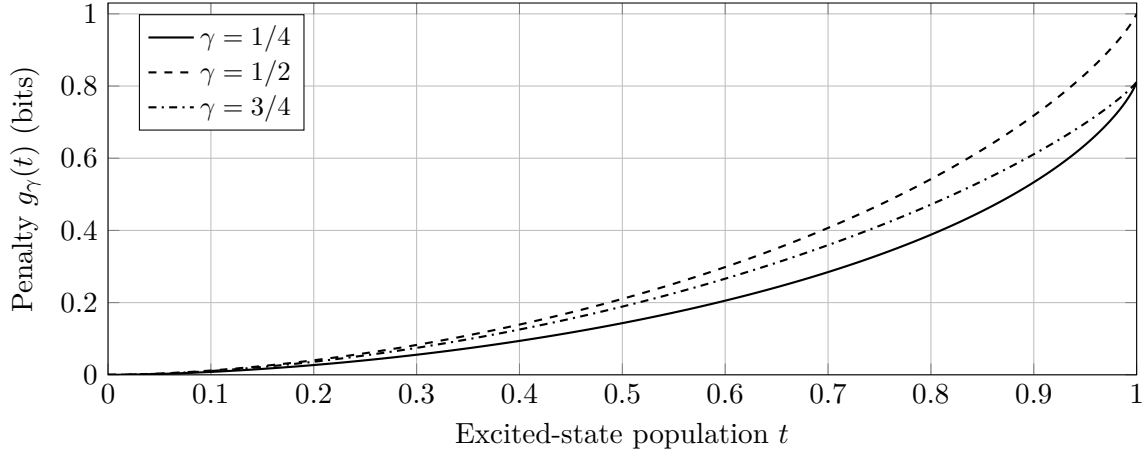
\FloatBarrier

\subsection{The exact channel-bound formula}
We can now handle arbitrary ensembles, not only phase-symmetric or diagonal ones. Dephasing bounds the entropy of the average output, while Lemma~\ref{lem:AD-penalty} bounds each subtracted penalty. Convexity then removes the entire ensemble optimization.

\begin{theorem}[Exact amplitude-damping evaluation]\label{thm:AD}
For every $0\leq\gamma\leq1$,
\begin{equation}
 \UF(\cA_\gamma)=\ent\!\left(\frac{1-\gamma}{4}\right)
 =h_2\!\left(\frac{1+\sqrt\gamma}{2}\right).
 \label{eq:AD-UF}
\end{equation}
In particular, this expression upper bounds $\Cfb(\cA_\gamma)$.
\end{theorem}
\begin{proof}
For an arbitrary input ensemble, write its members as $\rho(t_x,z_x)$ and define $\bar\rho\coloneqq\sum_xp_x\rho(t_x,z_x)$ and $\bar t\coloneqq\sum_xp_xt_x$. Dephasing the average output cannot decrease its entropy, so
$S(\cA_\gamma(\bar\rho))\leq h_2(a\bar t)$. Lemma~\ref{lem:AD-penalty} and Jensen's inequality give
\begin{align}
 S(\cA_\gamma(\bar\rho))-\sum_xp_x
 \left[S(\cA_\gamma(\rho^x))-f_{\cA_\gamma}(\rho^x)\right]
 &\leq h_2(a\bar t)-\sum_xp_xg_\gamma(t_x)\\
 &\leq h_2(a\bar t)-g_\gamma(\bar t)\\
 &=\ent(a\bar t(1-\bar t))\leq\ent(a/4).
 \label{eq:AD-UF-upper}
\end{align}
The last step uses $\bar t(1-\bar t)\leq1/4$ and monotonicity of $\ent$. For the reverse bound, use the one-element ensemble with $\rho\coloneqq I_2/2$. Its Holevo information is zero, but~\eqref{eq:AD-formation} gives $f_{\cA_\gamma}(I_2/2)=\ent(a/4)$. Thus it attains the objective in~\eqref{eq:UF}. The entropy symmetry in~\eqref{eq:det-entropy} gives the second equality in~\eqref{eq:AD-UF}. The endpoints also follow directly: $\gamma=0$ is noiseless and $\gamma=1$ is a replacer.
\end{proof}

This proof explains both the usefulness and the limitation of the bound. Its optimizer can consist of only one input state and create entanglement without sending any message information. Thus tightness of the channel-bound optimization does not imply that the same value is an achievable classical rate. The entropy in~\eqref{eq:AD-UF} also occurs as the Choi entanglement of formation and, in recent work, as the parallel entanglement cost of amplitude damping~\cite{TZBW26}. We have proved the feedback upper bound directly and do not rely on a simulation with free forward classical communication.

\subsection{Benchmarks and numerical comparison}
For comparison, two equiprobable pure inputs
$\ket{\psi_\pm(t)}\coloneqq\sqrt{1-t}\ket0\pm\sqrt t\ket1$ attain the single-use Holevo optimization. The additivity theorem of Tang, Zhu, Bai, and Wang~\cite[Theorem~5.5]{TZBW26} identifies its optimum with the unassisted classical capacity:
\begin{equation}
 C(\cA_\gamma)=\chi(\cA_\gamma)=\max_{0\leq t\leq1}
 \left\{h_2((1-\gamma)t)-\ent(\gamma(1-\gamma)t^2)\right\}.
 \label{eq:AD-Holevo}
\end{equation}
This is the established single-use optimization~\cite{GF05,KSW20}. For $0\leq\gamma<1$ its scalar objective is strictly concave: the first entropy is strictly concave in $t$, while the second equals $\mathcal E(2\sqrt{\gamma(1-\gamma)}\,t)$ and is convex in $t$. Indeed, with $s\coloneqq\sqrt{1-c^2}$,
\begin{equation}
 \mathcal E''(c)=\frac{\operatorname{atanh}s-s}{s^3\ln2}\geq0,
 \qquad 0<c<1,
 \label{eq:formation-convexity}
\end{equation}
with endpoint convexity obtained by continuity. Consequently the numerical maximization has a unique global maximizer; the completely damped case has the identically zero objective. This observation supplies an analytic check on the scalar optimization used in the plot. The binary-ensemble formula and achievability originate in~\cite{GF05,KSW20}; identifying this value with the regularized unassisted capacity uses the separately cited additivity theorem~\cite{TZBW26}. Figure~\ref{fig:AD-bounds} accordingly labels the lower curve as the unassisted classical capacity. Our feedback converse does not assume this additivity theorem.

The published semidefinite-programming bound is
\begin{equation}
 C_\beta(\cA_\gamma)=\log_2(1+\sqrt{1-\gamma}),
 \label{eq:AD-beta}
\end{equation}
with its evaluation given in~\cite{WXD18,KSW20}. The geometric-$\Upsilon$ comparison for ordinary amplitude damping in~\cite[Section~6.4]{FF21} was reported numerically to coincide with this expression. The feedback framework of~\cite{DKQSW23} also permits the $C_\beta$ strong-converse upper bound. We compare with the analytic $C_\beta$ curve rather than claim a separate exact evaluation of every optimized geometric divergence.

The improvement of~\eqref{eq:AD-UF} over~\eqref{eq:AD-beta} is strict for $0<\gamma<1$. Indeed, with $q\coloneqq(1+\sqrt\gamma)/2$, the first is the Shannon entropy of $(q,1-q)$, while the second is its R\'enyi entropy of order $1/2$:
\begin{equation}
 2\log_2(\sqrt q+\sqrt{1-q})
 =\log_2(1+2\sqrt{q(1-q)})
 =\log_2(1+\sqrt{1-\gamma}).
 \label{eq:AD-Renyi-comparison}
\end{equation}
Entropy monotonicity in the R\'enyi order is strict for a nonuniform distribution with two positive entries~\cite{Wilde17}.

The maximum-output-entropy bound of~\cite{DQSW19} equals one for $\gamma\leq1/2$ and $h_2(1-\gamma)$ for $\gamma\geq1/2$. To see this, a diagonal input produces excited-state probability $(1-\gamma)t$, while output coherence can only lower entropy at a fixed diagonal. The interval of permitted probabilities contains $1/2$ exactly when $1-\gamma\geq1/2$.

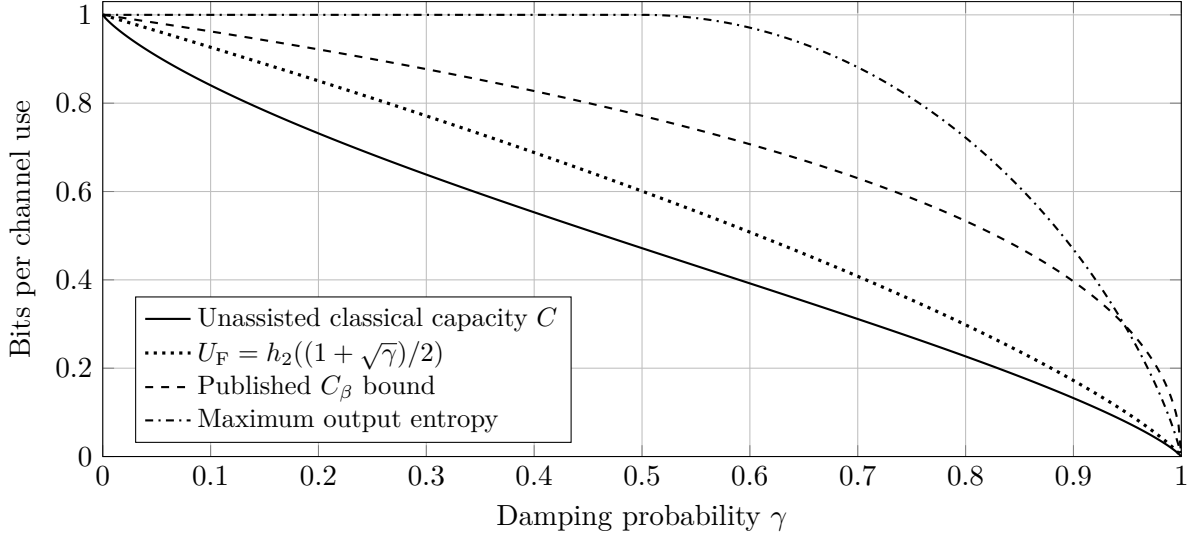
\begin{figure}[t]\centering\begin{tikzpicture}
\begin{axis}[width=.96\linewidth,height=7.6cm,xmin=0,xmax=1,ymin=0,ymax=1.03,
 xlabel={Damping probability $\gamma$},ylabel={Bits per channel use},grid=major,
 legend style={font=\small,at={(.03,.03)},anchor=south west},legend cell align=left]
\addplot[thick] table[x=gamma,y=Holevo,col sep=comma]{data/amplitude_damping.csv};\addlegendentry{Unassisted classical capacity $C$}
\addplot[very thick,dotted] table[x=gamma,y=UF,col sep=comma]{data/amplitude_damping.csv};\addlegendentry{$\UF=h_2((1+\sqrt\gamma)/2)$}
\addplot[thick,dashed] table[x=gamma,y=Cbeta,col sep=comma]{data/amplitude_damping.csv};\addlegendentry{Published $C_\beta$ bound}
\addplot[thick,dashdotted] table[x=gamma,y=Smax,col sep=comma]{data/amplitude_damping.csv};\addlegendentry{Maximum output entropy}
\end{axis}\end{tikzpicture}
\caption{Amplitude-damping comparison, in bits per channel use. The $\UF$ curve is the exact bound evaluation~\eqref{eq:AD-UF}, not a numerical ensemble approximation. The lower curve is the unassisted classical capacity~\eqref{eq:AD-Holevo}, using the additivity theorem of~\cite{TZBW26}. The other curves are the published $C_\beta$ and maximum-output-entropy upper bounds. The smaller $\UF$ threshold is proved here for ordinary feedback-assisted capacity only.}\label{fig:AD-bounds}\end{figure}

For example, at $\gamma=1/2$, the unassisted classical capacity is approximately $0.4717293906$, the bound $\UF$ is $0.6008760367$, and $C_\beta$ is $0.7715533032$. The gap between the achievable rate and $\UF$ is not evaluated as an operational feedback advantage.

\subsection{Why revealing a damping event is too generous}
The successful depolarizing flag construction does not mean that every natural environment flag produces a useful converse. Consider the Kraus-outcome extension
\begin{equation}
 \widehat{\cA}_\gamma(\rho)\coloneqq K_0\rho K_0^\dagger\otimes\proj0_F
                              +K_1\rho K_1^\dagger\otimes\proj1_F.
 \label{eq:AD-jump-flag}
\end{equation}
The basis input $\proj0$ produces $\proj0_B\otimes\proj0_F$, whereas $\proj1$ produces a mixture supported on $\ket{1}_B\ket0_F$ and $\ket0_B\ket1_F$. The two outputs have orthogonal supports. Thus this flagged channel can transmit one bit perfectly for every $\gamma$, including $\gamma=1$. Proposition~\ref{prop:basic} bounds its $\UF$ by the input dimension, so its unassisted and feedback-assisted capacities are both one bit. Discarding this flag simulates amplitude damping, but yields only a trivial upper bound.

This observation explains why the environment measurement in Lemma~\ref{lem:AD-penalty} must not be interpreted as free receiver access to the environment. Its role is to prove convexity of a penalty, not to enlarge the communication model until the noise becomes distinguishable and removable.

The channel analysis is now complete. We next turn to a bipartite source, where the cost of classical communication is subtracted from the randomness yield and the input-ensemble average is fixed by the source state. These distinctions will make the damping formulas useful in a different operational problem.

\section{Randomness distillation and a one-way state potential}\label{sec:randomness}
The preceding sections concern a dynamical resource: a channel that can be used repeatedly while classical feedback is free. There is a closely related static task in which Alice and Bob begin instead with many copies of a bipartite state and use classical communication to distill shared randomness. The communication is no longer free: its rate is subtracted from the amount of randomness produced. This difference makes the operational accounting more delicate, but the entanglement-of-formation idea still gives a useful one-way upper bound.

The purpose of this section is threefold. We first recall the one-way and two-way tasks in a form compatible with~\cite{LRWW23,NW24}. We then define a state analogue of $\UF$ and identify it with the same mixed-convex-roof quantity that appeared in Section~\ref{sec:structure}. Finally, we explain why its proven operational role is one-way. Examples based on Choi states of generalized amplitude damping channels will first demonstrate the usefulness of the bound. We then use isotropic states to compare the static and dynamical quantities and to show why an unrestricted two-way converse requires additional ideas.

\subsection{One-way and two-way distillable randomness}\label{subsec:randomness-task}
We first specify finite-copy protocols and their net yield, then take the same type of capacity limit as in Section~\ref{subsec:task}. The outputs must be nearly uniform and nearly identical, but need not be secret. Every classical message is included in the communication cost. In particular, no additional shared random seed is supplied for free in this source-distillation task.

For an integer $D\geq1$, the ideal shared-randomness state is defined as
\begin{equation}
 \overline\Phi^{D}_{K_AK_B}\coloneqq\frac1D\sum_{k=0}^{D-1}
 \proj k_{K_A}\otimes\proj k_{K_B}.
 \label{eq:ideal-randomness}
\end{equation}
It represents $\log_2D$ uniform shared bits. A one-way protocol on $\rho_{AB}^{\otimes n}$ consists of an Alice instrument with a classical output $K_A$ and a classical message $L$, followed by a Bob channel from $B^nL$ to a classical estimate $K_B$. Both final key registers have alphabet $\{0,\ldots,D-1\}$. Alice may retain quantum memory and use private randomness, and Bob may process his systems locally before or after receiving $L$. Without backward communication, all of Alice's processing can be performed first, and Bob's local processing can be absorbed into his final decoder. Several forward descriptions can be concatenated into $L$, with their logarithmic alphabet sizes added. Let $\omega^{\mathcal P}_{K_AK_B}$ denote the output of protocol $\mathcal P$, and let $d_L$ denote the message-alphabet size. See the top of Figure~\ref{fig:randomness-tasks} for a visual depiction.

For $\eps\in[0,1)$, a protocol is admissible when
\begin{equation}
 F\!\left(\omega^{\mathcal P}_{K_AK_B},\overline\Phi^D_{K_AK_B}\right)\geq1-\eps,
 \label{eq:randomness-finite-error}
\end{equation}
where the fidelity is defined in~\eqref{eq:fidelity-def}. Its net yield, in bits per source copy, is $(\log_2D-\log_2d_L)/n$. The $n$-shot one-way distillable randomness is defined as
\begin{equation}
 R_{\rightarrow}^{n,\eps}(\rho_{AB})\coloneqq
 \sup\left\{\frac{\log_2D-\log_2d_L}{n}:
 \mathcal P\text{ is an admissible one-way protocol on }\rho_{AB}^{\otimes n}\right\}.
 \label{eq:one-way-randomness-definition}
\end{equation}
The supremum includes all finite output and message alphabets. Taking $D=d_L=1$ shows that this quantity is nonnegative. A large message is not a free source of net randomness: its full logarithmic alphabet size is subtracted.

For a two-way protocol, let $L_1,\ldots,L_r$ be the successive classical messages, in either direction, and set $d_{L_i}\coloneqq|L_i|$. The number of rounds is finite for each protocol but need not be fixed as $n$ varies. Its $n$-shot net rate is defined as
\begin{equation}
 R_{\leftrightarrow}^{n,\eps}(\rho_{AB})\coloneqq
 \sup\left\{\frac{\log_2D-\sum_{i=1}^{r}\log_2d_{L_i}}{n}:
 \mathcal P\text{ is an admissible LOCC protocol on }\rho_{AB}^{\otimes n}\right\}.
 \label{eq:two-way-net-yield}
\end{equation}
Here LOCC means local operations and classical communication. Variable-length histories can be included in fixed finite message alphabets with padding; we use the alphabet-size cost, not an unannounced expected-length convention. The bottom of Figure~\ref{fig:randomness-tasks} depicts this two-way model.

Following the finite-resource convention of~\cite{LRWW23,NW24}, the asymptotic one-way and two-way quantities are defined as
\begin{align}
 \Rone(\rho_{AB})&\coloneqq\inf_{\eps\in(0,1)}\liminf_{n\to\infty}R_{\rightarrow}^{n,\eps}(\rho_{AB}),
 \label{eq:one-way-randomness-limit}\\
 \Rtwo(\rho_{AB})&\coloneqq\inf_{\eps\in(0,1)}\liminf_{n\to\infty}R_{\leftrightarrow}^{n,\eps}(\rho_{AB}).
 \label{eq:two-way-randomness-limit}
\end{align}
These limits describe vanishing-error protocols with a net rate eventually at least any prescribed number below the limit. To see the equivalence, fix a rate strictly below the displayed limit. For each integer $j\geq2$, the lower-limit definition supplies a threshold $n_j$ beyond which that rate, up to $1/j$, is feasible for every blocklength with error tolerance $1/j$. Choose the thresholds increasingly and use such protocols for $n_j\leq n<n_{j+1}$. This makes the error and rate tolerances tend to zero along all sufficiently large blocklengths, not only a subsequence. Conversely, a vanishing-error sequence is eventually admissible for every fixed $\eps>0$, so its lower limiting rate cannot exceed either defining expression. No secrecy from the public transcript or a purifying system is imposed.

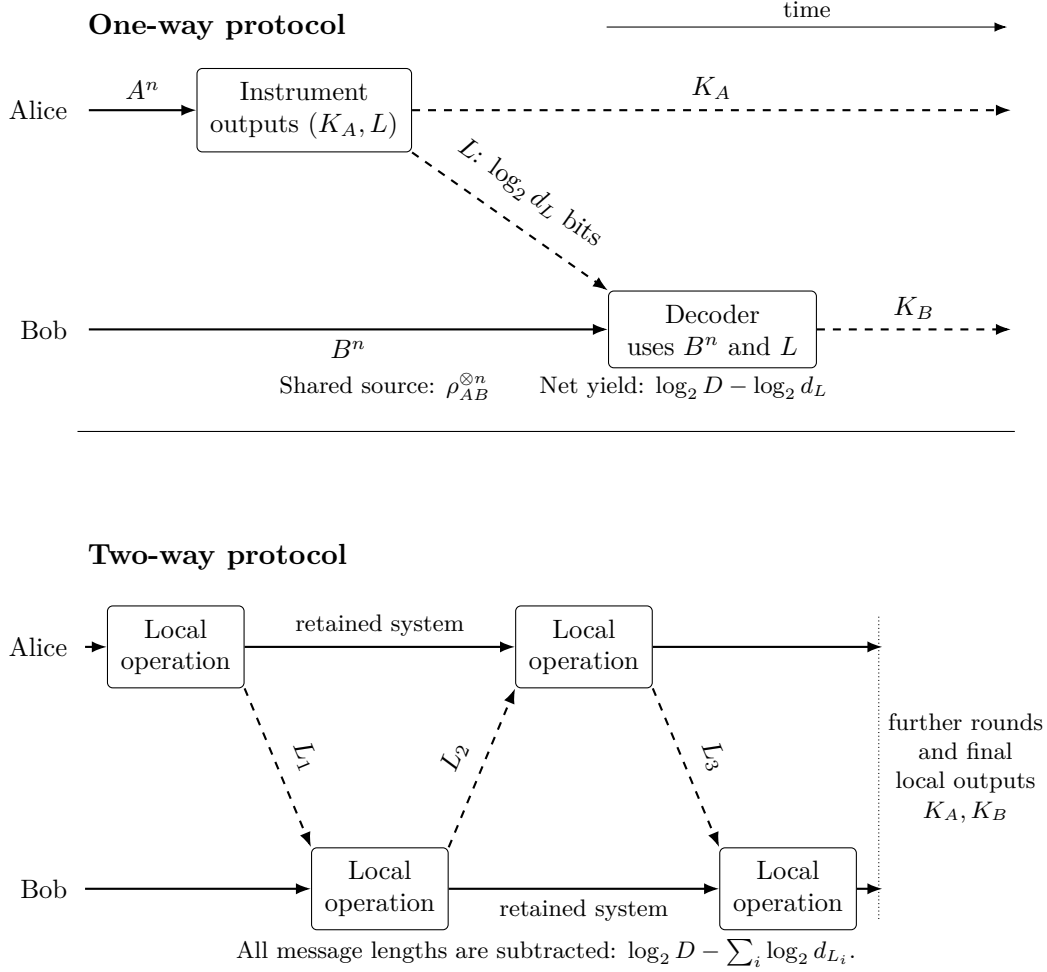
\begin{figure}[t]
 \centering\begin{tikzpicture}[x=1cm,y=1cm,font=\small]
\node[font=\bfseries,anchor=west] at (0,4.0) {One-way protocol};
\draw[-{Latex},thin] (7,4.0)--node[above,font=\footnotesize]{time}(12.3,4.0);
\node[anchor=east] at (0,2.9) {Alice};
\node[anchor=east] at (0,0) {Bob};
\node[op,minimum width=2.75cm] (enc) at (3.0,2.9) {Instrument\\outputs $(K_A,L)$};
\node[op,minimum width=2.75cm] (dec) at (8.4,0) {Decoder\\uses $B^n$ and $L$};
\draw[qwire] (.15,2.9)--node[above]{$A^n$}(enc.west);
\draw[qwire] (.15,0)--node[below]{$B^n$}(dec.west);
\draw[cwire] (enc.east)--node[above]{$K_A$}(12.35,2.9);
\draw[cwire] (enc.south east)--node[sloped,above=4pt]{$L$: $\log_2d_L$ bits}(dec.north west);
\draw[cwire] (dec.east)--node[above]{$K_B$}(12.35,0);
\node[font=\footnotesize,align=center] at (6.3,-.75) {Shared source: $\rho_{AB}^{\otimes n}$\qquad Net yield: $\log_2D-\log_2d_L$};
\draw[thin] (0,-1.35)--(12.4,-1.35);
\begin{scope}[yshift=-7.4cm]
\node[font=\bfseries,anchor=west] at (0,4.4) {Two-way protocol};
\node[anchor=east] at (0,3.2) {Alice};
\node[anchor=east] at (0,0) {Bob};
\node[op,minimum width=1.45cm] (a1) at (1.3,3.2) {Local\\operation};
\node[op,minimum width=1.45cm] (b1) at (4.0,0) {Local\\operation};
\node[op,minimum width=1.45cm] (a2) at (6.7,3.2) {Local\\operation};
\node[op,minimum width=1.45cm] (b2) at (9.4,0) {Local\\operation};
\draw[qwire] (.1,3.2)--(a1.west);
\draw[qwire] (a1.east)--node[above,font=\footnotesize]{retained system}(a2.west);
\draw[qwire] (a2.east)--(10.65,3.2);
\draw[qwire] (.1,0)--(b1.west);
\draw[qwire] (b1.east)--node[below,font=\footnotesize]{retained system}(b2.west);
\draw[qwire] (b2.east)--(10.65,0);
\draw[cwire] (a1.south east)--node[sloped,above=3pt]{$L_1$}(b1.north west);
\draw[cwire] (b1.north east)--node[sloped,above=3pt]{$L_2$}(a2.south west);
\draw[cwire] (a2.south east)--node[sloped,above=3pt]{$L_3$}(b2.north west);
\node[align=center,font=\footnotesize] at (11.75,1.6) {further rounds\\and final\\local outputs\\$K_A,K_B$};
\draw[densely dotted] (10.6,-.4)--(10.6,3.6);
\node[font=\footnotesize] at (6.2,-.85) {All message lengths are subtracted: $\log_2D-\sum_i\log_2d_{L_i}$.};
\end{scope}
\end{tikzpicture}
 \caption{Finite-copy randomness distillation. Time flows to the right. In the one-way model (top), Alice creates her final label $K_A$ and a description $L$, and Bob uses $L$ together with $B^n$ to produce $K_B$. The net yield subtracts $\log_2d_L$ from $\log_2D$. In the two-way model (bottom), local operations and messages alternate; the full cost $\sum_i\log_2d_{L_i}$ is subtracted. The retained systems may be quantum until the last local operations. The ideal final state is~\eqref{eq:ideal-randomness}, without a secrecy condition.}
 \label{fig:randomness-tasks}
\end{figure}
\FloatBarrier

Devetak and Winter proved that the one-way rate can be written as a regularized classical-correlation optimization~\cite{DW04}. An equivalent form follows from Koashi--Winter duality. If $\psi_{ABE}$ purifies $\rho_{AB}$, then
\begin{equation}
 \Rone(\rho_{AB})=S(B)_\rho-E_C(B;E)_\psi,
 \label{eq:one-way-EC}
\end{equation}
where the entanglement cost is characterized by the regularized entanglement of formation~\cite{HHT01} and is written as
\begin{equation}
 E_C(B;E)_\psi\coloneqq\lim_{n\to\infty}\frac1n
 E_F(B^n;E^n)_{\psi^{\otimes n}}.
 \label{eq:EC-definition}
\end{equation}
Equation~\eqref{eq:one-way-EC} is especially useful here because it puts the one-way randomness problem in direct contact with the mixed-convex-roof quantity already underlying $\UF$.

\subsection{A state analogue of the channel potential}
Let $\mathcal M_{A\to XA'}$ be an arbitrary quantum instrument written as a channel with a classical outcome register $X$ and a retained quantum system $A'$. Define
\begin{equation}
 \omega_{XA'B}\coloneqq(\mathcal M_{A\to XA'}\otimes\id_B)(\rho_{AB}).
 \label{eq:VF-output}
\end{equation}
The one-way state potential is defined as
\begin{equation}
 \VF(\rho_{AB})\coloneqq
 \sup_{\mathcal M_{A\to XA'}}
 \left\{I(X;B)_\omega+\EF(A';B|X)_\omega\right\}.
 \label{eq:VF-definition}
\end{equation}
The first term is classical correlation extracted immediately; the second records entanglement left in the conditional post-instrument states. This is the static counterpart of the two-term accounting in~\eqref{eq:potential}. The retained system is essential: forcing $A'$ to be trivial reduces the objective to the ordinary one-way measurement correlation.

To evaluate~\eqref{eq:VF-definition}, let $\psi_{ABE}$ purify $\rho_{AB}$ and use the environment-measured classical correlation $J(E;B)$ and its mixed convex roof $G_E$ from Section~\ref{sec:structure}. The Hughston--Jozsa--Wootters theorem~\cite{HJW93} identifies instruments on the purifying system $A$ with ensemble decompositions of $\rho_{BE}$. Keeping a local purification of each conditional component on $A'$ and applying Koashi--Winter duality gives the following identity.

\begin{proposition}[Mixed-roof representation]\label{prop:VF-roof}
For every bipartite state $\rho_{AB}$ with purification $\psi_{ABE}$,
\begin{equation}
 \VF(\rho_{AB})=S(B)_\rho-G_E(\rho_{BE}).
 \label{eq:VF-roof}
\end{equation}
\end{proposition}
\begin{proof}
We first justify a reduction that is important for the optimization. Consider an arbitrary instrument $\mathcal M_{A\to XA'}$. Write the operation for outcome $x$ as $\mathcal M_x(\xi)=\sum_jK_{xj}\xi K_{xj}^\dagger$. Replace it by the single operator $W_x\coloneqq\sum_jK_{xj}\otimes\ket j$, and let Alice retain the added register as part of $A'$. The operators satisfy $\sum_xW_x^\dagger W_x=I$, so they define an instrument. Individual $W_x$ need not be isometries: they include the outcome probabilities. This leaves the classical--quantum marginal on $XB$ unchanged, and hence leaves $I(X;B)$ unchanged. Tracing out the added local environment is a channel on Alice's side, so monotonicity of entanglement of formation under local channels implies that retaining it cannot decrease $E_F(A';B|X)$. Therefore the supremum in~\eqref{eq:VF-definition} can be restricted, without loss, to instruments for which the conditional state on $A'BE$ is pure for every $x$.

For such an instrument, Koashi--Winter duality gives
\begin{equation}
 \EF(A';B)_{\omega^x}=S(B)_{\omega^x}-J(E;B)_{\omega^x}.
 \label{eq:VF-KW}
\end{equation}
Adding this equality to
$I(X;B)_\omega=S(B)_\rho-\sum_xp_xS(B)_{\omega^x}$ yields
\begin{equation}
 I(X;B)_\omega+\sum_xp_x\EF(A';B)_{\omega^x}
 =S(B)_\rho-\sum_xp_xJ(E;B)_{\omega^x}.
 \label{eq:VF-cancel}
\end{equation}
By the Hughston--Jozsa--Wootters theorem~\cite{HJW93}, every ensemble decomposition $\rho_{BE}=\sum_xp_x\omega^x_{BE}$ can be generated by a suitable instrument on the purifying system $A$. Alice can retain a purification of each selected component in $A'$, so every decomposition appearing in the mixed roof is attainable in the restricted optimization above. Conversely, every such conditionally pure instrument produces an ensemble decomposition of $\rho_{BE}$ after $A'$ is discarded. Thus optimizing~\eqref{eq:VF-cancel} is exactly the minimization defining $G_E(\rho_{BE})$ in~\eqref{eq:mixed-roof}, proving~\eqref{eq:VF-roof}.
\end{proof}

\subsection{A finite-blocklength one-way converse}\label{subsec:finite-randomness}
The state potential in \eqref{eq:VF-definition} is useful at finite blocklength as well as asymptotically. We first isolate the measurement correlation available before Alice sends her classical message. We then construct a binary test from Bob's decoder. The test discounts overrepresented output labels, ensuring that the correction term depends on the approximation error but not on the amount of classical communication. This matters because the communication alphabet in~\eqref{eq:one-way-randomness-definition} has no imposed size bound.

For $n\geq1$, define the measurement-induced correlation
\begin{equation}
 \mathcal J_n(\rho_{AB})\coloneqq
 \sup_{\mathcal M_{A^n\to X}} I(X;B^n)_{(\mathcal M\otimes\id)(\rho^{\otimes n})},
 \label{eq:measurement-correlation-n}
\end{equation}
where $X$ is classical. This auxiliary quantity is not itself the $n$-shot, finite-error operational rate. Koashi--Winter duality~\cite{KW04} gives the finite-copy identity
\begin{equation}
 \mathcal J_n(\rho_{AB})
 =nS(B)_\rho-E_F(B^n;E^n)_{\psi^{\otimes n}},
 \label{eq:measurement-correlation-KW}
\end{equation}
where $\psi_{ABE}$ purifies $\rho_{AB}$. Indeed, refining a POVM into rank-one effects cannot decrease its outcome--receiver mutual information. Such a POVM prepares a pure-state ensemble of $\rho_{BE}^{\otimes n}$, and every pure-state ensemble arises in this way from the purifying system $A^n$. The average receiver state is fixed. Maximizing its entropy minus the average conditional entropy is therefore exactly the minimization defining the entanglement of formation on the right of~\eqref{eq:measurement-correlation-KW}.

The formation inequality of~\cite{YHHS05}, in the directional form of~\cite[Eq.~(17)]{BEHY11}, states that
\begin{equation}
 E_F(B_1B_2;E_1E_2)_\sigma
 \geq E_F(B_1;E_1)_\sigma+G_{E_2}(\sigma_{B_2E_2})
 \label{eq:formation-marginal-bound}
\end{equation}
for the indicated marginals of an arbitrary four-partite state. Appendix~\ref{app:formation-marginal} gives a self-contained proof using rank-one measurements, entropy data processing, and convexity of entanglement of formation. In particular, the proof does not assume strong superadditivity of entanglement of formation. For $n\geq2$ identical pairs, successive applications give
\begin{equation}
 \begin{aligned}
 E_F(B^n;E^n)_{\psi^{\otimes n}}
 &\geq E_F(B^{n-1};E^{n-1})_{\psi^{\otimes(n-1)}}+G_E(\rho_{BE})\\
 &\geq E_F(B;E)_\psi+(n-1)G_E(\rho_{BE})\\
 &\geq nG_E(\rho_{BE}).
 \end{aligned}
 \label{eq:G-lower-EC}
\end{equation}
For the last inequality, the pure ensembles defining $E_F$ are permissible ensembles in the mixed roof defining $G_E$, and $J(E;B)$ equals the marginal entropy on a pure state. For $n=1$, this last inequality alone proves the desired bound. Equations~\eqref{eq:VF-roof}, \eqref{eq:measurement-correlation-KW}, and~\eqref{eq:G-lower-EC} imply the finite-copy bound
\begin{equation}
 \mathcal J_n(\rho_{AB})\leq n\VF(\rho_{AB}).
 \label{eq:measurement-correlation-VF}
\end{equation}
Thus no regularization or limiting approximation has yet been used.

For convenience, define a nondecreasing binary-entropy bound
\begin{equation}
 \overline h_2(\delta)\coloneqq\max_{0\leq t\leq\delta}h_2(t)
 =h_2(\min\{\delta,1/2\}),\qquad 0\leq\delta\leq1.
 \label{eq:binary-entropy-envelope}
\end{equation}
It equals $h_2(\delta)$ for $\delta\leq1/2$ and equals one thereafter. This convention makes the following statement valid at every error tolerance, without treating binary entropy as monotone beyond $1/2$.

\begin{theorem}[Finite-blocklength one-way randomness bound]\label{thm:VF-oneway}
Let a one-way protocol on $\rho_{AB}^{\otimes n}$ use a classical message of alphabet size $d_L$ and produce two classical registers of size $D$, with fidelity at least $1-\eps$ to~\eqref{eq:ideal-randomness}, where $\eps\in[0,1)$. Then
\begin{equation}
 (1-\sqrt\eps)(\log_2D-\log_2d_L)
 \leq\mathcal J_n(\rho_{AB})+\overline h_2(\sqrt\eps)
 \leq n\VF(\rho_{AB})+\overline h_2(\sqrt\eps).
 \label{eq:finite-randomness-converse}
\end{equation}
Consequently,
\begin{equation}
 R_{\rightarrow}^{n,\eps}(\rho_{AB})
 \leq\frac{\VF(\rho_{AB})+\overline h_2(\sqrt\eps)/n}{1-\sqrt\eps},
 \qquad
 \Rone(\rho_{AB})\leq\VF(\rho_{AB}).
 \label{eq:VF-oneway-bound}
\end{equation}
There is no restriction on the growth of $D$ or $d_L$ with $n$.
\end{theorem}
\begin{proof}
\par\medskip\noindent\textbf{The state before the classical message is transmitted.}
Use the one-way normal form of Section~\ref{subsec:randomness-task}, depicted in the upper part of Figure~\ref{fig:randomness-tasks}, and temporarily denote Alice's final key register by $K$. Her instrument has joint classical output $(K,L)$. After discarding her remaining local memory, the state is
\begin{equation}
 \tau_{KLB^n}\coloneqq\sum_{k=0}^{D-1}\sum_{\ell=1}^{d_L}
 p_{k\ell}\proj{k,\ell}_{KL}\otimes\tau^{k\ell}_{B^n},
 \qquad \tau_{B^n}=\rho_B^{\otimes n}.
 \label{eq:randomness-pre-message-state}
\end{equation}
Zero-probability terms can be assigned arbitrary conditional states. Alice's local channel leaves Bob's unconditional state unchanged. The joint label $(K,L)$ is a measurement outcome on $A^n$, so
\begin{equation}
 I(KL;B^n)_\tau\leq\mathcal J_n(\rho_{AB}).
 \label{eq:pre-message-information-bound}
\end{equation}
For each received value $\ell$, Bob's decoder is represented by a POVM $(\Lambda_{j|\ell})_{j=0}^{D-1}$ on $B^n$. Define the final classical joint distribution and Alice's marginal as
\begin{equation}
 P_{kj}\coloneqq\sum_\ell p_{k\ell}\Tr[\Lambda_{j|\ell}\tau^{k\ell}_{B^n}],
 \qquad
 p_k\coloneqq\sum_\ell p_{k\ell}=\sum_jP_{kj}.
 \label{eq:randomness-output-distribution}
\end{equation}
Thus $\omega_{K_AK_B}=\sum_{k,j}P_{kj}\proj k\otimes\proj j$ is the actual output. Its fidelity to the ideal state is
\begin{equation}
 F\!\left(\omega,\overline\Phi^D\right)
 =\left(\frac1{\sqrt D}\sum_k\sqrt{P_{kk}}\right)^2\geq1-\eps.
 \label{eq:classical-randomness-fidelity}
\end{equation}
This condition controls both agreement and uniformity; agreement alone would allow the parties to output a constant.

\par\medskip\noindent\textbf{A decoder test with capped label probabilities.}
Alice's actual output need not be exactly uniform. Define
\begin{equation}
 \lambda_k\coloneqq
 \begin{cases}
 \min\{1,(Dp_k)^{-1}\},&p_k>0,\\
 1,&p_k=0,
 \end{cases}
 \qquad
 T_{KLB^n}\coloneqq\sum_{k,\ell}\lambda_k\proj{k,\ell}_{KL}\otimes\Lambda_{k|\ell}.
 \label{eq:weighted-decoder-test}
\end{equation}
Every $\lambda_k$ lies in $(0,1]$, and $0\leq\Lambda_{k|\ell}\leq I$, so $0\leq T\leq I$. The factor $\lambda_k$ decreases the acceptance associated with labels whose probability exceeds $1/D$; it is one for the other labels. This is a mathematical binary test, not a modification of the communication protocol.

Let its acceptance probabilities on the correlated and independent hypotheses be defined as
\begin{equation}
 a\coloneqq\Tr[T\tau_{KLB^n}]=\sum_k\lambda_kP_{kk},
 \qquad
 b\coloneqq\Tr[T(\tau_{KL}\otimes\rho_B^{\otimes n})].
 \label{eq:weighted-test-probabilities}
\end{equation}
Since $p_{k\ell}\leq p_k$ and $p_k\lambda_k\leq1/D$,
\begin{equation}
 \begin{aligned}
 b&=\sum_{k,\ell}p_{k\ell}\lambda_k\Tr[\Lambda_{k|\ell}\rho_B^{\otimes n}]\\
 &\leq\frac1D\sum_{\ell=1}^{d_L}\Tr\!\left[\left(\sum_k\Lambda_{k|\ell}\right)\rho_B^{\otimes n}\right]
 =\frac{d_L}{D}.
 \end{aligned}
 \label{eq:weighted-test-independent}
\end{equation}
Each possible classical description contributes only $1/D$ to this bound. The independent acceptance therefore depends on the net output size $D/d_L$, rather than the gross size $D$.

\par\medskip\noindent\textbf{Controlling the loss from the weighting.}
Set $\delta\coloneqq\sqrt\eps$ and define the total variation distance of Alice's marginal from uniform as
\begin{equation}
 \Delta\coloneqq\frac12\sum_k|p_k-1/D|.
 \label{eq:key-marginal-distance}
\end{equation}
Marginalizing the output and applying~\eqref{eq:fidelity-trace-distance} gives $\Delta\leq\delta$. The definition of the weights also gives
\begin{equation}
 \sum_k\frac1{D\lambda_k}
 =\sum_k\max\{p_k,1/D\}
 =1+\Delta.
 \label{eq:weight-normalization}
\end{equation}
For the last equality, use $\max\{x,y\}=(x+y+|x-y|)/2$ and sum the two normalized probability distributions. Cauchy--Schwarz now yields
\begin{equation}
 \begin{aligned}
 F\!\left(\omega,\overline\Phi^D\right)
 &=\left(\sum_k\sqrt{\lambda_kP_{kk}}\sqrt{\frac1{D\lambda_k}}\right)^2\\
 &\leq\left(\sum_k\lambda_kP_{kk}\right)
        \left(\sum_k\frac1{D\lambda_k}\right)
 =a(1+\Delta).
 \end{aligned}
 \label{eq:weighted-fidelity-Cauchy}
\end{equation}
Combining this with~\eqref{eq:classical-randomness-fidelity} gives
\begin{equation}
 a\geq\frac{1-\eps}{1+\Delta}
 \geq\frac{1-\eps}{1+\sqrt\eps}=1-\delta.
 \label{eq:weighted-test-success}
\end{equation}
Thus capping the overrepresented labels does not destroy the high acceptance needed for a converse.

\par\medskip\noindent\textbf{From the binary test to the net yield.}
If $D\leq d_L$, the first inequality in~\eqref{eq:finite-randomness-converse} is immediate, since its left side is nonpositive and its right side is nonnegative. Hence assume $D>d_L$. In this case~\eqref{eq:weighted-test-independent} implies $b<1$. Data processing of quantum relative entropy under the binary measurement $(T,I-T)$ gives
\begin{equation}
 \begin{aligned}
 I(KL;B^n)_\tau
 &=D(\tau_{KLB^n}\|\tau_{KL}\otimes\rho_B^{\otimes n})\\
 &\geq a\log_2\frac ab+(1-a)\log_2\frac{1-a}{1-b}\\
 &=-h_2(a)-a\log_2 b-(1-a)\log_2(1-b)\\
 &\geq a\log_2(D/d_L)-h_2(a)\\
 &\geq(1-\delta)(\log_2D-\log_2d_L)-\overline h_2(\delta).
 \end{aligned}
 \label{eq:randomness-binary-relative-entropy}
\end{equation}
The penultimate inequality uses~\eqref{eq:weighted-test-independent} and $-\log_2(1-b)\geq0$. For the last, use~\eqref{eq:weighted-test-success}, positivity of $\log_2(D/d_L)$, and $h_2(a)=h_2(1-a)\leq\overline h_2(\delta)$. If $b=0<a$, the relative entropy is infinite and the same argument is read in the extended-value sense; the mutual information of our finite-dimensional state is finite, so that case cannot actually occur here.

Equations~\eqref{eq:pre-message-information-bound} and~\eqref{eq:measurement-correlation-VF} now prove~\eqref{eq:finite-randomness-converse}. Divide by $n(1-\sqrt\eps)$ and optimize over protocols to obtain the finite-error bound in~\eqref{eq:VF-oneway-bound}. Taking the lower limit in $n$ and then the infimum over $\eps\in(0,1)$ proves its asymptotic consequence according to~\eqref{eq:one-way-randomness-limit}.
\end{proof}

The proof controls the net yield directly. An entropy-continuity estimate applied only to the gross key size would instead produce an error correction proportional to $\log_2D$, which cannot be dropped when the subtracted communication is arbitrarily large. The capped test avoids that problem. For $\eps\leq1/4$, its additive correction is simply $h_2(\sqrt\eps)$. The factor $1-\sqrt\eps$ makes this a weak converse, not a strong-converse threshold at $\VF$. It also does not make $\VF$ a two-way upper bound: the single Alice outcome $(K,L)$ before communication, used in~\eqref{eq:pre-message-information-bound}, is specific to the one-way model.

If Alice's final key is exactly uniform, then $\Delta=0$ and all weights in~\eqref{eq:weighted-decoder-test} are one. The same argument yields the stronger factor $1-\eps$ and correction $\overline h_2(\eps)$, instead of $1-\sqrt\eps$ and $\overline h_2(\sqrt\eps)$. The theorem is stated without that exact-uniformity assumption.

The bounds of~\cite{LRWW23,NW24} additionally cover arbitrary LOCC rounds and provide strong-converse and smooth-divergence statements. Theorem~\ref{thm:VF-oneway} supplies a different finite-blocklength one-way converse. The damping examples below show a strict improvement in its vanishing-error threshold over the earlier computable state restrictions; the isotropic example later explains why that threshold cannot be extended to unrestricted two-way extraction.

\subsection{Choi sources and the fixed average input}\label{subsec:Choi-fixed}
The channel and state potentials have closely related objectives, but they need not have the same value. The difference is an average-input constraint. Let $A_0$ be a copy of $A$, let
\begin{equation}
 \ket{\Phi}_{AA_0}\coloneqq\frac1{\sqrt{d_A}}\sum_{j=0}^{d_A-1}\ket j_A\ket j_{A_0},
 \qquad \pi_{A_0}\coloneqq I_{A_0}/d_A,
\end{equation}
and define the normalized Choi state of $\N\colon A_0\to B$ as
$\theta^{\N}_{AB}\coloneqq(\id_A\otimes\N)(\Phi_{AA_0})$.
The reference $A$ is the sender's share in the one-way distillation task. For an input state $\tau$, define its conditional penalty as
\begin{equation}
 g_{\N}(\tau)\coloneqq S(\N(\tau))-f_{\N}(\tau).
 \label{eq:general-channel-penalty}
\end{equation}

\begin{proposition}[Fixed-average Choi correspondence]\label{prop:Choi-fixed}
For every finite-dimensional channel,
\begin{equation}
 \VF(\theta^{\N}_{AB})
 =S(\N(\pi))-
 \inf_{(p_x,\tau^x)_x}
       \left\{\sum_xp_xg_{\N}(\tau^x): \sum_xp_x\tau^x=\pi\right\}
 \leq\UF(\N).
 \label{eq:Choi-fixed-potential}
\end{equation}
The input states $\tau^x$ may be mixed. The average constraint cannot in general be dropped.
\end{proposition}
\begin{proof}
As in Proposition~\ref{prop:VF-roof}, Alice may retain every local dilation system. Her instrument can then be represented by single operators $(W_x)_x$ satisfying $\sum_xW_x^\dagger W_x=I_A$. Apply it before $\N$, which acts on the disjoint system $A_0$. The Choi identity gives the unnormalized transmitted marginal
\begin{equation}
 \Tr_{A'}[(W_x\otimes I)\Phi(W_x^\dagger\otimes I)]
 =\frac{(W_x^\dagger W_x)^{\mathsf T}}{d_A}.
\end{equation}
For an outcome of positive probability, define
\begin{equation}
 p_x\coloneqq\frac{\Tr[W_x^\dagger W_x]}{d_A},\qquad
 \tau^x\coloneqq\frac{(W_x^\dagger W_x)^{\mathsf T}}{d_Ap_x}.
 \label{eq:Choi-instrument-ensemble}
\end{equation}
Completeness gives $\sum_xp_x\tau^x=\pi$. Before the channel, the normalized conditional state on $A'A_0$ is pure and purifies $\tau^x$. After the channel, its retained entanglement is therefore $f_\N(\tau^x)$, while Bob's conditional state is $\N(\tau^x)$. Substituting these states into~\eqref{eq:VF-definition} gives the objective on the right of~\eqref{eq:Choi-fixed-potential}.

Conversely, an arbitrary ensemble with average $\pi$ can be realized using
\begin{equation}
 W_x\coloneqq\sqrt{d_Ap_x(\tau^x)^{\mathsf T}}.
\end{equation}
Its completeness follows from $\sum_xd_Ap_x(\tau^x)^{\mathsf T}=I_A$. The preceding calculation recovers precisely that ensemble, with Alice retaining its conditional purification. This proves both directions of the optimization identity. The final inequality follows because~\eqref{eq:UF} optimizes over all average inputs as well.
\end{proof}

This constraint reflects a physical distinction: Alice may steer the transmitted subsystem of a Choi state, but cannot change its unconditional marginal by a local instrument. The following damping examples retain this constraint explicitly. They contrast an exactly solved benchmark with a full-rank source for which the bound supplies nontrivial information about an unevaluated regularized rate.

\section{One-way randomness from damping Choi states}\label{sec:damping-states}
The general one-way converse from Theorem~\ref{thm:VF-oneway} becomes informative only when its value can be controlled more sharply than other available estimates. We consider two related families to distinguish this issue from exact evaluation of an operational rate. Ordinary amplitude damping gives a closed-form state potential and, using a recent additivity theorem, an exact rate against which its slack can be measured. Generalized amplitude damping at nonzero temperature gives a full-rank state for which we do not evaluate the regularized rate exactly. On that state, the potential produces a certified upper bound smaller than the earlier computable state-$\Upsilon$ restrictions. Both examples use the average-input constraint in Proposition~\ref{prop:Choi-fixed}.

\subsection{Ordinary amplitude damping as an exact benchmark}\label{subsec:AD-state}
Let $\cA_\gamma$ be the amplitude-damping channel defined in~\eqref{eq:AD-kraus}, with $0\leq\gamma\leq1$. Its normalized Choi state is defined as
\begin{equation}
 \theta^\gamma_{AB}\coloneqq(\id_A\otimes\cA_{\gamma,A_0\to B})(\Phi_{AA_0})
 =\frac12\begin{pmatrix}
 1&0&0&\sqrt{1-\gamma}\\
 0&0&0&0\\
 0&0&\gamma&0\\
 \sqrt{1-\gamma}&0&0&1-\gamma
 \end{pmatrix},
 \label{eq:AD-source}
\end{equation}
where the matrix basis is $00,01,10,11$ and communication is from $A$ to $B$. The marginals are
\begin{equation}
 \theta_A^\gamma=I_A/2,\qquad
 \theta_B^\gamma=\diag((1+\gamma)/2,(1-\gamma)/2).
 \label{eq:AD-source-marginals}
\end{equation}

\begin{proposition}[Amplitude-damping source potential]\label{prop:AD-state-VF}
The state potential is given by
\begin{equation}
 \VF(\theta^\gamma)=\EF(A;B)_{\theta^\gamma}
 =h_2\!\left(\frac{1+\sqrt\gamma}{2}\right).
 \label{eq:AD-state-VF}
\end{equation}
In particular, this expression upper-bounds the one-way distillable randomness. The potential is attained by the identity instrument, with a constant classical outcome.
\end{proposition}
\begin{proof}
Proposition~\ref{prop:Choi-fixed} gives
$\VF(\theta^\gamma)\leq\UF(\cA_\gamma)$.
Theorem~\ref{thm:AD} evaluates this channel bound as the last expression in~\eqref{eq:AD-state-VF}, and its optimizer is the one-element ensemble with input $I/2$. Thus the optimizer already has the average required by the Choi source.

Equivalently, for a matching state-side lower bound, Alice can keep her entire share and produce a constant outcome $X$. This gives $I(X;B)=0$ and $\EF(A';B|X)=\EF(A;B)_{\theta^\gamma}$. The concurrence of~\eqref{eq:AD-source} is $\sqrt{1-\gamma}$, as calculated in Appendix~\ref{app:concurrence}. Wootters' formula~\cite{Wootters98} gives the last expression in~\eqref{eq:AD-state-VF}. The upper and lower bounds agree. Apply Theorem~\ref{thm:VF-oneway} for the operational conclusion.
\end{proof}

For this family the bound is strictly smaller than the receiver entropy whenever $0<\gamma<1$. Indeed, the binary entropy is increasing on $[0,1/2]$, and
\begin{equation}
 \frac{1-\sqrt\gamma}{2}<\frac{1-\gamma}{2}<\frac12.
\end{equation}
Using its symmetry about $1/2$ gives
\begin{equation}
 \VF(\theta^\gamma)
 < S(B)_{\theta^\gamma}
 =h_2((1-\gamma)/2).
 \label{eq:AD-state-entropy-improvement}
\end{equation}
Nevertheless, equality of the potential and $\EF$ is special to this family. A separable, perfectly classically correlated state has positive distillable randomness but zero entanglement of formation, so $\EF$ alone is not a general randomness converse.

\subsection{Exact operational rates and their attribution}\label{subsec:AD-exact-rate}
The source potential should not be confused with the exact rate. A simple measurement already gives a smaller achievable value, and a separately established additivity result proves that value optimal. Define
\begin{equation}
 L_{\rightarrow}(\gamma)\coloneqq
 h_2\!\left(\frac{1-\gamma}{2}\right)
 -h_2\!\left(\frac{1+\sqrt{1-\gamma(1-\gamma)}}2\right).
 \label{eq:AD-state-lower}
\end{equation}
Measuring $A$ in the basis $\ket\pm\coloneqq(\ket0\pm\ket1)/\sqrt2$ gives equiprobable outcomes and Bob states
\begin{equation}
 \omega_B^\pm\coloneqq\frac12
 \begin{pmatrix}1+\gamma&\pm\sqrt{1-\gamma}\\
                 \pm\sqrt{1-\gamma}&1-\gamma\end{pmatrix}.
 \label{eq:AD-state-X-output}
\end{equation}
Their average is $\theta_B^\gamma$, and each has determinant $\gamma(1-\gamma)/4$. Their Holevo information is therefore~\eqref{eq:AD-state-lower}. One-way reconciliation with quantum side information achieves this net rate~\cite{DW03,DW04}.

To prove optimality, use the purification
\begin{equation}
 \ket{\psi^\gamma}_{ABE}\coloneqq\frac{
 \ket{000}+\sqrt{1-\gamma}\ket{110}+\sqrt\gamma\ket{101}}{\sqrt2}.
 \label{eq:AD-state-purification}
\end{equation}
Its complementary marginal is
\begin{equation}
 \psi^\gamma_{BE}=\frac12\proj{00}+\frac12\proj{v_\gamma},\qquad
 \ket{v_\gamma}\coloneqq\sqrt{1-\gamma}\ket{10}+\sqrt\gamma\ket{01}.
 \label{eq:AD-state-complement}
\end{equation}
It has no support on $\ket{11}$. Theorem~4.1 of Tang, Zhu, Bai, and Wang~\cite{TZBW26} implies $E_C=E_F$ for this support class. Its concurrence is $\sqrt{\gamma(1-\gamma)}$: the $01,10$ coherence contributes $2(\sqrt{\gamma(1-\gamma)}/2)$ in the two-qubit concurrence formula, and the $11$ diagonal entry is zero. Hence Wootters' formula and~\eqref{eq:one-way-EC} give
\begin{equation}
 \Rone(\theta^\gamma)
 =S(B)_{\theta^\gamma}-E_C(B;E)_{\psi^\gamma}
 =L_{\rightarrow}(\gamma).
 \label{eq:AD-state-exact-rate}
\end{equation}
Thus~\eqref{eq:AD-state-exact-rate} is a consequence of~\cite{KW04,TZBW26}, not a new conclusion obtained by saturating our converse. The achievability argument above does not require that additivity theorem.

The state is asymmetric. For communication from $B$ to $A$, the same argument uses $\psi^\gamma_{AE}$, which is the Choi state of $\cA_{1-\gamma}$ and again has a product vector in its kernel. Its concurrence is $\sqrt\gamma$. Therefore
\begin{equation}
 R_{\leftarrow}(\theta^\gamma)
 =1-h_2\!\left(\frac{1+\sqrt{1-\gamma}}2\right).
 \label{eq:AD-reverse-rate}
\end{equation}
An $X$-basis measurement on $B$ achieves this value directly. The two directions need not have equal rates. At $\gamma=1/2$, the forward and reverse rates are approximately $0.4566992218$ and $0.3991239633$, respectively, whereas the forward potential is $0.6008760367$ bits per copy. Figure~\ref{fig:AD-state-rates} separates these operational rates from the potential and the elementary entropy bound.

\begin{figure}[t]
 \centering\begin{tikzpicture}
\begin{axis}[width=.95\linewidth,height=7.3cm,xlabel={Damping probability $\gamma$},ylabel={Bits per source copy},xmin=0,xmax=1,ymin=0,ymax=1.02,grid=major,legend style={font=\footnotesize,at={(.98,.98)},anchor=north east,draw=none,fill=white},tick label style={font=\small}]
\addplot[thick] table[x=gamma,y=forward,col sep=comma]{data/ad_state_rates.csv};
\addlegendentry{Forward one-way rate}
\addplot[thick,densely dashed] table[x=gamma,y=reverse,col sep=comma]{data/ad_state_rates.csv};
\addlegendentry{Reverse one-way rate}
\addplot[thick,dashdotted] table[x=gamma,y=vf,col sep=comma]{data/ad_state_rates.csv};
\addlegendentry{Forward state potential $V_{\rm F}^{\rightarrow}$}
\addplot[thick,dotted] table[x=gamma,y=entropy,col sep=comma]{data/ad_state_rates.csv};
\addlegendentry{Receiver entropy}
\end{axis}
\end{tikzpicture}
 \caption{One-way randomness from the amplitude-damping Choi state. The two operational rates follow from the separately cited support-additivity theorem and Koashi--Winter duality, as derived in~\eqref{eq:AD-state-exact-rate}--\eqref{eq:AD-reverse-rate}. The state-potential curve is an upper bound for the forward rate, not the rate itself. The receiver entropy gives a looser forward upper bound. This exactly solved family is a benchmark; the finite-temperature example below is used to demonstrate the quantitative value of the new converse.}
 \label{fig:AD-state-rates}
\end{figure}
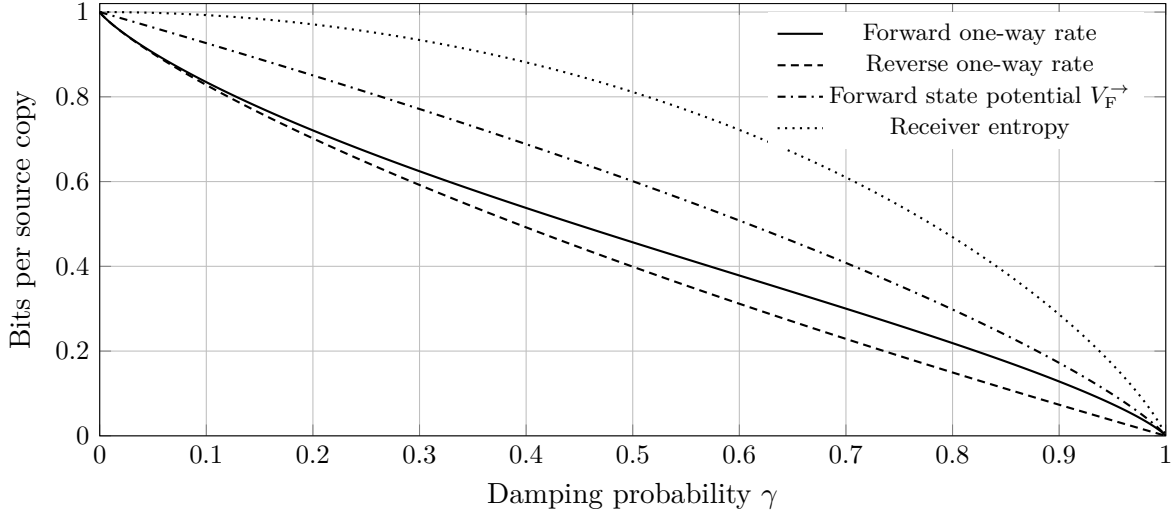

\subsection{A full-rank finite-temperature source}\label{subsec:GAD-source}
Generalized amplitude damping allows upward as well as downward transitions. This removes the support condition used to evaluate the zero-temperature rate and provides a less symmetric test for a one-way converse. Following~\cite{KSW20}, define the generalized amplitude damping channel as
\begin{equation}
 \mathcal G_{\gamma,\nu}(\tau)\coloneqq
 (1-\nu)\cA_\gamma(\tau)
 +\nu\sigma_x\cA_\gamma(\sigma_x\tau\sigma_x)\sigma_x,
 \qquad \gamma,\nu\in[0,1],
 \label{eq:GAD-definition}
\end{equation}
where $\nu$ is the equilibrium excited-state population. The symbol $\sigma_x$ here denotes the Pauli bit flip, not a conditional density operator. The corresponding Choi state is defined as
\begin{equation}
 \theta^{\gamma,\nu}_{AB}\coloneqq
 (\id_A\otimes\mathcal G_{\gamma,\nu})(\Phi)
 =\frac12\begin{pmatrix}
 1-\gamma\nu&0&0&\sqrt{1-\gamma}\\
 0&\gamma\nu&0&0\\
 0&0&\gamma(1-\nu)&0\\
 \sqrt{1-\gamma}&0&0&1-\gamma(1-\nu)
 \end{pmatrix}.
 \label{eq:GAD-Choi}
\end{equation}
We concentrate on the explicit state
\begin{equation}
 \rho^\star_{AB}\coloneqq\theta^{1/2,1/10}_{AB}
 =\frac1{40}\begin{pmatrix}
 19&0&0&10\sqrt2\\0&1&0&0\\0&0&9&0\\10\sqrt2&0&0&11
 \end{pmatrix}.
 \label{eq:thermal-star}
\end{equation}
Its marginals are $\rho_A^\star=I/2$ and $\rho_B^\star=\diag(7/10,3/10)$, and its eigenvalues are
\begin{equation}
 \frac{15+6\sqrt6}{40},\quad\frac{15-6\sqrt6}{40},\quad
 \frac1{40},\quad\frac9{40}.
 \label{eq:thermal-eigenvalues}
\end{equation}
They are all positive. The channel is nonunital, and, as the determinant formula below shows, every pure input has a mixed output. Thus neither the isotropic additivity calculation nor the pure-output support theorem evaluates this source's regularized one-way rate. We do not claim that failure of those methods proves the absence of other exact results; no exact rate for~\eqref{eq:thermal-star} is established here.

\subsection{Reducing all instruments to a scalar convex envelope}\label{subsec:GAD-roof}
A numerical search over a few instruments gives a lower bound on $\VF$, not the upper bound needed for a converse. We therefore reduce the full optimization first and then certify its value. This reduction uses the environment-measured penalty derived by Fang~\cite{Fang26}; the source fixes the average input to $I/2$, whereas a channel-capacity bound also optimizes that average.

Write an input state as $\tau(t,z)\coloneqq\left(\begin{smallmatrix}1-t&z\\z^*&t\end{smallmatrix}\right)$ and define
$u\coloneqq\det\tau=t(1-t)-|z|^2$. Its physical domain is
\begin{equation}
 0\leq t\leq1,\qquad 0\leq u\leq t(1-t).
\end{equation}
Define $a\coloneqq1-\gamma$ and
\begin{align}
 b_{\gamma,\nu}(t)&\coloneqq\gamma\nu+a t,\\
 \Delta_{\gamma,\nu}(t)&\coloneqq
 \gamma\big[\nu(1-\nu)+a(t-\nu)^2\big],\\
 c_{\gamma,\nu}&\coloneqq
 \big[\sqrt a-\gamma\sqrt{\nu(1-\nu)}\big]_+.
 \label{eq:GAD-scalars}
\end{align}
The output is $\left(\begin{smallmatrix}1-b&\sqrt a\,z\\\sqrt a\,z^*&b\end{smallmatrix}\right)$, so expanding its determinant gives
\begin{align}
 \det\mathcal G_{\gamma,\nu}(\tau)
 &=b_{\gamma,\nu}(t)\big(1-b_{\gamma,\nu}(t)\big)-a|z|^2\\
 &=b_{\gamma,\nu}(t)\big(1-b_{\gamma,\nu}(t)\big)-at(1-t)+au\\
 &=\Delta_{\gamma,\nu}(t)+au.
 \label{eq:GAD-determinant}
\end{align}
For $u=0$, the determinant is at least $\gamma\nu(1-\nu)>0$ when both parameters are in $(0,1)$. This proves the asserted lack of pure outputs.

The Choi concurrence calculated from~\eqref{eq:GAD-Choi} is $c_{\gamma,\nu}$: its only nonzero antidiagonal entry contributes $\sqrt a$, while the square root of the two middle diagonal entries contributes $\gamma\sqrt{\nu(1-\nu)}$. The local-filter identity of Appendix~\ref{app:concurrence} therefore gives
\begin{equation}
 f_{\mathcal G_{\gamma,\nu}}(\tau)
 =\ent(c_{\gamma,\nu}^{\,2}u).
 \label{eq:GAD-formation}
\end{equation}
Here $\ent$ is the qubit determinant-entropy function from~\eqref{eq:det-entropy}. Thus the conditional penalty is defined as
\begin{equation}
 g_{\gamma,\nu}(t,u)\coloneqq
 \ent\big(\Delta_{\gamma,\nu}(t)+au\big)
 -\ent(c_{\gamma,\nu}^{\,2}u).
 \label{eq:GAD-penalty}
\end{equation}
Define its minimum at fixed population and its lower convex envelope as
\begin{equation}
 \ell_{\gamma,\nu}(t)\coloneqq
 \min_{0\leq u\leq t(1-t)}g_{\gamma,\nu}(t,u),
 \qquad \operatorname{co}\ell_{\gamma,\nu}
 \text{ is the greatest convex function below }\ell_{\gamma,\nu}.
 \label{eq:GAD-envelope-definition}
\end{equation}
The minimization is over a compact interval. Continuity follows, including at the endpoints, by writing $u=v t(1-t)$ with $v\in[0,1]$ and using continuity on a fixed compact square.

\begin{proposition}[Exact fixed-input envelope]\label{prop:GAD-VF-envelope}
For the generalized amplitude-damping Choi source,
\begin{equation}
 \VF(\theta^{\gamma,\nu})
 =h_2\big(b_{\gamma,\nu}(1/2)\big)
 -(\operatorname{co}\ell_{\gamma,\nu})(1/2).
 \label{eq:GAD-VF-envelope}
\end{equation}
At most four possibly mixed input states are needed to attain the corresponding constrained ensemble optimization. No convexity of $\ell_{\gamma,\nu}$ itself is assumed.
\end{proposition}
\begin{proof}
Every ensemble averaging to $I/2$ has $\sum_xp_xt_x=1/2$. Each penalty obeys
$g_{\gamma,\nu}(t_x,u_x)\geq\ell_{\gamma,\nu}(t_x)
\geq(\operatorname{co}\ell_{\gamma,\nu})(t_x)$. Hence
\begin{equation}
 \sum_xp_xg_{\gamma,\nu}(t_x,u_x)
 \geq\sum_xp_x(\operatorname{co}\ell_{\gamma,\nu})(t_x)
 \geq(\operatorname{co}\ell_{\gamma,\nu})(1/2).
 \label{eq:GAD-envelope-lower}
\end{equation}
The second inequality follows from Jensen's inequality. Proposition~\ref{prop:Choi-fixed} gives the required upper bound on $\VF$.

For equality, the lower boundary of the convex hull of the graph of a continuous function on an interval can be attained at a prescribed population using at most two graph points. Indeed, take a supporting line at population $1/2$. The graph points in an optimal decomposition lie on that line; the smallest and largest contact populations bracket $1/2$ and suffice to reproduce it. Write those populations as $t_1,t_2$ and their weights as $\lambda_1,\lambda_2$, with $\lambda_1t_1+\lambda_2t_2=1/2$. A single point is allowed. Choose a minimizing determinant $u_j$ at each $t_j$ and set
\begin{equation}
 \tau_j^\pm\coloneqq
 \begin{pmatrix}
 1-t_j&\pm\sqrt{t_j(1-t_j)-u_j}\\
 \pm\sqrt{t_j(1-t_j)-u_j}&t_j
 \end{pmatrix}.
 \label{eq:GAD-opposite-coherences}
\end{equation}
These states are positive, have determinant $u_j$, and have the same penalty. The ensemble with weights $\lambda_j/2$ cancels the coherences and has average $I/2$. Its average penalty equals the convex-envelope value. The converse Choi construction in Proposition~\ref{prop:Choi-fixed} realizes it by an instrument, proving attainability and the four-state bound.
\end{proof}

This is the fixed-average specialization of the environment-measured convexification in~\cite{Fang26}, with an explicit realization of its contact points as source instruments. It should not be presented as an improvement over every fixed-input use of that earlier method. Its role here is to evaluate a one-way state converse while preserving the constraint imposed by the source.

\subsection{A certified interval for the thermal example}\label{subsec:thermal-interval}
For $\rho^\star$, set
\begin{equation}
 \alpha\coloneqq0.4287673,\qquad \beta\coloneqq0.7391729,
 \qquad c_*\coloneqq1/\sqrt2-3/20.
 \label{eq:thermal-certificate-constants}
\end{equation}
All finite decimals in this certificate are exact rational constants. The global inequality
\begin{equation}
 g_{1/2,1/10}(t,u)\geq\alpha+\beta(t-1/2)
 \quad\text{for }0\leq t\leq1,\ 0\leq u\leq t(1-t)
 \label{eq:thermal-affine-certificate}
\end{equation}
is verified by an exhaustive outward-interval calculation detailed in Appendix~\ref{app:thermal-cert}. Its checker covers the complete domain; it does not trust the output of a numerical optimizer. Averaging~\eqref{eq:thermal-affine-certificate} over an arbitrary admissible ensemble cancels its linear term and yields
\begin{equation}
 \VF(\rho^\star)\leq h_2(3/10)-\alpha
 =0.4525235992306926\ldots<0.45252360.
 \label{eq:thermal-VF-upper}
\end{equation}

For a lower bound on the value of the potential, choose $t=1/2$ and the exact rational determinant $u_0\coloneqq0.0583061$, with the two opposite coherences in~\eqref{eq:GAD-opposite-coherences}. This ensemble is feasible and has objective
\begin{equation}
 h_2(3/10)-\ent(17/200+u_0/2)+\ent(c_*^2u_0)
 =0.4525235098589629\ldots>0.45252350.
 \label{eq:thermal-VF-lower}
\end{equation}
Consequently,
\begin{equation}
 0.45252350<\VF(\rho^\star)<0.45252360.
 \label{eq:thermal-VF-interval}
\end{equation}
A feasible instrument supplies the lower inequality on the potential; the global affine certificate, not that instrument, supplies the capacity upper bound. The distinction is important because maximizing over a sampled set of instruments alone would not establish a converse.

For an achievable randomness rate, instead measure Alice in the $X$ basis. The equiprobable Bob states are
\begin{equation}
 \omega_B^\pm=\begin{pmatrix}7/10&\pm1/(2\sqrt2)\\
                    \pm1/(2\sqrt2)&3/10\end{pmatrix}.
\end{equation}
Their Bloch length is $\sqrt{33/50}$, so one-way reconciliation achieves
\begin{equation}
 L_*\coloneqq h_2(3/10)-h_2\!\left(\frac{1+\sqrt{33/50}}2\right)
 =0.4322690641609457\ldots.
 \label{eq:thermal-achievable}
\end{equation}
Combining~\eqref{eq:thermal-achievable}, Theorem~\ref{thm:VF-oneway}, and~\eqref{eq:thermal-VF-upper} proves
\begin{equation}
 0.43226906<\Rone(\rho^\star)\leq\VF(\rho^\star)<0.45252360.
 \label{eq:thermal-operational-interval}
\end{equation}
The interval width is less than $0.020255$ bits per source copy. We do not claim that either endpoint is the exact rate. Theorem~\ref{thm:VF-oneway} also gives, for every $n\geq1$ and $\eps\in[0,1)$,
\begin{equation}
 R_{\rightarrow}^{n,\eps}(\rho^\star)
 <\frac{0.45252360+\overline h_2(\sqrt\eps)/n}{1-\sqrt\eps}.
 \label{eq:thermal-finite-randomness}
\end{equation}
The correction is independent of the classical communication alphabet.

The nearly optimal potential instrument is an unsharp $X$ measurement with effects $(I\pm s\sigma_x)/2$, where $s\coloneqq\sqrt{1-4u_0}\approx0.87565724$, and with its quantum output retained. Its two potential contributions are approximately $0.31995541$ bits of measurement-induced correlation and $0.13256810$ bits of retained entanglement of formation. Their sum is a potential value, not an achievable net randomness rate. Figure~\ref{fig:thermal-penalty} shows why an affine bound on the minimized penalty controls even ensembles that use different populations.

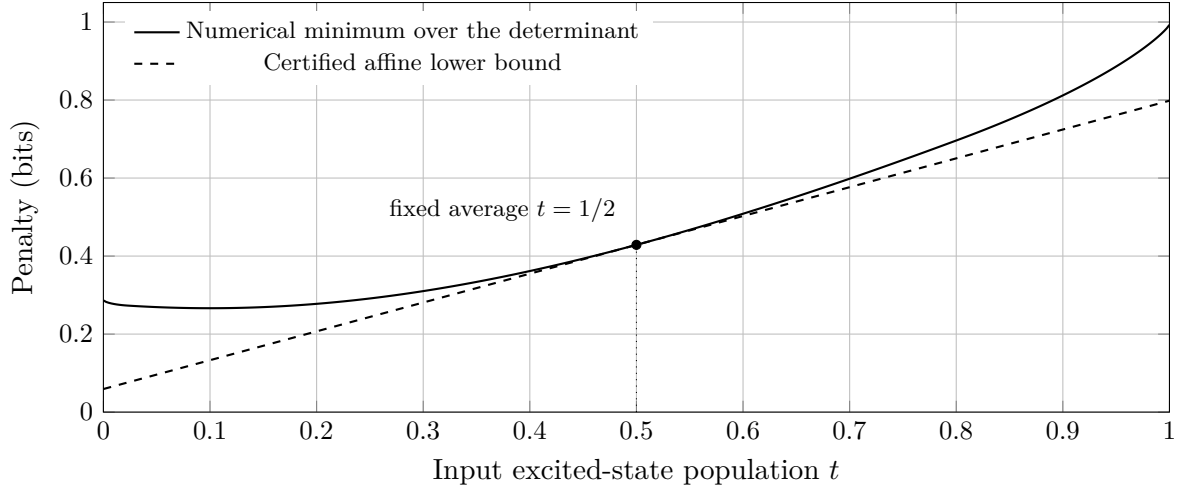
\begin{figure}[t]
 \centering\begin{tikzpicture}
\begin{axis}[width=.95\linewidth,height=7.0cm,xlabel={Input excited-state population $t$},ylabel={Penalty (bits)},xmin=0,xmax=1,ymin=0,ymax=1.05,grid=major,legend style={font=\footnotesize,at={(.02,.98)},anchor=north west,draw=none,fill=white},tick label style={font=\small}]
\addplot[thick] table[x=t,y=penalty,col sep=comma]{data/thermal_penalty.csv};
\addlegendentry{Numerical minimum over the determinant}
\addplot[thick,dashed] table[x=t,y=affine,col sep=comma]{data/thermal_penalty.csv};
\addlegendentry{Certified affine lower bound}
\addplot[dotted] coordinates {(.5,0)(.5,.4287673)};
\addplot[only marks,mark=*,mark size=1.7pt] coordinates {(.5,.4287673)};
\node[anchor=south east,font=\footnotesize] at (axis cs:.49,.46) {fixed average $t=1/2$};
\end{axis}
\end{tikzpicture}
 \caption{The minimized generalized amplitude-damping penalty at $\gamma=1/2,\nu=1/10$, and the globally certified affine function in~\eqref{eq:thermal-affine-certificate}. The plotted minimum is a numerical illustration; the inequality over every population and determinant is established separately by the interval checker. An ensemble with average input $I/2$ has average population $1/2$, so averaging the affine function always gives $\alpha$, irrespective of the ensemble. This is the upper-bound mechanism in~\eqref{eq:thermal-VF-upper}.}
 \label{fig:thermal-penalty}
\end{figure}

\subsection{Comparison with previous state converses}\label{subsec:thermal-comparison}
The improvement should be assessed against the correct state quantities. In~\cite[Section~VI]{LRWW23}, fixing one factor of a bilinear comparison gives two computable restrictions, denoted here by $\Upsilon_A(\rho)$ and $\Upsilon_B(\rho)$. These are not the channel-$\Upsilon$ quantity used in Section~\ref{sec:depolarizing}. Let $T_B$ denote partial transpose in the chosen basis. For a fixed target state $\rho_{AB}$, define
\begin{equation}
 \begin{aligned}
 \mathfrak F_A(\rho)\coloneqq\{\sigma_{AB}\geq0:\ &\exists\,Z_{AB},L_B\text{ Hermitian},\quad
 T_B(Z_{AB}\pm\sigma_{AB})\geq0,\\
 &\rho_A\otimes L_B\pm Z_{AB}\geq0,\quad\Tr L_B\leq1\},
 \end{aligned}
 \label{eq:state-Upsilon-feasible-A}
\end{equation}
with both signs required. Define $\mathfrak F_B(\rho)$ by replacing $\rho_A\otimes L_B$ with $K_A\otimes\rho_B$ and $\Tr L_B\leq1$ with $\Tr K_A\leq1$. The corresponding relative-entropy quantities are
\begin{equation}
 \Upsilon_A(\rho)\coloneqq\inf_{\sigma\in\mathfrak F_A(\rho)}D(\rho\|\sigma),
 \qquad
 \Upsilon_B(\rho)\coloneqq\inf_{\sigma\in\mathfrak F_B(\rho)}D(\rho\|\sigma).
 \label{eq:state-Upsilon-definition}
\end{equation}
The positive comparison operators need not have unit trace. The earlier results give
$\Rtwo(\rho)\leq\Gamma(\rho)\leq\min\{\Upsilon_A(\rho),\Upsilon_B(\rho)\}$, together with stronger one-shot and strong-converse variants~\cite{LRWW23,NW24}.

To prove that a new upper bound is smaller, a feasible point of~\eqref{eq:state-Upsilon-definition} is insufficient: it only upper-bounds the older quantity. We need a lower bound on its optimized value. Appendix~\ref{app:state-upsilon-cert} constructs global supporting-hyperplane certificates and evaluates them with outward intervals. They give
\begin{equation}
 \Upsilon_A(\rho^\star)>0.66600769,\qquad
 \Upsilon_B(\rho^\star)>0.66688169.
 \label{eq:thermal-old-certificates}
\end{equation}
Thus our upper bound $0.45252360$ is strictly below both optimized restrictions. We do not infer an ordering with the unrestricted $\Gamma$ minimization, and the improvement in the displayed numerical threshold concerns the vanishing-error one-way rate. Theorem~\ref{thm:VF-oneway} also gives its finite-error version, but not a two-way or strong-converse bound.

A second benchmark uses more structure than a receiver-entropy estimate. The source is a mixture of the zero-temperature damping and raising Choi states. Expanding $n$ copies gives a mixture of tensor products of the corresponding channels. For a fixed POVM on Alice's Choi shares, the induced input ensemble has average $I/2^n$ for every component. Holevo information is convex in the channel for a fixed ensemble. The constrained Holevo additivity from~\cite[Corollary~5.1]{TZBW26} bounds each component's Holevo information by $nL_{\rightarrow}(\gamma)$; the raising case is related by bit flips, which preserve the average. Averaging the bounds and optimizing Alice's measurement gives
\begin{equation}
 \Rone(\theta^{\gamma,\nu})\leq L_{\rightarrow}(\gamma).
 \label{eq:thermal-component-bound}
\end{equation}
For clarity, the constrained step uses the output--environment entropy roof, not merely the unconstrained classical capacity: if $\widehat S_\mathcal N(\tau)$ is the minimum average output entropy over pure ensembles of average $\tau$, then the cited theorem gives
\begin{equation}
 \widehat S_{\mathcal N_1\otimes\cdots\otimes\mathcal N_n}(I/2^n)
 \geq\sum_i\widehat S_{\mathcal N_i}(I/2).
\end{equation}
The average output is a product, so its entropy is the sum of its marginal entropies. Subtracting proves the component Holevo bound used above. At $\gamma=1/2$,~\eqref{eq:thermal-component-bound} is $0.4566992218\ldots$, still strictly larger than~\eqref{eq:thermal-VF-upper}.

\begin{table}[t]
 \centering
 \caption{Benchmarks for $\rho^\star$, in bits per source copy. The older-converse row is a lower certificate on the value of a published upper-bound optimization, not a rate claimed to be achievable or a numerical upper endpoint for the operational rate.}
 \label{tab:thermal-bounds}
 \begin{tabular}{@{}p{8.0cm}r@{}}\toprule
 Quantity & Value or certified statement\\\midrule
 Explicit $X$-measurement achievable rate $L_*$ & $0.4322690642\ldots$\\
 New state potential $\VF(\rho^\star)$ & $(0.45252350,\,0.45252360)$\\
 Exact-component rate upper bound & $0.4566992218\ldots$\\
 Value of $\min\{\Upsilon_A,\Upsilon_B\}(\rho^\star)$ & $>0.66600769$\\
 Receiver entropy $S(B)_{\rho^\star}$ & $0.8812908992\ldots$\\\bottomrule
 \end{tabular}
\end{table}

Fang's generalized-amplitude-damping channel bound also uses the penalty in~\eqref{eq:GAD-penalty}, together with an output-measured penalty, and maximizes over the average occupation~\cite{Fang26}. Retaining the fixed input occupation in its environment-measured part recovers~\eqref{eq:GAD-VF-envelope}. We accordingly credit this overlap rather than claim improvement over all refinements of that method. The comparison in Table~\ref{tab:thermal-bounds} has two narrower, proved meanings: improvement over the specified state-$\Upsilon$ restrictions and over the exact-component estimate~\eqref{eq:thermal-component-bound}.

\subsection{A nearby state outside the thermal family}\label{subsec:thermal-perturbation}
The improvement need not depend on one exactly specified thermal channel. Define a slightly noisier state by
\begin{equation}
 \rho_\eta^\star\coloneqq(1-\eta)\rho^\star+\eta I_{AB}/4,
 \qquad 0\leq\eta\leq1.
 \label{eq:thermal-perturbed-state}
\end{equation}
Its reference marginal remains $I/2$. For a fixed input ensemble with that average, both the output Holevo information and each retained entanglement term are convex in the channel. The supremum over this common constrained set is therefore convex as well. The replacer channel has zero constrained potential, giving
\begin{equation}
 \VF(\rho_\eta^\star)\leq(1-\eta)\VF(\rho^\star).
 \label{eq:thermal-perturbed-upper}
\end{equation}
This does not assert convexity of the state potential under arbitrary mixtures with different reference marginals.

For an $X$-basis measurement, the average Bob excited population is $3/10+\eta/5$, and the conditional Bloch length is $(1-\eta)\sqrt{33/50}$. Thus an achievable rate is
\begin{equation}
 L_\eta\coloneqq h_2(3/10+\eta/5)
 -h_2\!\left(\frac{1+(1-\eta)\sqrt{33/50}}2\right).
 \label{eq:thermal-perturbed-lower}
\end{equation}
At $\eta=1/100$, outward evaluation and~\eqref{eq:thermal-perturbed-upper} imply
\begin{equation}
 0.42154663<\Rone(\rho_{1/100}^\star)
 \leq\VF(\rho_{1/100}^\star)<0.44799837.
 \label{eq:thermal-perturbed-interval}
\end{equation}
The comparison certificates in Appendix~\ref{app:state-upsilon-cert} also give
\begin{equation}
 \Upsilon_A(\rho_{1/100}^\star)>0.53805672,
 \qquad\Upsilon_B(\rho_{1/100}^\star)>0.52363972.
\end{equation}
Here the upper bound on $\VF$ is a convexity estimate, not an exact evaluation. Figure~\ref{fig:thermal-perturbation} illustrates the achievable and upper curves for this perturbation family without depicting an unknown optimum as an evaluated curve.

\begin{figure}[t]
 \centering\begin{tikzpicture}
\begin{axis}[width=.95\linewidth,height=6.8cm,xlabel={Additional white-noise weight $\eta$},ylabel={Bits per source copy},xmin=0,xmax=.1,xtick={0,.02,.04,.06,.08,.1},scaled x ticks=false,xticklabel style={/pgf/number format/fixed,/pgf/number format/precision=2},ymin=.32,ymax=.47,grid=major,legend style={font=\footnotesize,at={(.98,.98)},anchor=north east,draw=none,fill=white},tick label style={font=\small}]
\addplot[thick,dashed] table[x=eta,y=upper,col sep=comma]{data/thermal_perturbation.csv};
\addlegendentry{Certified upper bound on $V_{\rm F}^{\rightarrow}$}
\addplot[thick] table[x=eta,y=achievable,col sep=comma]{data/thermal_perturbation.csv};
\addlegendentry{Achievable $X$-measurement rate}
\addplot[dotted] coordinates {(.01,.32)(.01,.4479983632)};
\end{axis}
\end{tikzpicture}
 \caption{A neighborhood of the full-rank thermal source. The lower curve is the explicit measurement rate $L_\eta$ in~\eqref{eq:thermal-perturbed-lower}. The upper curve is the certified one-way bound $(1-\eta)[h_2(3/10)-\alpha]$ from~\eqref{eq:thermal-perturbed-upper}. It is not claimed to equal the state potential for $\eta>0$. The stated old-bound certificates apply at $\eta=0$ and $\eta=0.01$; no old-bound curve is inferred between these points.}
 \label{fig:thermal-perturbation}
\end{figure}
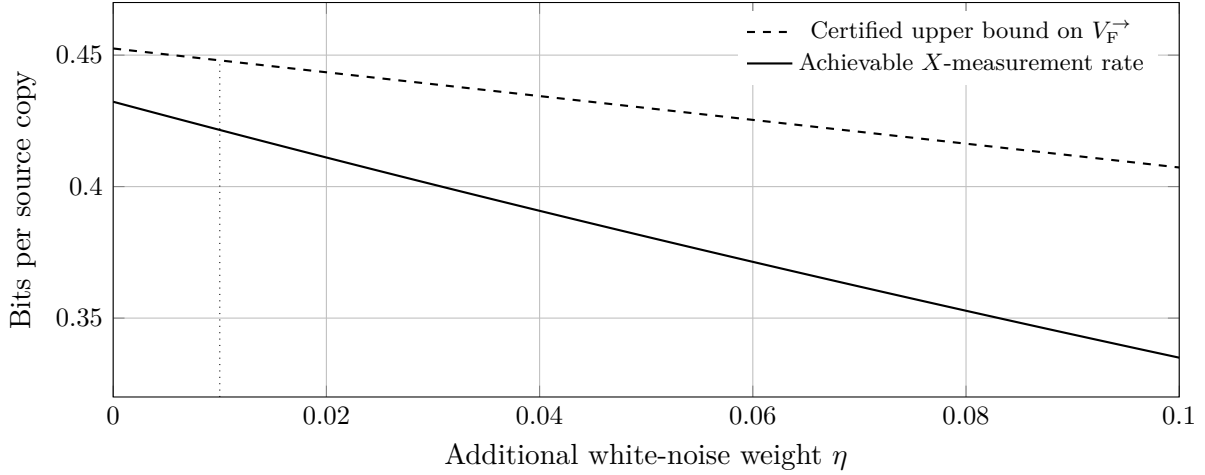

The thermal sources demonstrate a useful one-way converse without requiring an exact operational solution. We next return to the isotropic family for a different purpose: its one-way rate is exactly known, allowing us to prove that alternating communication directions changes the attainable randomness and invalidates a naive two-way extension of the same potential.

\section{Isotropic states and the state--channel bridge}\label{sec:isotropic}
We now specialize the one-way theory to qubit isotropic states. This example is useful for two reasons. First, their one-way distillable randomness can be evaluated exactly, so the looseness of the state potential can be seen explicitly. Second, an isotropic state is the normalized Choi state of a depolarizing channel, allowing the static quantity $\VF$ to be identified with the channel quantity $\UF$ already evaluated in Section~\ref{sec:depolarizing}. Having seen an effective thermal-state converse, we use this exactly solved family to isolate the change caused by a second direction of communication.

Define
\begin{equation}
 \rho^{(p)}_{AB}\coloneqq(1-p)\Phi_{AB}+p\frac{I_{AB}}4,
 \qquad 0\leq p\leq1,
 \label{eq:isotropic-state}
\end{equation}
where $\Phi_{AB}$ is a maximally entangled qubit state. Define the qubit depolarizing channel as $\cD_p(\tau)\coloneqq(1-p)\tau+pI_2/2$. Equivalently,
\begin{equation}
 \rho^{(p)}_{AB}=(\id_A\otimes\cD_{p,B})(\Phi_{AB}).
 \label{eq:isotropic-Choi}
\end{equation}

\subsection{Exact one-way distillable randomness}
The one-way distillable randomness of isotropic states is particularly simple. The following statement is a direct consequence of Devetak--Winter's regularized formula and King's additivity theorem~\cite{DW04,King03}. We record the derivation as an exact benchmark rather than an independent new coding theorem.

\begin{proposition}[One-way isotropic rate]\label{prop:isotropic-oneway}
For every $p\in[0,1]$,
\begin{equation}
 \Rone(\rho^{(p)})=1-h_2(p/2).
 \label{eq:isotropic-oneway}
\end{equation}
The same value holds when the direction of one-way communication is reversed.
\end{proposition}
\begin{proof}
The Devetak--Winter formula expresses the asymptotic one-way rate as the regularization of the measurement correlation $\mathcal J_n$ in~\eqref{eq:measurement-correlation-n}~\cite{DW04}. This is a maximization of the Holevo information generated at Bob by a measurement on Alice's $n$ systems, not an identification of $\mathcal J_n/n$ with the finite-error operational quantity. Let $(M_x)_x$ be such a POVM. Using the Choi identity, the conditional state on Bob is
\begin{equation}
 \cD_p^{\otimes n}\!\left(\frac{M_x^{\mathsf T}}{\Tr M_x}\right),
 \qquad p_x=\frac{\Tr M_x}{2^n},
 \label{eq:isotropic-conditional-output}
\end{equation}
while the average state is $I/2^n$. Therefore the resulting Holevo information is at most
\begin{equation}
 n-\min_{\tau}S(\cD_p^{\otimes n}(\tau)).
 \label{eq:isotropic-holevo-upper}
\end{equation}
King's minimum-output-entropy additivity theorem for the depolarizing channel~\cite{King03} evaluates the minimum as $nh_2(p/2)$. Hence no $n$-copy measurement exceeds $n[1-h_2(p/2)]$. Measuring each share in the computational basis attains this value already on one copy. Regularization proves~\eqref{eq:isotropic-oneway}. Exchange symmetry gives the reverse-direction statement.
\end{proof}

This formula is also the lower curve used for the isotropic example in~\cite{LRWW23}, and it is the first-order limit underlying the corresponding discussion in~\cite{NW24}. The next section shows that allowing a second communication round can do strictly better.

\subsection{Equality of the static and dynamical potentials}
The Choi representation~\eqref{eq:isotropic-Choi} makes $\VF$ a constrained version of the channel optimization. An instrument on Alice's share of $\Phi$ prepares, after transposition, an ensemble of input states for $\cD_p$ whose average is maximally mixed. Alice's retained system purifies each prepared input, so the two terms in~\eqref{eq:VF-definition} become exactly the two terms in the $\UF$ ensemble objective. Figure~\ref{fig:isotropic-bridge} summarizes this state--channel correspondence.

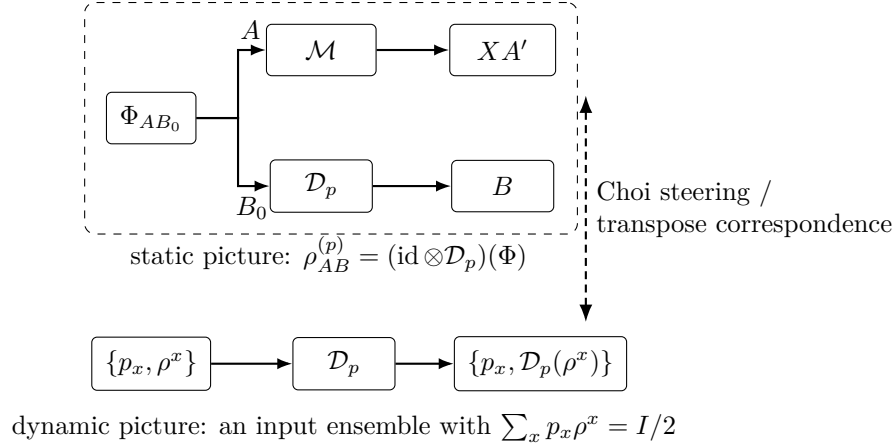
\begin{figure}[t]
 \centering
 \begin{tikzpicture}[font=\small,>=Latex]
  \node[op,minimum width=1.15cm] (bell) at (0,1.25) {$\Phi_{AB_0}$};
  \node[op,minimum width=1.35cm] (dep) at (2.25,0.35) {$\mathcal D_p$};
  \node[op,minimum width=1.45cm] (inst) at (2.25,2.15) {$\mathcal M$};
  \node[op,minimum width=1.4cm] (xa) at (4.65,2.15) {$X A'$};
  \node[op,minimum width=1.35cm] (bout) at (4.65,0.35) {$B$};
  \draw[qwire] (bell.east) -- ++(.55,0) |- (inst.west) node[pos=.72,above] {$A$};
  \draw[qwire] (bell.east) -- ++(.55,0) |- (dep.west) node[pos=.72,below] {$B_0$};
  \draw[qwire] (inst.east) -- (xa.west);
  \draw[qwire] (dep.east) -- (bout.west);
  \node[align=center] at (2.35,-.55) {static picture: $\rho^{(p)}_{AB}=(\id\otimes\mathcal D_p)(\Phi)$};

  \node[draw,dashed,rounded corners,inner sep=8pt,fit=(bell)(inst)(dep)(xa)(bout)] (topbox) {};

  \node[op,minimum width=1.45cm] (ens) at (0,-2.0) {$\{p_x,\rho^x\}$};
  \node[op,minimum width=1.35cm] (dep2) at (2.55,-2.0) {$\mathcal D_p$};
  \node[op,minimum width=1.65cm] (out2) at (5.15,-2.0) {$\{p_x,\mathcal D_p(\rho^x)\}$};
  \draw[qwire] (ens.east) -- (dep2.west);
  \draw[qwire] (dep2.east) -- (out2.west);
  \node[align=center] at (2.55,-2.85) {dynamic picture: an input ensemble with $\sum_xp_x\rho^x=I/2$};

  \draw[<->,thick,densely dashed] (5.75,1.55) -- (5.75,-1.45)
    node[midway,right,align=left] {Choi steering /\\transpose correspondence};
\end{tikzpicture}
 \caption{The state--channel bridge for qubit isotropic states. In the static picture, Alice applies an instrument to her share of the Choi state while Bob's share has passed through $\mathcal D_p$. By the Choi steering correspondence, the instrument prepares an ensemble of channel inputs, up to transposition, whose average is $I/2$. Retaining $A'$ purifies each ensemble member, so the classical-correlation and entanglement-of-formation terms match the two terms in the channel objective $\UF$.}
 \label{fig:isotropic-bridge}
\end{figure}

\begin{proposition}[Isotropic state--channel identity]\label{prop:isotropic-UF}
For every $p\in[0,1]$,
\begin{equation}
 \VF(\rho^{(p)})=\UF(\cD_p).
 \label{eq:isotropic-UF}
\end{equation}
\end{proposition}
\begin{proof}
The preceding Choi-instrument correspondence shows that $\VF(\rho^{(p)})$ is the supremum of the $\UF(\cD_p)$ objective over ensembles with average input $I/2$. It is therefore no larger than $\UF(\cD_p)$. Conversely, Proposition~\ref{prop:dep-UF} shows that for every candidate input state, its uniform Pauli orbit has average $I/2$ and preserves every conditional penalty in the objective. Thus an optimal $\UF$ ensemble can be chosen with the required average, proving the reverse inequality.
\end{proof}

Combining Propositions~\ref{prop:isotropic-oneway} and~\ref{prop:isotropic-UF} gives
\begin{equation}
 1-h_2(p/2)=\Rone(\rho^{(p)})\leq\VF(\rho^{(p)})=\UF(\cD_p).
 \label{eq:isotropic-chain}
\end{equation}
For $0<p<2/3$ the last inequality is strict, as proved in Proposition~\ref{prop:strict-UF}. For $p\geq2/3$, the depolarizing channel is entanglement breaking and equality holds. The equality in that regime will be especially useful as a diagnostic for the two-way problem.

\section{Two-round randomness distillation from isotropic states}\label{sec:twoway-randomness}
The one-way potential tracks one direction of classical extraction together with the entanglement that remains afterward. An unrestricted LOCC protocol can alternate the direction of extraction. We now show that this distinction is not merely technical: on every nontrivial qubit isotropic state, two rounds of communication achieve a strictly larger net randomness rate than the exact one-way optimum.

The construction has a simple operational interpretation. Bob performs a weak measurement, creating a classical variable while retaining a quantum output. He compresses that variable using Alice's quantum system as side information. Alice decodes it gently, so the residual bipartite state is essentially preserved. Only then does Alice make a complementary measurement and compress her new variable using Bob's retained quantum system as side information. Figure~\ref{fig:two-round-randomness} summarizes the order. This order matters: measuring Alice's system before the first decoder would destroy the quantum side information that makes the first message short.

\subsection{Bob's weak measurement and the first message}
Set $q\coloneqq1-p$ and choose a measurement strength $s\in[0,1]$. Bob applies the two-outcome instrument
\begin{equation}
 K_y(s)\coloneqq\sqrt{\frac{I+y s\sigma_z}{2}},\qquad y\in\{-1,+1\}.
 \label{eq:weak-instrument}
\end{equation}
He records $Y=y$ and retains the quantum output $B'$. Since the marginal of $B$ is maximally mixed, $Y$ is an unbiased bit. The normalized conditional state is
\begin{equation}
 \omega^y_{AB'}\coloneqq2(I_A\otimes K_y)\rho^{(p)}_{AB}(I_A\otimes K_y),
 \label{eq:weak-conditional-state}
\end{equation}
with marginals
\begin{equation}
 \omega_A^y=\frac{I+yqs\sigma_z}{2},\qquad
 \omega_{B'}^y=\frac{I+ys\sigma_z}{2}.
 \label{eq:weak-marginals}
\end{equation}
The full classical--quantum state is defined as
$\omega_{YAB'}\coloneqq\frac12\sum_{y=\pm1}\proj y_Y\otimes\omega^y_{AB'}$.

Thus
\begin{equation}
 I(Y;A)_\omega=1-h_2\!\left(\frac{1+qs}{2}\right),
 \qquad
 H(Y|A)_\omega=h_2\!\left(\frac{1+qs}{2}\right).
 \label{eq:first-compression-rate}
\end{equation}

On $n$ copies, classical data compression with quantum side information~\cite{DW03} therefore allows Bob to send a fixed-length description $L_1=f_1(Y^n)$ at any asymptotic rate larger than $H(Y|A)_\omega$. For each received description $\ell_1$, let $(\Lambda_{\widehat y^n|\ell_1})_{\widehat y^n}$ be Alice's decoding POVM on $A^n$. The codes can be chosen so that the average decoding error $\epsilon_1$ tends to zero.

The first decoder must preserve the correlations needed in the second round. Implement its successful effect by the square-root operation. If $\omega^{y^n}_{A^nB'^n}$ denotes the ideal conditional state before decoding, the gentle-measurement lemma, with $B'^n$ treated as an untouched reference, gives
\begin{equation}
 \sum_{y^n}p(y^n)\left\|\omega^{y^n}_{A^nB'^n}
 - \left(\sqrt{\Lambda_{y^n|f_1(y^n)}}\otimes I\right)\omega^{y^n}_{A^nB'^n}
   \left(\sqrt{\Lambda_{y^n|f_1(y^n)}}\otimes I\right)\right\|_1
 \leq 2\sqrt{\epsilon_1}.
 \label{eq:first-decoder-gentle}
\end{equation}
Including the incorrect decoding outcomes changes the corresponding classical--quantum state by only a further vanishing amount. Thus there is an implementation for which Alice obtains a copy of $Y^n$ and the post-decoding joint state is $o(1)$ in trace distance from the ideal state in which $Y^n$ was copied without disturbing $A^nB'^n$; see also the gentle-decoding discussion in~\cite{DW03,Wilde17}. This preservation is the resource used in the second round.

\subsection{Alice's complementary measurement and the second message}
After decoding $Y^n$, Alice measures each retained qubit in the $X$ basis and records $X\in\{-1,+1\}$. In the ideal post-decoding state, $X$ is uniform and independent of $Y$ because the marginal in~\eqref{eq:weak-marginals} is diagonal in the $Z$ basis. Explicitly,
\begin{equation}
 \Pr(X=x,Y=y)=\frac14,\qquad x,y\in\{-1,+1\}.
 \label{eq:XY-uniform}
\end{equation}
To derive Bob's conditional state explicitly, write the isotropic state as
$\rho^{(p)}=\frac14[I\otimes I+q(\sigma_x\otimes\sigma_x-\sigma_y\otimes\sigma_y+\sigma_z\otimes\sigma_z)]$.
Alice's $X$-outcome $x$ selects the unnormalized Bob operator $(I+xq\sigma_x)/4$ before his filter. Applying $K_y$, and dividing by the joint probability $1/4$, therefore gives $K_y(I+xq\sigma_x)K_y$. The identities
$K_y^2=(I+ys\sigma_z)/2$ and
$K_y\sigma_xK_y=\sqrt{1-s^2}\,\sigma_x/2$
give Bob's normalized conditional state
\begin{equation}
 \sigma_{B'}^{x,y}=\frac12\left[I+ys\sigma_z+xq\sqrt{1-s^2}\,\sigma_x\right].
 \label{eq:second-conditional-state}
\end{equation}
The classical--quantum state used below is defined as
$\sigma_{XYB'}\coloneqq\frac14\sum_{x,y=\pm1}\proj x_X\otimes\proj y_Y\otimes\sigma_{B'}^{x,y}$.

Its Bloch-vector length is
\begin{equation}
 v(p,s)\coloneqq\sqrt{q^2+(1-q^2)s^2}.
 \label{eq:vps}
\end{equation}
Therefore
\begin{align}
 I(X;B'|Y)_\sigma
 &=h_2\!\left(\frac{1+s}{2}\right)
   -h_2\!\left(\frac{1+v(p,s)}{2}\right),\label{eq:second-mutual-information}\\
 H(X|YB')_\sigma
 &=1-h_2\!\left(\frac{1+s}{2}\right)
   +h_2\!\left(\frac{1+v(p,s)}{2}\right).
 \label{eq:second-compression-rate}
\end{align}
Alice now applies classical compression with quantum side information in the opposite direction, sending $X^n$ to Bob at any rate larger than~\eqref{eq:second-compression-rate}. The first decoder's vanishing disturbance changes the second decoder's final error only by a vanishing amount, by trace-distance contractivity. At the end, both parties possess the pair $(X^n,Y^n)$, which in the ideal calculation is uniform over $4^n$ possibilities.

\begin{figure}[t]
 \centering
 \begin{tikzpicture}[x=1cm,y=1cm,font=\small]
\draw[-{Latex},thin] (.3,4.15)--node[above,font=\footnotesize]{time}(12.75,4.15);
\node[anchor=east] at (0,2.95) {Alice};
\node[anchor=east] at (0,0) {Bob};
\node[op,minimum width=1.65cm] (weak) at (1.4,0) {Weak $Z$\\instrument};
\node[op,minimum width=1.65cm] (dec1) at (5.0,2.95) {Gentle\\decoder};
\node[op,minimum width=2.35cm] (mx) at (8.45,2.95) {Measure $X^n$;\\compress its result};
\node[op,minimum width=1.8cm] (dec2) at (11.25,0) {Decode $X^n$};
\draw[qwire] (.1,2.95)--node[above]{$A^n$}(dec1.west);
\draw[qwire] (dec1.east)--node[above]{$A^n$}(mx.west);
\draw[cwire] (mx.east)--node[above,font=\footnotesize]{$(X^n,Y^n)$}(12.8,2.95);
\draw[qwire] (.05,0)--(weak.west);
\node[font=\footnotesize,anchor=east] at (.13,-.53) {$B^n$};
\draw[qwire] (weak.east)--node[below=4pt]{$B'^n$ retained until the second decoding}(dec2.west);
\node[op,minimum width=1.8cm,font=\footnotesize] (enc1) at (2.5,1.02) {Compress $Y^n$};
\draw[cwire] (weak.north)--(enc1.south);
\draw[cwire] (enc1.north east)--node[sloped,above=3pt]{$L_1$}(dec1.south west);
\draw[cwire] (mx.south)--node[sloped,above=3pt]{$L_2$}(dec2.north);
\draw[cwire] (dec2.east)--(13.3,0);
\node[font=\footnotesize,anchor=east] at (13.3,.68) {$(\widehat X^n,Y^n)$};
\node[font=\footnotesize,align=center,text width=4.4cm] at (6.5,1.25) {First decode $Y^n$ without appreciably disturbing the joint quantum state.};
\node[font=\footnotesize,align=center] at (6.5,-.95) {$L_1$: Bob to Alice\qquad $L_2$: Alice to Bob\qquad both lengths are subtracted.};
\end{tikzpicture}
 \caption{Two-round randomness distillation from an isotropic state. Bob first applies the weak instrument~\eqref{eq:weak-instrument}, retains $B'$, and compresses the classical outcome $Y^n$ using Alice's $A^n$ systems as quantum side information. Alice's decoder is chosen gentle, preserving the residual state. Alice then measures $X^n$ in the complementary basis and compresses it using Bob's retained systems $B'^n$ and the already shared $Y^n$ as quantum side information. The final common random variable is the pair $(X^n,Y^n)$. Both classical message rates are subtracted from the gross two-bit-per-copy randomness.}
 \label{fig:two-round-randomness}
\end{figure}

\subsection{Achievable rate and strict two-way improvement}
Subtracting the two communication rates from the gross rate of two shared bits per copy gives the following theorem.

\begin{theorem}[Two-round isotropic randomness rate]\label{thm:two-round-randomness}
For every $p,s\in[0,1]$,
\begin{equation}
 \Rtwo(\rho^{(p)})\geq r_2(p,s),
 \label{eq:two-round-achievable}
\end{equation}
where
\begin{equation}
 \begin{aligned}
 r_2(p,s)\coloneqq{}&1-h_2\!\left(\frac{1+(1-p)s}{2}\right)
 +h_2\!\left(\frac{1+s}{2}\right)\\
 &-h_2\!\left(\frac{1+\sqrt{(1-p)^2+[1-(1-p)^2]s^2}}{2}\right).
 \end{aligned}
 \label{eq:two-round-rate}
\end{equation}
Equivalently,
\begin{equation}
 r_2(p,s)=I(Y;A)_\omega+I(X;B'|Y)_\sigma.
 \label{eq:two-round-information-form}
\end{equation}
\end{theorem}
\begin{proof}
The first compression uses $H(Y|A)$ bits per copy and the second uses $H(X|YB')$ bits per copy. The ideal pair $(X,Y)$ is uniformly distributed, so its entropy is two bits. Hence the net rate is
$2-H(Y|A)-H(X|YB')$, which is~\eqref{eq:two-round-rate} after substituting~\eqref{eq:first-compression-rate} and~\eqref{eq:second-compression-rate}. The equality~\eqref{eq:two-round-information-form} follows from $H(X)=H(Y)=1$ and independence of $X$ and $Y$ in the ideal state.

For completeness, choose first- and second-stage compression codes with errors $\epsilon_1,\epsilon_2\to0$. Equation~\eqref{eq:first-decoder-gentle} and the probability of an incorrect first decoding show that the state entering the second round differs from its ideal counterpart by at most $2\sqrt{\epsilon_1}+\epsilon_1$ in trace norm. Alice's $X$-measurement, the second compression encoder, and Bob's decoder are quantum channels and therefore cannot increase this distance. The ideal second decoder disagrees with probability at most $\epsilon_2$, so the final pair of classical registers differs from the perfectly shared pair $(X^n,Y^n)$ by at most $2\sqrt{\epsilon_1}+\epsilon_1+2\epsilon_2$ in trace norm. Writing $\zeta_n\coloneqq2\sqrt{\epsilon_1}+\epsilon_1+2\epsilon_2$, inequality~\eqref{eq:fidelity-trace-distance} implies fidelity at least $[\max\{0,1-\zeta_n/2\}]^2$, which tends to one. Thus the protocol satisfies~\eqref{eq:randomness-finite-error} in the capacity limit. Both classical descriptions have fixed asymptotic lengths, and both lengths are subtracted. Thus the rate is achievable.
\end{proof}

The formula strictly improves on one-way extraction throughout the open noise interval. Let $q=1-p$. Expanding~\eqref{eq:two-round-rate} around $s=0$ gives
\begin{equation}
 r_2(p,s)=1-h_2(p/2)
 +\frac{1-q^2}{2\ln2}\left(\frac{\operatorname{atanh}q}{q}-1\right)s^2
 +O(s^4).
 \label{eq:two-round-expansion}
\end{equation}
To see the coefficient directly, use
$h_2((1+a s)/2)=1-a^2s^2/(2\ln2)+O(s^4)$ and
$v(p,s)=q+\left(1-q^2\right)s^2/(2q)+O(s^4)$, and then expand the last entropy in~\eqref{eq:two-round-rate} about $(1+q)/2$. This gives~\eqref{eq:two-round-expansion}. For $0<p<1$, one has $0<q<1$ and $\operatorname{atanh}q>q$, so the quadratic coefficient is strictly positive.

\begin{corollary}[Strict two-way-over-one-way separation]\label{cor:two-way-separation}
For every $0<p<1$,
\begin{equation}
 \Rtwo(\rho^{(p)})>\Rone(\rho^{(p)})=1-h_2(p/2).
 \label{eq:strict-two-way-randomness}
\end{equation}
\end{corollary}

A fixed measurement strength already gives substantial gains and avoids relying on numerical optimization. At $s=4/5$,
\begin{equation}
 r_2\!\left(\frac12,\frac45\right)
 =1-h_2(7/10)+h_2(9/10)-h_2\!\left(\frac{10+\sqrt{73}}{20}\right)
 =0.211418501897\ldots,
 \label{eq:p-half-fixed-rate}
\end{equation}
compared with the one-way value $1-h_2(1/4)=0.188721875541\ldots$. Numerical optimization over the one-parameter family gives a slightly larger value, approximately $0.21152695$ at $p=1/2$. The fixed-strength value proves the displayed numerical improvement at $p=1/2$; the small-strength expansion proves the all-parameter separation, without a global numerical optimization.

Figure~\ref{fig:isotropic-randomness-rates} compares the exact one-way rate, the fixed-strength two-round construction, the numerically optimized two-round family, and the state potential. The optimization is only over the displayed protocol family; it does not evaluate the unrestricted two-way rate.

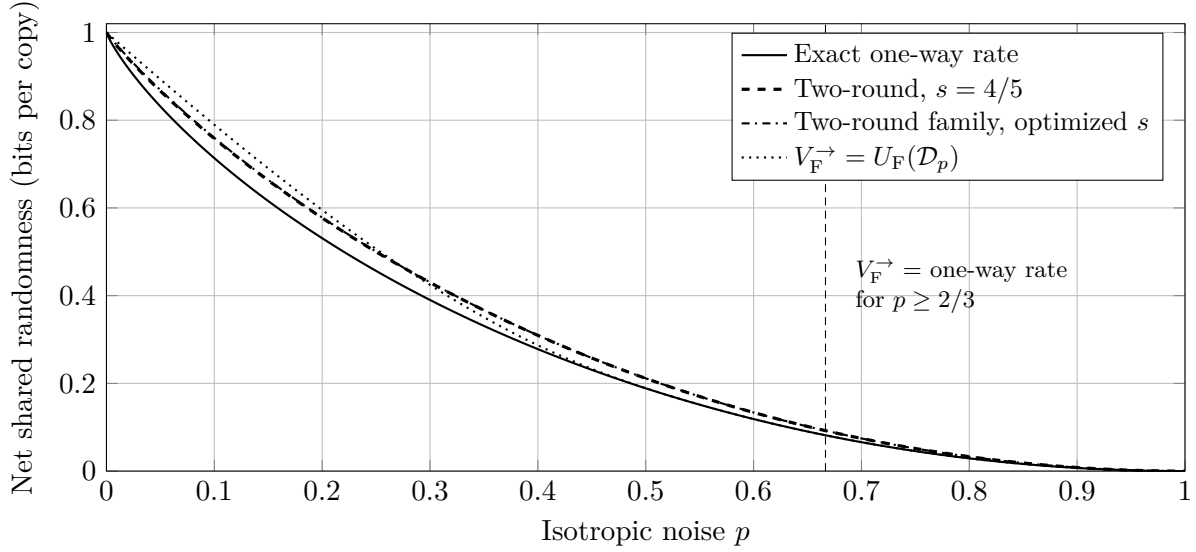
\begin{figure}[t]
 \centering
 \begin{tikzpicture}
\begin{axis}[width=.96\linewidth,height=7.5cm,xmin=0,xmax=1,ymin=0,ymax=1.02,
 xlabel={Isotropic noise $p$},ylabel={Net shared randomness (bits per copy)},grid=major,
 legend style={font=\small,at={(.98,.98)},anchor=north east},legend cell align=left]
\addplot[thick] table[x=p,y=oneway,col sep=comma]{data/isotropic_randomness.csv};\addlegendentry{Exact one-way rate}
\addplot[very thick,dashed] table[x=p,y=two_round_s08,col sep=comma]{data/isotropic_randomness.csv};\addlegendentry{Two-round, $s=4/5$}
\addplot[thick,dashdotted] table[x=p,y=two_round_opt,col sep=comma]{data/isotropic_randomness.csv};\addlegendentry{Two-round family, optimized $s$}
\addplot[thick,dotted] table[x=p,y=VF,col sep=comma]{data/isotropic_randomness.csv};\addlegendentry{$\VF=\UF(\cD_p)$}
\draw[densely dashed] (axis cs:.6666667,0)--(axis cs:.6666667,1.02);
\node[anchor=west,font=\footnotesize,align=left] at (axis cs:.685,.42) {$\VF=$ one-way rate\\for $p\geq2/3$};
\end{axis}
\end{tikzpicture}
 \caption{Randomness rates and the one-way potential for qubit isotropic states. The solid curve is the exact one-way rate $1-h_2(p/2)$ from Proposition~\ref{prop:isotropic-oneway}. The dashed curve is the explicit two-round rate at fixed measurement strength $s=4/5$, and the dash-dotted curve is the numerical maximum of~\eqref{eq:two-round-rate} over $s\in[0,1]$. The dotted curve is $\VF(\rho^{(p)})=\UF(\cD_p)$. For $2/3\leq p<1$ the dotted and one-way curves coincide, while the two-round curve lies strictly above them; hence $\VF$ cannot upper bound unrestricted two-way randomness distillation. The optimized two-round curve is not claimed to be the two-way capacity.}
 \label{fig:isotropic-randomness-rates}
\end{figure}

\subsection{Why the one-way potential cannot be a two-way converse}
The preceding construction provides a direct stress test for the candidate state potential. In the entanglement-breaking depolarizing regime $p\geq2/3$, Proposition~\ref{prop:isotropic-UF} and the channel evaluation give
\begin{equation}
 \VF(\rho^{(p)})=\UF(\cD_p)=1-h_2(p/2)=\Rone(\rho^{(p)}).
 \label{eq:VF-eb-regime}
\end{equation}
Corollary~\ref{cor:two-way-separation} then implies the analytic violation
\begin{equation}
 \Rtwo(\rho^{(p)})>\VF(\rho^{(p)}),\qquad 2/3\leq p<1.
 \label{eq:VF-not-two-way}
\end{equation}
Thus $\VF$ is not an unrestricted LOCC monotone after communication cost is accounted for in the manner needed for a two-way distillation converse.

The reason is conceptual. In feedback-assisted channel coding, a single message exists before the protocol starts, and the potential follows information about that same message as the protocol evolves. In randomness distillation, the parties may create different components of the final shared random variable at different stages. Here Bob creates $Y$, Alice learns it while retaining a useful quantum system, and Alice subsequently creates $X$. The final variable is the pair $(X,Y)$. An accounting quantity centered on one initial direction of extraction need not control this alternation.

The upper bounds $\Gamma$ and their fidelity-smoothed variants in~\cite{LRWW23,NW24} were designed to survive arbitrary LOCC rounds and remain valid. Our result does not improve those general two-way upper bounds. Instead, it improves the explicit isotropic-state achievability benchmark and shows what any sharper two-way converse must accommodate. A non-asymptotic version of the two-round protocol could combine finite-blocklength classical compression with quantum side information and an explicit disturbance budget; we leave that analysis for future work.

\section{Numerical evaluation and reproducibility}\label{sec:numerics}
The general converses and the ordinary amplitude-damping evaluation are analytic. The numerical enclosure for the thermal state additionally uses the global certificates in Appendices~\ref{app:thermal-cert} and~\ref{app:state-upsilon-cert}. Numerical work is used to display the scalar depolarizing bound, evaluate the single-use formula for the unassisted amplitude-damping capacity, and check finite-dimensional identities appearing in the proofs. We describe the computations separately so that a numerical maximization is not mistaken for a theorem about operational capacity.

\subsection{Stable depolarizing maximization}
Proposition~\ref{prop:dep-UF} reduces $\UF(\cD_p)$ to a one-dimensional maximization, but accurately evaluating that formula requires care. As shown in Section~\ref{subsec:strict-UF}, the maximizing input can be extremely close to pure, and the gain above the unassisted capacity can be much smaller than the capacities themselves. A uniform grid in the Bloch-vector length can miss this gain, while subtracting separately rounded capacities can erase it. We therefore describe three complementary checks: a stable coordinate and objective for locating the optimizer, a global subdivision estimate for its value, and high-precision recomputation of the smallest quoted gaps. These computations evaluate the analytical bound; they do not estimate the unknown feedback-assisted capacity.

The smaller input eigenvalue $\delta=(1-r)/2$ is a better numerical coordinate than $r$ near a pure input. The optimization scans a grid in $\log\delta$, supplemented by an ordinary grid away from zero, and refines stationary points and compares endpoints. The gain above capacity is evaluated as
\begin{equation}
 \mathcal E(2\kappa\sqrt{\delta(1-\delta)})-
 \big[h_2(p/2+(1-p)\delta)-h_2(p/2)\big],
 \label{eq:stable-gain}
\end{equation}
using cancellation-resistant expressions. In particular, the small eigenvalue in~\eqref{eq:dep-small-eigen} is evaluated as
\begin{equation}
 z_\delta=\frac{2\kappa^2\delta(1-\delta)}{1+\sqrt{1-4\kappa^2\delta(1-\delta)}}.
 \label{eq:stable-eigenvalue}
\end{equation}
Appendix~\ref{app:optimizer} proves that the optimizer is either the maximally mixed input or the unique interior stationary point, which provides an additional analytic check on the search.

A separate subdivision calculation gives a global upper estimate. On an interval $[\delta_-,\delta_+]\subseteq[0,1/2]$, the subtracted entropy increase in~\eqref{eq:stable-gain} is nondecreasing in $\delta$ and the entanglement term is also nondecreasing. Evaluating the former at $\delta_-$ and the latter at $\delta_+$ therefore bounds every gain in the interval from above. Subdividing intervals with an upper estimate above the best candidate bounds the unresolved absolute optimization error. The supplied data use an absolute tolerance of $2\times10^{-8}$ bits at 502 depolarizing parameter values. These calculations use ordinary floating-point arithmetic, not outward-rounded interval arithmetic, and are not described as formal numerical certificates.

Six representative stationary points were separately recomputed at 110-decimal-digit working precision, including $p=0.55$ and $p=0.6$. The absolute global tolerance is larger than the exceptionally small gap at the latter point. Its positivity is established by Proposition~\ref{prop:strict-UF}; its quoted magnitude is a high-precision numerical evaluation, not a consequence of the coarse global tolerance.

\subsection{Amplitude damping and matrix checks}
The amplitude-damping $\UF$ curve uses~\eqref{eq:AD-UF} directly at 401 parameter values. The unassisted-capacity formula is a scalar maximization in $t$, checked against an independent dense grid including endpoints. Strict concavity, proved after~\eqref{eq:AD-Holevo}, rules out distinct local maxima. Five representative optima are additionally recomputed by high-precision bisection of the derivative. The penalty plot uses~\eqref{eq:AD-g}. The exact $\UF$ proof does not rely on solving a discretized convex-roof problem.

The accompanying scripts additionally compare direct Kraus-matrix calculations with the concurrence and determinant formulas, test the fixed-measurement identity~\eqref{eq:AD-g-relative-entropy}, and check the scalar reduction and Jensen inequality on random ensembles. They also evaluate the actual-laboratory potential before and after random local feedback instruments, and test channel-step inequalities on mixed, potentially entangled two-qubit laboratory states. In these mixed channel-step tests the sender has no remaining quantum subsystem after the transmission; the output entanglement term is consequently zero. The tests cover a finite-dimensional subset of the theorem's inputs and do not replace the proof for arbitrary memories.

The matrix checks include 900 direct channel/concurrence comparisons, 1,200 instances each of the amplitude-damping diagonal reduction, scalar Jensen inequality, and arbitrary-ensemble upper bound, 350 feedback instruments, and 600 mixed channel-step comparisons. An additional 240 matrix calculations check explicit feasible witnesses for the semidefinite reduction in Appendix~\ref{app:upsilon}. The largest discrepancy in the fixed-environment measurement information is below $6\times10^{-16}$ bits. Direct spin-flip diagonalization for the rank-deficient amplitude-damping states has a larger numerical residual, below $2.8\times10^{-8}$ bits in entanglement of formation. No positive violation of the tested potential inequalities was observed. These are numerical consistency checks, not independent proof verification.

An additional revision check uses 240 message- and transcript-controlled Alice channels to verify the unchanged receiver marginal and entanglement monotonicity. It tests the channel-step bound on 180 pure and 180 locally flagged mixed erasure-channel ensembles with both retained memories nontrivial. Locally readable flags reduce the required entanglement values to marginal entropies and two-qubit entanglement of formation. These are restricted test families, not a solver for arbitrary convex roofs. The three dense-coding stages are also checked directly.

\subsection{Randomness-distillation checks}
The isotropic-state figure is generated directly from the closed expressions in Propositions~\ref{prop:isotropic-oneway} and~\ref{prop:isotropic-UF} and Theorem~\ref{thm:two-round-randomness}. For each of 501 equally spaced values of $p$ in $[0,1]$, the scalar maximization of $r_2(p,s)$ over $s\in[0,1]$ is performed by a bounded one-dimensional optimizer. As an independent check, the resulting value is compared with the largest value on a uniform grid of 20,001 measurement strengths. The largest optimizer--grid discrepancy is below $3.9\times10^{-10}$ bits per copy. At $p=1/2$, the optimizing strength is $s\approx0.8218742823$ and the corresponding rate is $0.2115269457$ bits per copy; the fixed choice $s=4/5$ gives $0.2114185019$ bits per copy. The small-$s$ expansion in~\eqref{eq:two-round-expansion} is also checked numerically against symmetric finite differences at representative noise values. These checks verify the plotted implementation of the formulas; the strict separation in Corollary~\ref{cor:two-way-separation} follows analytically from the positive quadratic coefficient and does not depend on numerical optimization.

The source package contains all plot data, editable TikZ/PGFPlots sources, deterministic generation code, fixed random seeds for the matrix checks, and numerical verification records. The figures state which curves are exact formulas, achievable rates, or numerically evaluated bounds. In particular, no plotted curve is labeled as an unknown feedback-assisted capacity.

The finite-blocklength randomness test is checked separately by \texttt{check\_finite\_randomness.py}. It tests 3,500 classical joint distributions, including nonuniform and zero-probability labels, 936 classical--quantum decoder instances, and 2,000 cases of the error-tolerance form of the channel Fano bound. The correlated and independent test probabilities are computed directly; the latter is compared with $d_L/D$. Independent uniformly shared padding of the key and message checks that the test depends on their net size. Boundary square roots are numerically ill-conditioned at unit fidelity, so the probability residuals and square-root residuals are recorded separately. A second implementation recomputes classical-test inequalities at high precision. These tests support the implementation and boundary-case analysis; the proof of Theorem~\ref{thm:VF-oneway} is the argument establishing the converse for arbitrary inputs.

\subsection{Damping-source plots and numerical certificates}
The new damping-source figures separate three kinds of quantities. Figure~\ref{fig:AD-state-rates} uses exact formulas for the state potential and the two one-way rates, with the rates attributed to the support-additivity result as explained in Section~\ref{subsec:AD-exact-rate}. Figure~\ref{fig:thermal-penalty} numerically minimizes over a single determinant at 401 populations, scanning endpoints and a combined logarithmic and uniform grid before local refinement. An independent 20,001-point grid at every population gives a maximum grid-to-refinement discrepancy below $1.4\times10^{-10}$ bits. This comparison checks the displayed curve only. The actual affine lower bound is established on the full continuous domain by the exhaustive interval calculation in Appendix~\ref{app:thermal-cert}.

The thermal certificate has 475 nonempty terminal rectangles and reconstructs the complete subdivision from its root. A separate interval computation checks both state-$\Upsilon$ supporting hyperplanes, all matrix-positivity conditions, and the stated perturbation. Their output files record enclosing intervals, not only midpoint estimates. These computations use outward interval arithmetic at 60-decimal-digit working precision. A separate check verifies the exact rational area and pairwise disjoint interiors of all terminal rectangles, compares the matrix-logarithm derivative with direct 80-digit quadrature, and repeats the interval enclosures at 90-digit working precision with all constants reinitialized. These are computer-assisted certificates in the stated arithmetic implementation, not a formal proof-assistant verification. By contrast, the earlier depolarizing subdivision remains an ordinary floating-point consistency bound; its certification status has not been changed by adding the thermal interval analysis.

Direct source checks cover 500 generalized-amplitude-damping channel/Choi identities and 500 retained-instrument instances. They compare the determinant and concurrence expressions with Kraus-matrix calculations, and verify the projected measurement rate. A further 600 ordinary-damping measurement checks and 600 retained-instrument checks verify the benchmark formulas and the source-potential upper bound on sampled instruments. Sampled instruments only supply consistency tests; the proofs in Section~\ref{sec:damping-states} and the global certificates are what cover every instrument. The source package includes the exact rational certificate matrices, the exhaustive partition, all new plotting data, and scripts that regenerate the tables and checks.

\section{Conclusion}\label{sec:conclusion}

\subsection{Summary}
We established an exact amortization theorem for message mutual information plus conditional entanglement of formation across the communicating laboratories. It gives a classical-feedback capacity converse for arbitrary finite-dimensional channels without initial shared entanglement. The channel quantity is an earlier Winter--Yang mixed-convex-roof expression~\cite[Theorem~10]{WY16}; the operational contribution is the proof that it bounds arbitrary adaptive feedback protocols. The resulting bound improves a published depolarizing-channel converse and is evaluated in closed form for amplitude damping. The same method determines the capacity of a flagged noiseless/entanglement-breaking mixture.

For state distillation, a weighted decoder test gives a finite-blocklength converse directly for the net randomness yield, without an error term proportional to the communication length. The finite-copy entanglement-of-formation inequality connects that test to the mixed-roof state potential, and the capacity limit recovers the one-way randomness bound. The ordinary amplitude-damping Choi family provides exact, separately attributed operational benchmarks in both communication directions. The generalized amplitude-damping Choi state $\rho^\star$ supplies a practical application beyond those solved examples: its one-way rate is at least $0.43226906$ and at most $0.45252360$ bits per copy. Global interval-verified certificates place the values of both earlier computable state-$\Upsilon$ restrictions above $0.66600769$. A small white-noise perturbation retains a strict improvement. These comparisons do not order the new bound against the unrestricted $\Gamma$ optimization or replace its two-way and strong-converse properties.

Finally, a two-round weak-measurement protocol strictly improves on the exact one-way isotropic-state rate throughout the nontrivial noise interval. Gentle first-round decoding preserves the quantum correlations used in the reverse-direction extraction. This protocol exceeds the one-way state potential in part of the parameter range, demonstrating why a fixed-message channel converse and an alternating randomness-extraction protocol cannot be identified. The formulas, examples, and certificates distinguish achievable rates, exact potential evaluations, and rigorous upper bounds on an unresolved operational rate.

\subsection{Directions for further work}
The present channel converse is weak. A hypothesis-testing or R\'enyi analogue of the potential could potentially establish a strong converse at a similarly small threshold, but such a replacement must preserve both free-operation monotonicity and a valid channel-step inequality. The regularized channel bound $\inf_k k^{-1}\UF(\N^{\otimes k})$ is another concrete target.

On the randomness side, the immediate problem is to develop an accounting quantity that remains valid under successive rounds of classical communication in both directions, after subtracting their communication cost, while remaining significantly smaller than the existing $\Gamma$-type strong-converse bounds. Equation~\eqref{eq:VF-not-two-way} is a nontrivial consistency test for any proposal. The two-round construction also suggests a hierarchy of multi-round lower bounds in which weak instruments and gentle decoders alternate between the parties. Determining whether this hierarchy converges to the true two-way distillable randomness, even for qubit isotropic states, appears to be an accessible and potentially revealing problem.

The finite-blocklength analysis is another natural direction. The smooth-divergence framework of~\cite{NW24} provides upper bounds, while non-asymptotic classical compression with quantum side information can be used to quantify the two messages and the disturbance caused by the first decoder. A matching second-order analysis would make it possible to compare the new interactive lower bound directly with the existing one-shot upper bounds.

For the full-rank thermal examples, an exact one-way rate has not been established here. Improving the achievable measurement rate or evaluating the complementary entanglement cost would reduce the remaining interval. The convex-envelope calculation is exact for the state potential, not a proof that the potential equals the rate.

More general flagged extensions may also retain enough noise to yield useful channel converses while exposing enough structure for tractable analysis. The contrast between the depolarizing extension and the amplitude-damping jump flag shows that the choice of revealed information is decisive. None of the results proved here establishes a strict operational feedback advantage for the depolarizing or amplitude-damping channel, nor do they determine the optimal two-way randomness rate of isotropic states.

\section*{Acknowledgments}
The author acknowledges support from the Cornell University School of Electrical and Computer Engineering.

\section*{Statement on AI-assisted preparation}
ChatGPT (OpenAI) Pro 6 (Astra) was used extensively for mathematical exploration, drafting and revision, bibliographic checks, and preparation of LaTeX, TikZ, and numerical code. The research direction and manuscript requirements were supplied by the author. AI-assisted mathematical, editorial, and numerical review passes were followed by revisions. These were internal self-reviews, not independent external peer review or formal verification. Numerical tests were used as consistency checks rather than substitutes for proofs. Responsibility for the mathematical statements, citations, and presentation rests with the author.

\bibliographystyle{alphaurl}
\bibliography{references}

\appendix

\section{Concurrence after a qubit channel}\label{app:concurrence}
This appendix explains the common two-qubit calculation behind the depolarizing and amplitude-damping examples. It also identifies which part is an established concurrence identity and which part is the subsequent channel-bound optimization.

Let $\Phi_2\coloneqq\proj\Phi$ with $\ket\Phi\coloneqq(\ket{00}+\ket{11})/\sqrt2$, and define the normalized Choi state as $\theta_\N\coloneqq(\id\otimes\N)(\Phi_2)$. For a qubit state $\rho$, a purification can be written as
\begin{equation}
 \ket{\phi^\rho}\coloneqq(L_\rho\otimes I)\ket\Phi,\qquad
 L_\rho\coloneqq\sqrt{2\rho^{\mathsf T}}.
 \label{eq:filter-purification}
\end{equation}
The transpose depends on the chosen basis and only ensures that the transmitted marginal is $\rho$. Since $\Tr_B\theta_\N=I/2$, the output $(L_\rho\otimes I)\theta_\N(L_\rho^\dagger\otimes I)$ has trace one.

The homogeneous concurrence of a possibly subnormalized pure two-qubit vector with coefficient matrix $C$ is defined as $2|\det C|$. Applying $L_\rho$ to its first subsystem multiplies this by $|\det L_\rho|$. Taking the convex roof gives the same scaling for mixed states: when $L_\rho$ is invertible, local filtering gives a bijection between decompositions before and after the filter, so both inequalities follow. Singular filters follow by continuity. This is the local-filter form of the concurrence factorization law~\cite{Konrad08}. Therefore
\begin{equation}
 c\big((\id\otimes\N)(\phi^\rho)\big)
 =2\sqrt{\det\rho}\,c(\theta_\N).
 \label{eq:concurrence-factorization}
\end{equation}
Wootters' formula~\cite{Wootters98} then evaluates the entanglement of formation.

For $\cD_p$, the Choi state is Bell diagonal, with eigenvalues $1-3p/4,p/4,p/4,p/4$. Its concurrence is $(1-3p/2)_+$. Since $\det\rho=(1-r^2)/4$, equation~\eqref{eq:concurrence-factorization} gives the expression in Proposition~\ref{prop:dep-UF}. For amplitude damping,
\begin{equation}
 \theta_{\cA_\gamma}=\frac12
 \begin{pmatrix}
 1&0&0&\sqrt{1-\gamma}\\
 0&0&0&0\\
 0&0&\gamma&0\\
 \sqrt{1-\gamma}&0&0&1-\gamma
 \end{pmatrix},
 \label{eq:AD-Choi}
\end{equation}
whose concurrence is $\sqrt{1-\gamma}$. Equations~\eqref{eq:concurrence-factorization} and~\eqref{eq:det-entropy} give $f_{\cA_\gamma}(\rho)=\ent((1-\gamma)\det\rho)$.

\section{The depolarizing \texorpdfstring{$\Upsilon$}{Upsilon} expression}\label{app:upsilon}
We derive the published comparison curve from its semidefinite definition, rather than inferring an analytic formula from a plotted line. The normalization of Choi operators matters in this calculation. We use $\theta_\M\coloneqq(\id\otimes\M)(\Phi_d)$ and the unnormalized Choi operator $J_\M\coloneqq d\theta_\M$, where $\Phi_d$ has trace one.

\subsection{Symmetry reduction of the comparison set}
For a point-to-point completely positive linear map $\M$, the comparison quantity of~\cite{WXD18,WFT19,DKQSW23} has the semidefinite representation
\begin{equation}
 \beta(\M)\coloneqq\inf_{R,S}\left\{\Tr S:
 R\pm J_\M^{\mathsf T_B}\geq0,\quad
 I\otimes S\pm R^{\mathsf T_B}\geq0\right\},
 \label{eq:beta-SDP}
\end{equation}
where $R,S$ are Hermitian and $\mathsf T_B$ denotes partial transpose. Both signs in each constraint are required. The $\Upsilon$ quantity minimizes a channel relative entropy over completely positive $\M$ with $\beta(\M)\leq1$~\cite{WFT19,DKQSW23}. These comparison maps need not preserve trace.

The depolarizing channel is covariant under all input unitaries. Twirling a comparison map and its feasible witnesses preserves the $\beta$ constraint and cannot increase the optimized relative entropy. For covariant channels the optimizing reference input can be taken maximally entangled; see the covariance reduction in~\cite[Section~VI]{DKQSW23}. Equivalently, the covariant channel and comparison map have the same teleportation processing of their Choi operators, so data processing gives the upper bound by their Choi divergence, with equality for a maximally entangled input. In this calculation the Choi operators commute, so the Umegaki and Belavkin--Staszewski relative entropies coincide. This explains the use of the same curve in the published feedback comparison~\cite{FF21,DKQSW23}.

A twirled comparison operator has the form
\begin{equation}
 \theta_\M=u\Phi_d+v(I-\Phi_d),\qquad u,v\geq0.
 \label{eq:isotropic-comparison}
\end{equation}
We claim that
\begin{equation}
 \beta(\M)=d\max\{u+(d-1)v,dv\}.
 \label{eq:isotropic-beta}
\end{equation}
To verify this, write $F$ for the swap, $P_\pm\coloneqq(I\pm F)/2$, and
\begin{equation}
 a_+\coloneqq u+(d-1)v,\qquad a_-\coloneqq|(d+1)v-u|.
 \label{eq:beta-AB}
\end{equation}
The partial transpose of $J_\M$ is $(u-v)F+dvI$. Twirling the witness in~\eqref{eq:beta-SDP} permits $R=r_+P_++r_-P_-$ and $S=sI$. Its first constraints become $r_+\geq a_+$ and $r_-\geq a_-$. Its second constraints require $s$ to dominate the absolute values of the two eigenvalues
\begin{equation}
 \frac{(d+1)r_+-(d-1)r_-}{2},\qquad
 \frac{r_++r_-}{2}
 \label{eq:beta-witness-eigenvalues}
\end{equation}
of $R^{\mathsf T_B}$ on $\Phi_d$ and its orthogonal complement.

If $u\geq v$, then $a_+\geq a_-$, and the choice $r_+=r_-=a_+$ attains $s=a_+$. Conversely, the larger of the two quantities in~\eqref{eq:beta-witness-eigenvalues}, before absolute values, is at least $r_+$: one of $r_-\geq r_+$ or $r_-\leq r_+$ always holds. Thus every feasible $s$ is at least $a_+$. If $u<v$, then $a_-=(d+1)v-u>a_+$, and every feasible $s$ is at least $(a_++a_-)/2=dv$. The choice $r_+=a_+,r_-=a_-$ attains this value, since the first eigenvalue is $du$ and the second is $dv$. This proves~\eqref{eq:isotropic-beta}.

\subsection{Optimizing the commuting divergence}
The normalized Choi state of $\cD_{p,d}$ is
\begin{equation}
 \theta_{p,d}=a\Phi_d+b(I-\Phi_d),\qquad
 a\coloneqq1-p+p/d^2,\quad b\coloneqq p/d^2.
 \label{eq:dep-choi-eigenvalues}
\end{equation}
Equations~\eqref{eq:isotropic-beta} and $\beta\leq1$ give
$u+(d-1)v\leq1/d$ and $v\leq1/d^2$. The divergence to~\eqref{eq:isotropic-comparison} is
\begin{equation}
 a\log_2(a/u)+(d^2-1)b\log_2(b/v).
 \label{eq:choi-divergence}
\end{equation}
For positive $a,b$, the minimum uses $u=1/d-(d-1)v$. Differentiating the remaining scalar objective gives the unconstrained optimum
\begin{equation}
 u=a/d,\qquad v=(d+1)b/d.
 \label{eq:choi-optimizer}
\end{equation}
This is feasible exactly when $p\leq d/(d+1)$. If $p$ is larger, the optimum is the boundary $u=v=1/d^2$. Endpoints follow by continuity. Consequently,
\begin{equation}
 \Upsilon(\cD_{p,d})=
 \begin{cases}
 \log_2d-\dfrac{d^2-1}{d^2}p\log_2(d+1),&p\leq d/(d+1),\\
 2\log_2d-H(a,b,\ldots,b),&p\geq d/(d+1).
 \end{cases}
 \label{eq:upsilon-all-d}
\end{equation}
Specializing to $d=2$ proves~\eqref{eq:dep-upsilon}. This evaluation concerns the limiting $\Upsilon$ expression; it does not identify every finite-order R\'enyi version with the same formula.

\section{Global structure of the depolarizing optimizer}\label{app:optimizer}
The scalar formula~\eqref{eq:dep-UF} suffices to state the bound. We include a derivative argument that identifies all possible maximizers and explains why logarithmic resolution near a pure input is important. This argument is independent of the numerical subdivision check.

Assume $0<p<2/3$, and set $t\coloneqq1-p$, $\kappa\coloneqq1-3p/2$, and
\begin{equation}
 s(r)\coloneqq\sqrt{1-\kappa^2+\kappa^2r^2},\qquad
 G(r)\coloneqq h_2((1+tr)/2)-h_2((1+s(r))/2).
 \label{eq:G-optimizer}
\end{equation}
Maximizing~\eqref{eq:dep-UF} is equivalent to minimizing $G$ on $[0,1]$. For $0<r<1$,
\begin{equation}
 \frac{\ln2}{r}G'(r)=\kappa^2 L(s(r))-t^2 L(tr),\qquad
 L(x)\coloneqq\frac{\operatorname{atanh}x}{x}=\int_0^1\frac{du}{1-x^2u^2},
 \label{eq:G-derivative}
\end{equation}
with $L(0)=1$. Combining the integrands yields
\begin{equation}
 \frac{\ln2}{r}G'(r)=\int_0^1N(u)w_r(u)\,du,
 \label{eq:G-sign-integral}
\end{equation}
where
\begin{align}
 N(u)&\coloneqq\kappa^2-t^2+t^2(1-\kappa^2)u^2,\\
 w_r(u)&\coloneqq\frac1{[1-(1-\kappa^2+\kappa^2r^2)u^2][1-t^2r^2u^2]}.
 \label{eq:G-integrands}
\end{align}
Because $0<\kappa<t<1$, $N(0)<0$ and $N(1)=\kappa^2(1-t^2)>0$. It has exactly one zero $u_0$ in $(0,1)$.

For $r_2>r_1$, the ratio $w_{r_2}(u)/w_{r_1}(u)$ is strictly increasing in $u$. To check this, factor it into terms of the form $(1-c(u)r_1^2)/(1-c(u)r_2^2)$, with
$c(u)=\kappa^2u^2/[1-(1-\kappa^2)u^2]$ and $c(u)=t^2u^2$, respectively. Each $c$ is increasing. Hence
\begin{equation}
 N(u)\left[\frac{w_{r_2}(u)}{w_{r_1}(u)}-
                \frac{w_{r_2}(u_0)}{w_{r_1}(u_0)}\right]>0
 \quad\text{for }u\neq u_0,\quad 0<u<1.
 \label{eq:single-crossing}
\end{equation}
Multiplying by $w_{r_1}$ and integrating shows that, once the integral in~\eqref{eq:G-sign-integral} is nonnegative, it is strictly positive for every larger $r$. As $r\uparrow1$, $L(s(r))$ diverges while $L(tr)$ remains bounded. The derivative sign is therefore eventually positive.

It follows that if the continuous sign expression at $r=0$ is nonnegative, $G$ is minimized at zero. Otherwise it decreases initially, has exactly one interior stationary point, and increases thereafter. This proves the global characterization used to check the scalar optimization. It does not rule out an optimizer extremely close to $r=1$, which is why the implementation uses $\delta=(1-r)/2$ and high-precision arithmetic for selected parameters.

\section{A global certificate for the thermal source potential}\label{app:thermal-cert}
The scalar search in Section~\ref{subsec:GAD-roof} suggests the value of the state potential, but a search cannot by itself exclude better ensembles. This appendix specifies the inequality actually certified, the elementary bounds used on each rectangle, and how the supplied checker verifies coverage. The result is a computer-assisted enclosure using outward interval arithmetic, not a theorem-prover verification.

\subsection{The function to be bounded}
At the parameters of~\eqref{eq:thermal-star}, the penalty in~\eqref{eq:GAD-penalty} is
\begin{equation}
 g_*(t,u)\coloneqq
 \ent\!\left(\frac{u}{2}+\frac9{200}+\frac{(t-1/10)^2}{4}\right)
 -\ent(c_*^2u),
 \qquad c_*\coloneqq\frac1{\sqrt2}-\frac3{20}.
 \label{eq:certificate-g}
\end{equation}
Define the residual
\begin{equation}
 F_*(t,u)\coloneqq g_*(t,u)-\alpha-\beta(t-1/2),
 \qquad \alpha=\frac{4287673}{10^7},\quad
 \beta=\frac{7391729}{10^7}.
 \label{eq:certificate-residual}
\end{equation}
The required claim is $F_*\geq0$ on $0\leq t\leq1$, $0\leq u\leq t(1-t)$. The check proves a strictly positive margin for the deliberately rounded-down affine function. It need not locate a stationary point or prove convexity of the minimized penalty.

For a rectangle $[t_-,t_+]\times[u_-,u_+]$, define
\begin{equation}
 c(t)\coloneqq\frac9{200}+\frac{(t-1/10)^2}{4},\qquad
 c_-\coloneqq\min_{[t_-,t_+]}c(t),\quad
 c_+\coloneqq\max_{[t_-,t_+]}c(t).
\end{equation}
The minimum is obtained at the closest point to $1/10$ and the maximum at an endpoint. A rectangle is discarded as empty only if
\begin{equation}
 u_->\max_{t\in[t_-,t_+]}t(1-t)
 \quad\text{or}\quad u_-/2+c_->1/4.
 \label{eq:certificate-empty}
\end{equation}
The first condition excludes the input-state domain; the second excludes a physically possible qubit output determinant.

\subsection{Two elementary rectangle bounds}
Monotonicity of $\ent$, together with $\beta>0$, gives the following lower bound on every physical point of a nonempty rectangle:
\begin{equation}
 F_*(t,u)\geq
 \ent(u_-/2+c_-)-\ent(c_*^2u_+)-\alpha-\beta(t_+-1/2).
 \label{eq:certificate-basic}
\end{equation}
This bound needs no differentiability at the boundary $u=0$.

Away from that boundary, a mean-value bound is more efficient. Set
$y_-\coloneqq u_-/2+c_-$ and $y_+\coloneqq u_+/2+c_+$. If $u_->0$ and $y_+<1/4$, then the entire rectangle lies inside the differentiability domain of both entropy functions. The derivatives are
\begin{align}
 \partial_tF_*&=\frac{t-1/10}{2}\ent'(u/2+c(t))-\beta,\\
 \partial_uF_*&=\frac12\ent'(u/2+c(t))-c_*^2\ent'(c_*^2u).
 \label{eq:certificate-gradients}
\end{align}
The positive function $\ent'$ decreases, as shown in~\eqref{eq:det-entropy-derivative}. Its range on the first argument is therefore contained in
$[\ent'(y_+),\ent'(y_-)]$, and on the second in
$[\ent'(c_*^2u_+),\ent'(c_*^2u_-)]$.
Interval multiplication in~\eqref{eq:certificate-gradients} bounds each derivative even when $t-1/10$ changes sign. Let $G_t,G_u$ be upper bounds on their absolute values throughout the rectangle. At its midpoint $(t_0,u_0)$, the mean value theorem gives
\begin{equation}
 F_*(t,u)\geq F_*(t_0,u_0)
 -G_t\frac{t_+-t_-}{2}-G_u\frac{u_+-u_-}{2}.
 \label{eq:certificate-Taylor}
\end{equation}
Both~\eqref{eq:certificate-basic} and~\eqref{eq:certificate-Taylor} are evaluated with outward interval arithmetic, so a positive lower interval endpoint proves positivity on the entire rectangle, not just at its midpoint.

\subsection{Coverage, margins, and resulting enclosure}
The root rectangle is $[0,1]\times[0,1/4]$. The supplied file \texttt{minorant\_boxes.json} records a finite dyadic partition. Its checker reconstructs the tree from the root, bisecting the $t$ interval when its width is more than twice the $u$ width and bisecting the $u$ interval otherwise. A terminal rectangle must either satisfy~\eqref{eq:certificate-empty} or have a positive lower bound from one of the preceding two estimates. A missing terminal record causes failure rather than being treated as verified.

The reconstructed tree has 983 nodes, with 475 verified nonempty terminal rectangles and 17 discarded empty rectangles. At 60-decimal-digit interval working precision, the smallest verified residual margin is greater than $8.3\times10^{-9}$. Every defining decimal constant is supplied as an exact rational decimal, and all rectangle endpoints are exact dyadic numbers. Transcendental quantities are enclosed, rather than evaluated only at floating-point midpoints. The code, partition, and complete output are included in the accompanying numerical materials under \texttt{code/thermal/}.

The resulting upper estimate is
\begin{equation}
 \VF(\rho^\star)\leq h_2(3/10)-\alpha
 =0.4525235992306926182\ldots.
\end{equation}
For the feasible ensemble in~\eqref{eq:thermal-VF-lower}, the same interval arithmetic encloses its objective around
$0.4525235098589629611\ldots$.
Their outward-rounded interval endpoints prove~\eqref{eq:thermal-VF-interval}. The proof of the operational upper bound uses the former estimate; the latter only quantifies the tightness of the evaluation of the potential.

\section{Lower certificates for the earlier state converses}\label{app:state-upsilon-cert}
To establish improvement over an optimized converse, one needs a lower bound on its value. This appendix constructs such bounds for the two state-$\Upsilon$ restrictions defined in~\eqref{eq:state-Upsilon-definition}. The certificates are global: the selected reference operator need not be an optimizer, or even a feasible point of the older minimization.

\subsection{A relative-entropy supporting hyperplane}
Fix a full-rank target $\rho_{AB}$ and a positive definite reference operator $\sigma_0$. Define
\begin{equation}
 W\coloneqq D\ln(\sigma_0)[\rho]
 =\int_0^\infty(\sigma_0+tI)^{-1}\rho(\sigma_0+tI)^{-1}\,dt.
 \label{eq:old-certificate-W}
\end{equation}
The Fr\'echet derivative is with respect to the matrix argument, and the logarithm here is natural. The integral follows by differentiating the logarithm representation
\begin{equation}
 \ln S=\int_0^\infty\left[\frac{I}{1+t}-(S+tI)^{-1}\right]dt;
\end{equation}
the derivative of the inverse in direction $R$ is $-(S+tI)^{-1}R(S+tI)^{-1}$.
 The integral shows that $W\geq0$. The derivative is self-adjoint with respect to the trace pairing, and
\begin{equation}
 \Tr[W\sigma_0]=\Tr[\rho\,D\ln(\sigma_0)[\sigma_0]]=\Tr\rho=1.
\end{equation}
Convexity of $D(\rho\|\sigma)$ in $\sigma$ therefore gives
\begin{align}
 D(\rho\|\sigma)
 &\geq D(\rho\|\sigma_0)-\frac{\Tr[W(\sigma-\sigma_0)]}{\ln2}\\
 &=D(\rho\|\sigma_0)+\frac{1-\Tr[W\sigma]}{\ln2}.
 \label{eq:old-certificate-tangent}
\end{align}
A singular $\sigma$ gives infinite relative entropy for these full-rank targets, so it does not obstruct this bound.

Choose an operator $Q$ satisfying
\begin{equation}
 Q\geq W,\qquad Q\geq0,\qquad T_B(Q)\geq0.
 \label{eq:old-certificate-Q}
\end{equation}
For every $\sigma\in\mathfrak F_A(\rho)$, with witnesses $Z,L$, the feasible constraints imply
\begin{align}
 \Tr[W\sigma]
 &\leq\Tr[Q\sigma]\\
 &\leq\Tr[QZ]\\
 &\leq\Tr[Q(\rho_A\otimes L)]\\
 &=\Tr\!\left[L\,\Tr_A[(\rho_A\otimes I)Q]\right]
 \leq t_A,
 \label{eq:old-certificate-linear-bound}
\end{align}
where
\begin{equation}
 t_A\coloneqq\left\|\Tr_A[(\rho_A\otimes I)Q]\right\|.
 \label{eq:old-certificate-tA}
\end{equation}
The norm is the operator norm. The first step uses $\sigma\geq0$ and $Q-W\geq0$. The second uses $T_B(Z-\sigma)\geq0$ and $T_B(Q)\geq0$. The third uses $\rho_A\otimes L-Z\geq0$ and $Q\geq0$. Summing the two constraints $\rho_A\otimes L\pm Z\geq0$ shows $L\geq0$, since $\rho_A$ has full rank, and $\Tr L\leq1$ gives the final step. The contracted operator in~\eqref{eq:old-certificate-tA} is positive: under the partial trace it can equivalently be written with $\sqrt{\rho_A}$ on both sides of $Q$.

It follows that
\begin{equation}
 \Upsilon_A(\rho)\geq D(\rho\|\sigma_0)+\frac{1-t_A}{\ln2}.
 \label{eq:old-certificate-A}
\end{equation}
For $\mathfrak F_B(\rho)$, the same argument uses
\begin{equation}
 t_B\coloneqq\left\|\Tr_B[(I\otimes\rho_B)Q]\right\|,
 \qquad
 \Upsilon_B(\rho)\geq D(\rho\|\sigma_0)+\frac{1-t_B}{\ln2}.
 \label{eq:old-certificate-B}
\end{equation}
The operators $\sigma_0,Q$ can be chosen separately for the two certificates.

\subsection{Explicit certificates for the thermal state}
All matrices used here are real and have nonzero entries only on the diagonal and at positions $00,11$ and $11,00$. Denote such a matrix by the tuple
\begin{equation}
 (a,b,c,d;w)\quad\longleftrightarrow\quad
 \begin{pmatrix}a&0&0&w\\0&b&0&0\\0&0&c&0\\w&0&0&d\end{pmatrix}.
\end{equation}
Use the following exact rational-decimal data:
\begin{center}
\begin{tabular}{@{}cl@{}}\toprule
 Certificate & Matrix data\\\midrule
 $A:\ \sigma_0$ & $(0.333,0.0302,0.2181,0.1451;0.1371)$\\
 $A:\ Q$ & $(0.978,0.834,1.033,1.178;0.927)$\\\addlinespace[2pt]
 $B:\ \sigma_0$ & $(0.331,0.02596,0.23978,0.152;0.12788)$\\
 $B:\ Q$ & $(1.013,0.966,0.941,1.135;0.951)$\\\bottomrule
\end{tabular}
\end{center}
Positive definiteness of $\sigma_0$ follows by checking $b,c>0$, $a,d>0$, and $ad>w^2$. To compute~\eqref{eq:old-certificate-W}, diagonalize only its $00,11$ block $S$. If $\lambda_+>\lambda_->0$ and $P_+,P_-$ are its two spectral projectors, and $R$ is the target's corresponding block, then
\begin{equation}
 D\ln(S)[R]
 =\frac{P_+RP_+}{\lambda_+}+\frac{P_-RP_-}{\lambda_-}
 +\frac{\ln(\lambda_+/\lambda_-)}{\lambda_+-\lambda_-}
 (P_+RP_-+P_-RP_+).
 \label{eq:old-certificate-two-by-two}
\end{equation}
The other two diagonal entries are $\rho_{01,01}/b$ and $\rho_{10,10}/c$. This formula avoids numerical integration or a matrix-logarithm derivative approximation.

The supplied interval checker verifies the diagonal entries and determinant of each $2\times2$ block in $Q-W$, as well as $Q\geq0$ and $T_B(Q)\geq0$. For a tuple $(q_0,q_1,q_2,q_3;q_z)$, the last two tests reduce to
\begin{equation}
 q_0q_3\geq q_z^2,\qquad q_1q_2\geq q_z^2,
 \qquad q_i\geq0.
\end{equation}
For $\rho^\star_A=I/2$, the value of $t_A$ is
$\max\{(q_0+q_2)/2,(q_1+q_3)/2\}$.
For $\rho^\star_B=\diag(7/10,3/10)$, the value of $t_B$ is
$\max\{7q_0/10+3q_1/10,7q_2/10+3q_3/10\}$.
Outward evaluation of~\eqref{eq:old-certificate-A}--\eqref{eq:old-certificate-B} gives
\begin{align}
 \Upsilon_A(\rho^\star)&\geq0.6660076961830062256\ldots>0.66600769,\\
 \Upsilon_B(\rho^\star)&\geq0.6668816989645803520\ldots>0.66688169.
 \label{eq:old-certificate-results}
\end{align}
The long decimals indicate the evaluated supporting-hyperplane values; the final strict decimal inequalities are the enclosures used in the main text. They do not claim exact minimizers of the old programs.

\subsection{Certificates under an additional small perturbation}
For $\rho^\star_\eta$ in~\eqref{eq:thermal-perturbed-state}, retain each $\sigma_0$ above. Linearity of the derivative in its direction gives
\begin{equation}
 W_\eta=(1-\eta)W+\frac\eta4\sigma_0^{-1}.
 \label{eq:old-certificate-perturbed-W}
\end{equation}
Set $k_A\coloneqq9$ for the first certificate and $k_B\coloneqq10$ for the second. Their reference matrices satisfy
$\sigma_0\geq I/(4k_A)$ and $\sigma_0\geq I/(4k_B)$, respectively. Consequently the operators
\begin{equation}
 Q_\eta\coloneqq(1-\eta)Q+\eta k I
 \label{eq:old-certificate-perturbed-Q}
\end{equation}
satisfy~\eqref{eq:old-certificate-Q} with $W$ replaced by $W_\eta$. When contracting to calculate $t_B$, use the perturbed marginal
$\diag(7/10-\eta/5,3/10+\eta/5)$, not the unperturbed one.
At $\eta=1/100$, the same supporting-hyperplane argument yields
\begin{equation}
 \Upsilon_A(\rho^\star_{1/100})>0.53805672,
 \qquad
 \Upsilon_B(\rho^\star_{1/100})>0.52363972.
\end{equation}
The script \texttt{perturbation\_check.py} checks these inequalities using outward intervals. This establishes robustness of the stated comparison for a specified nearby state without assuming continuity of an optimizer of either semidefinite restriction.

\section{A self-contained proof of the formation marginal inequality}\label{app:formation-marginal}
The one-way converse in Section~\ref{subsec:finite-randomness} depends on the inequality
\begin{equation}
 E_F(B_1B_2;E_1E_2)_\sigma
 \geq E_F(B_1;E_1)_\sigma+G_{E_2}(\sigma_{B_2E_2}),
 \label{eq:app-formation-marginal}
\end{equation}
which is~\eqref{eq:formation-marginal-bound}. This is the directional form of the formation inequality of~\cite{YHHS05}; see also~\cite[Eq.~(17)]{BEHY11}. We give a proof from the definitions so that the finite-blocklength argument does not require this result as an additional assumption. Throughout, $J(E_2;B_2)$ means that the measurement is on $E_2$, and $G_{E_2}$ is its mixed convex roof, as in~\eqref{eq:J-env} and~\eqref{eq:mixed-roof}. The state $\sigma$ is arbitrary; it need not be a product across the two pairs or have full rank.

The idea is to measure only $E_2$. For a pure four-partite state, this measurement divides the entropy of $B_1B_2$ into two contributions: the classical information in its outcome and the average entropy that remains. The latter is large enough to cover the entanglement of formation of the unmeasured pair $B_1E_1$. Averaging this pure-state argument then produces the mixed convex roof on the second pair.

\subsection{The pure-state inequality}
Let $\phi_{B_1B_2E_1E_2}$ be pure. It suffices when optimizing $J(E_2;B_2)_\phi$ to consider rank-one POVMs on $E_2$. Indeed, a spectral decomposition of every POVM effect gives a rank-one refinement. Forgetting the refinement label is a classical channel, so data processing implies that refinement cannot decrease the outcome--$B_2$ mutual information.

Fix such a POVM, with effects $M_x=\ket{v_x}\bra{v_x}$ and $\sum_xM_x=I_{E_2}$; the vectors $\ket{v_x}$ need not be normalized. Define
\begin{equation}
 \ket{\xi_x}_{B_1B_2E_1}
 \coloneqq (I_{B_1B_2E_1}\otimes\bra{v_x}_{E_2})\ket\phi,
 \qquad q_x\coloneqq\langle\xi_x|\xi_x\rangle.
 \label{eq:formation-measurement-vectors}
\end{equation}
For $q_x>0$, let $\phi^x\coloneqq\ket{\xi_x}\bra{\xi_x}/q_x$. Terms with zero probability are omitted. Completeness of the POVM implies
\begin{equation}
 \sum_xq_x\phi^x_{B_1B_2E_1}=\phi_{B_1B_2E_1},
 \qquad
 \sum_xq_x\phi^x_{B_1E_1}=\phi_{B_1E_1}.
 \label{eq:formation-measurement-averages}
\end{equation}
Thus the measurement leaves the unconditional state of the unmeasured systems unchanged.

Convexity of entanglement of formation, followed by its monotonicity under discarding $B_2$, gives
\begin{align}
 E_F(B_1;E_1)_\phi
 &\leq\sum_xq_x E_F(B_1;E_1)_{\phi^x}\\
 &\leq\sum_xq_x E_F(B_1B_2;E_1)_{\phi^x}\\
 &=\sum_xq_x S(B_1B_2)_{\phi^x}.
 \label{eq:formation-residual-entanglement}
\end{align}
The final equality uses purity of each $\phi^x_{B_1B_2E_1}$. In the second line, $B_1B_2$ is one side of the entanglement cut, so discarding $B_2$ is a local channel.

Record the outcome in the classical register $X$ and define
\begin{equation}
 \theta_{XB_1B_2E_1}\coloneqq
 \sum_xq_x\proj x_X\otimes\phi^x_{B_1B_2E_1}.
 \label{eq:formation-postmeasurement-state}
\end{equation}
Using~\eqref{eq:formation-measurement-averages}, the classical--quantum entropy identity, and data processing under discarding $B_1$, we obtain
\begin{align}
 S(B_1B_2)_\phi-E_F(B_1;E_1)_\phi
 &\geq S(B_1B_2)_\phi-\sum_xq_xS(B_1B_2)_{\phi^x}\\
 &=I(X;B_1B_2)_\theta\\
 &\geq I(X;B_2)_\theta.
 \label{eq:formation-pure-information}
\end{align}
This holds for every rank-one POVM on $E_2$. Taking the supremum over those POVMs therefore gives
\begin{equation}
 S(B_1B_2)_\phi
 \geq E_F(B_1;E_1)_\phi+J(E_2;B_2)_{\phi_{B_2E_2}}.
 \label{eq:formation-pure-bound}
\end{equation}
No existence assumption for an optimizing measurement is needed: a supremum of quantities bounded by the same left side is bounded by that left side.

\subsection{Passing to an arbitrary mixed state}
Consider an arbitrary pure-state ensemble decomposition
\begin{equation}
 \sigma_{B_1B_2E_1E_2}
 =\sum_zp_z\phi^z_{B_1B_2E_1E_2}.
 \label{eq:formation-global-decomposition}
\end{equation}
Apply~\eqref{eq:formation-pure-bound} to each component and average:
\begin{equation}
 \sum_zp_zS(B_1B_2)_{\phi^z}
 \geq\sum_zp_zE_F(B_1;E_1)_{\phi^z}
       +\sum_zp_zJ(E_2;B_2)_{\phi^z_{B_2E_2}}.
 \label{eq:formation-average-pure-bound}
\end{equation}
The two sums on the right are controlled separately. Convexity of entanglement of formation gives
\begin{equation}
 \sum_zp_zE_F(B_1;E_1)_{\phi^z}
 \geq E_F(B_1;E_1)_\sigma.
 \label{eq:formation-first-marginal}
\end{equation}
Also, $\sigma_{B_2E_2}=\sum_zp_z\phi^z_{B_2E_2}$ is a permissible decomposition in the definition of the mixed convex roof. Its component marginals need not be pure, which is precisely why that definition is appropriate. Consequently,
\begin{equation}
 \sum_zp_zJ(E_2;B_2)_{\phi^z_{B_2E_2}}
 \geq G_{E_2}(\sigma_{B_2E_2}).
 \label{eq:formation-second-marginal}
\end{equation}
Combining these inequalities bounds the average entropy on the left of~\eqref{eq:formation-average-pure-bound} from below by the right side of~\eqref{eq:app-formation-marginal}, independently of the chosen ensemble. Taking the infimum over all pure-state decompositions~\eqref{eq:formation-global-decomposition} gives $E_F(B_1B_2;E_1E_2)_\sigma$ and proves~\eqref{eq:app-formation-marginal}.

This proof uses neither additivity nor strong superadditivity of entanglement of formation. The second marginal contributes $G_{E_2}$, not $E_F(B_2;E_2)$. With $B_1=B^{n-1}$, $E_1=E^{n-1}$, and $(B_2,E_2)$ the last pair, it gives exactly the iteration in~\eqref{eq:G-lower-EC} and hence the finite-copy upper bound~\eqref{eq:measurement-correlation-VF}.

\end{document}